\documentclass[preprint]{imsart}

\RequirePackage[OT1]{fontenc}
\RequirePackage{amsthm,amsmath}
\RequirePackage[round]{natbib}

\RequirePackage[colorlinks,citecolor=blue,urlcolor=blue]{hyperref}
\usepackage{graphicx,epsfig,color,bm}
\usepackage{amssymb}
\usepackage{tabularx}
\usepackage{longtable}
 \usepackage{supertabular} 
 \usepackage{threeparttable}
 \usepackage{hyperref}
\allowdisplaybreaks[2]

\arxiv{arXiv:0000.0000}

\allowdisplaybreaks[2]
\newcommand{\bA}{{\bf A}}
\newcommand{\bB}{{\bf B}}
\newcommand{\bD}{{\bf D}}
\newcommand{\bE}{{\bf E}}
\newcommand{\bbe}{{\bf e}}
\newcommand{\bx}{{\bf x}}

\newcommand{\bDe}{{\bm \Delta}}

\newcommand{\um}{\underline{m}}
\newcommand{\bR}{{\bf R}}

\newcommand{\bu}{{\bf u}}

\newcommand{\bo}{{\bf 0}}

\newcommand{\bH}{{\bf H}}

\newcommand{\bI}{{\bf I}}
\newcommand{\br}{{\bf r}}

\newcommand{\bY}{{\bf Y}}
\newcommand{\by}{{\bf y}}
\newcommand{\bX}{{\bf X}}

\newcommand{\bW}{{\bf W}}

\newtheorem{theorem}{Theorem}[section]
\newtheorem{lemma}{Lemma}[section]
\newtheorem{remark}{Remark}[section]

\newtheorem{corollary}{Corollary}[section]
\newtheorem{proposition}{Proposition}[section]

\newenvironment{seqnarray}{\small\begin{eqnarray*}}{\end{eqnarray*}}

\newcommand{\bqa}{\begin{eqnarray}}
\newcommand{\eqa}{\end{eqnarray}}
\newcommand{\bqn}{\begin{eqnarray*}}
\newcommand{\eqn}{\end{eqnarray*}}
\newcommand{\be}{\begin{equation}}
\newcommand{\ee}{\end{equation}}
\newcommand{\non}{\nonumber\\}

\newcommand{\rE}{{\rm E}}
\newcommand{\tr}{{\rm tr}}

\newcommand{\md}{\mbox{d}}

\startlocaldefs
\numberwithin{equation}{section}
\theoremstyle{plain}

\endlocaldefs

\begin{document}

\begin{frontmatter}
\title{Generalized  Four Moment  Theorem and an Application to  CLT  for Spiked Eigenvalues of high-dimensional Covariance Matrices.}
\runtitle{Generalized Four Moment Theorem}

\begin{aug}
\author{\fnms{Dandan} \snm{Jiang}\thanksref{t1}\ead[label=e1]{jddpjy@163.com}}
\and
\author{\fnms{Zhidong} \snm{Bai}\thanksref{t2}\ead[label=e2]{baizd@nenu.edu.cn}}

\thankstext{t1}{Corresponding author; Supported by Project 11471140 from NSFC.}
\thankstext{t2}{Supported by Project 11571067 from NSFC.}

\runauthor{Jiang, Bai \& Wang}

\affiliation{School of Mathematics and Statistics, Xi'an Jiaotong University, Xi'an, China\thanksmark{m1} and KLASMOE and School of Mathematics and Statistics, Northeast Normal University, ChangChun, China.\thanksmark{m2}}

 \address{School of Mathematics and Statistics, \\
Xi'an Jiaotong University \\
No.28, Xianning West Road,\\
Xi'an {\rm 710049}, China. \\ \printead{e1}}

 \address{KLASMOE and School of Mathematics and Statistics, \\
Northeast Normal University, \\
No. 5268 People's Street,\\
Changchun {\rm 130024}, China.\\
 \printead{e2}}
 
\end{aug}

\begin{abstract}

We consider a more generalized spiked covariance matrix $\Sigma$, which is a general  non-definite matrix with the spiked eigenvalues scattered into a few  bulks and the largest ones allowed to tend to infinity.  By relaxing the matching of the 4th moment to a tail probability decay,  a  {\it Generalized Four Moment Theorem} (G4MT) is proposed to show the
universality of the asymptotic law for the local spectral statistics of generalized spiked covariance matrices, which implies the limiting distribution of  the spiked  eigenvalues of the generalized spiked covariance matrix is independent of the actual distributions of the samples satisfying our relaxed assumptions. Moreover, by applying it to  the Central Limit Theorem (CLT) for the spiked  eigenvalues of the generalized spiked covariance matrix,  we remove the strict constraint of diagonal block independence for the population covariance matrix  given in \cite{BaiYao2012}, and 
extend their result to a general case  that the 4th moment and the spiked eigenvalues are not necessarily required  to be bounded and the population covariance matrix is in a general form, thus meeting the actual cases better.
\end{abstract}

\begin{keyword}[class=MSC]
\kwd[Primary ]{60B20}
\kwd{62H25}
\kwd[; secondary ]{60F05}
\kwd{62H10}
\end{keyword}

\begin{keyword}
\kwd{Generalized  Four Moment  Theorem }
\kwd{Spiked Eigenvalues}
\kwd{high-dimensional Covariance Matrices}
\kwd{Central Limit Theorem}
\end{keyword}

\end{frontmatter}

\section{Introduction} \label{Int}

The study on the universality conjecture for the local spectral statistics of random matrices, which is motivated by similar phenomena in physics, has been one of the key topics in random matrix theory. It not only plays an important role in the local field of statistics, but has also been widely used in many other fields, such as mathematical physics, combinatorics and computing science. In this paper, we are going to propose  a  {\it Generalized Four Moment Theorem} (G4MT) to prove  the  universality of the asymptotic law for the local  spiked eigenvalues of generalized spiked covariance matrices, and then apply it to  the Central Limit Theorem (CLT) for the spiked  eigenvalues of the generalized spiked covariance matrix in a general case.
\subsection{Background of universality}
As is well known, universality has been conjectured by many statisticians since the 1960s, including \cite{Wigner1958}, \cite{Dyson1970}, and \cite{Mehta1967}; it states that local statistics are universal, implying that the conclusions hold not only for the Gaussian Unitary Ensemble (GUE) but also the general Wigner random matrix. It provides  new ideas and techniques for the  research of random matrix theory, which implies that to prove one result suitable for Non-Gaussian case, it is sufficient  to show the same result under the Gaussian assumption if the universality is true.

The  similar  universality  phenomena of  the bulk of the spectrum has been also investigated  in many studies. 
A rigorous result has emerged  in \cite{Soshnikov1999},  which proved that  the universality of the joint distribution of the largest $k$ eigenvalues (for any fixed $k$)  hold under  the symmetric assumption of the  atom distribution.
 \cite{Johansson2001}, \cite{BP2005} focused on the Gauss divisible, which is a strong regularity assumption on the atom distribution. Further,  \cite{Erdos2010a} relaxed the above  regularity assumption to a distribution family with a explicit form. 
  \cite{Erdos2010b} improved the work by the analysis of the Dyson Brownian motion but still requires a high degree of regularity on the atom distribution. 
  Most recently,  \cite{TaoVu2015} showed the universality of the asymptotic law for the local spectral statistics of the Wigner matrix by the {\it Four Moment Theorem},  which is based on the Lindeberg strategy in \cite{Lindeberg1922} of replacing non-Gaussian random variables with Gaussian ones. This method  assumes that the moments of the entries match that of the complex  standardized Gaussian ensemble up to the 4th order and requires the $C_0$ condition to hold, which states that  the independent distributed entries have  zero mean  and identity  variance and satisfy the uniform exponential decay,  with the form 
 \[P(|x_{ij}| \geq t^C) \leq e^{-t}\]
 for all $t\geq C'$, $1\le i, j \leq n$ and $C, C'$ being some constants. Although they asserted that the fine spacing statistics of a random Hermitian matrix in the bulk of the spectrum are only sensitive to the first four moments of the entries, they also conjectured that  it may be possible to reduce the number of matching moments in their theorem. 

\subsection{Our contribution to universality}\label{sec12}
Inspired by these previous works, the G4MT is proposed by replacing the condition of
matching the 4th  moment  by a tail probability as detailed  in  Assumption~$\bB$. 
 Then the universality of the asymptotic law for  the bulk of spiked eigenvalues  of generalized  covariance matrices is
 automatically  proved by the proposed G4MT.
 By weakening the constrains, it takes several advantages  as follows: First,  
when proving the universality of the asymptotic law for the bulk of the  spiked eigenvalues, 
It only requires the condition of  matching moments up to the 3th order and the fourth moments to satisfy the tail probability,
which is a regular and necessary condition in the weak convergence of the largest eigenvalue.
For the case of  symmetric distribution, it is only needed to  consider the first and the second moments. 
Second, 
we reduce the study of universality  of a asymptotic law to the eigenvalues of a low-dimensional matrix, 
unlike \cite{TaoVu2015} which involves the partial derivative operation  of the whole  large dimensional random matrices.
As a by-product, 
 the  rigorous $C_0$ condition having  uniform exponential decay  in \cite{TaoVu2015}  is also not necessary. 
Finally, it shows that  the limiting distribution of  the spiked  eigenvalues of a generalized spiked covariance matrix is independent of the actual distributions of the samples satisfying our relaxed assumptions. 

As an application,  we also  apply the proposed G4MT to  the CLT for the spiked  eigenvalues of the generalized spiked covariance matrix. By relaxing the constrains,  we remove some of the strict conditions given in 
 \cite{BaiYao2012}, 
  and then make the result efficient in a wider usage,  where the 4th moment and the spiked eigenvalues are not necessarily required to be bounded and the population covariance matrix is in a general form without diagonal block independent assumption,  thus meeting the actual cases better.

%
%
 
\subsection{Related works of spiked model}
The spiked model in the high-dimensional setting is originated from the common phenomenon of large or even huge dimensionality $p$ compared to the sample size $n$, occurring in  many modern scientific fields, such as wireless communication, gene expression and climate studies.  It  was first proposed by \cite{Johnstone2001}  
 under the assumptions of high dimensionality and  an identity population covariance matrix  with fixed and relatively small spikes. Since the study of  spiked covariance matrices  has a close relationship with  Principal Component Analysis (PCA) or Factor Analysis (FA), which are important and powerful tools in dimension reduction, data  visualization and feature extraction,  it has inspired great interest on the part of researchers in the limiting behaviors of the eigenvalues and eigenvectors of such high-dimensional spiked sample covariance.

 Within this context, many  impressive works are devoted to investigate on the limiting properties of the spiked eigenvalues of the high-dimensional covariance matrix. The initial focus was on 
  the simplest situation that the population covariance matrix is a small perturbation of the identity covariance matrix.
  Under  this simplified  assumption, \cite{Baik2005} investigated the exact scaling rates of the asymptotic distributions of the empirical eigenvalues in both cases of below and above the related threshold.  \cite{BaikSilverstein2006}  provided  the almost sure limits of the sample eigenvalues   in the simplified spiked model for a general class of samples when both population size and sample size tend to infinity with a finite ratio. \cite{Paul2007}
showed the asymptotic structure  of the sample eigenvalues and eigenvectors with bounded spikes  in the setting of $p/n \rightarrow  c_0 \in (0,1)$ as $n \rightarrow \infty$.   \cite{BaiYao2008} derived the phase transition and the CLT of the spiked eigenvalues when the entries of the samples are independent and identically distributed (${\rm i.i.d.}$).

%
To improve the simplified assumptions,  \cite{BaiYao2012}  contributed  to deal with  a more general spiked covariance matrix,  which assumed the conditions of the diagonal  block independence and finite 4th moments. Efforts have  also been devoted to PCA or FA as a different way to improve the work on the spiked population model. For example,  \cite{BaiNg2002} focused on the determination of the number of factors and  first established the convergence rate for the factor estimates with  the constrains of 
 the independence of the components and the existence of the 8th moment.  \cite{HoyleRattray2004} used the replica method to evaluate the expected eigenvalue distribution  as the $p/n \to c$, a fixed constant. The work is considered in the case of  a number of symmetry-breaking directions. \cite{Onatski2009} derived accurate approximations to the finite sample distribution of the principal components  estimator in the large factor model with weakly influential factors. 
 The more general works are  the recent contributions from  \cite{FanWang2015} and  \cite{CaiHanPan2017} , which both investigate the asymptotic distributions of the spiked eigenvalues  and eigenvectors of a general covariance matrix.
 However, the result of  \cite{FanWang2015} only has one threshold, which  is the same as the case of block independence indeed. More importantly, their main theorems are involved with
the difference between the ratio  $\lambda_i/\alpha_i$ and 1,  with $\lambda_i$  being the corresponding sample eigenvalue, which is given as an unspecified "O'' term. Furthermore, both of the works in  \cite{FanWang2015} and  \cite{CaiHanPan2017}  require the bounded 4th moments and  the condition  $p/(n\alpha_i) \to 0$, with $\alpha_i, i=1,\cdots, K$ being the spikes, so that it  
limits the relationship between the dimensionality and the spikes. 

On the basis of these works, we further
consider a general spiked covariance matrix and study the asymptotic law for its spiked eigenvalues under relaxed assumptions. 
Since the main cause of this unspecified  "O''  term is the use of  the population spiked  eigenvalue  in the ratio  $\lambda_i/\alpha_i$, but not its phase transition, so that  we consider  to use the  phase transition of the spiked eigenvalues instead and then give the explicit CLT for the spiked eigenvalues of high-dimensional  generalized covariance matrices.

\subsection{Our contribution to spiked model}

To improve the related works on spiked model, we shall apply the proposed G4MT to
the CLT for the spiked  eigenvalues of the generalized spiked covariance matrix as  mentioned in Sec~\ref{sec12}. We consider a general spiked covariance matrix $\Sigma$, which is a general  non-definite matrix with  the  spectrum formed as 
\begin{equation}
\beta_{p,1}, \cdots,  \beta_{p,j},\cdots,\beta_{p,p}\label{array}
\end{equation}
in descending order
and 
$\beta_{p,j_k+1}, \cdots \beta_{p, j_k+m_k}$
are equal to $\alpha_k, k=1,\cdots, K$, respectively, where $j_k$  is the  rank of the eigenvalue in front of the first $\alpha_k$ in the array 
(\ref{array}) and $j_k$ also may take value at 0.
Then,  $\alpha_1, \cdots, \alpha_K$ with multiplicity $m_k, k=1,\cdots,K$, respectively, satisfying $m_1+\cdots+m_K=M$, a fixed integer, are the spiked eigenvalues of $\Sigma$ lined arbitrarily in groups among all the eigenvalues.
We apply the G4MT to the spiked  eigenvalues of such general covariance matrix $\Sigma$,  and provide a  universal asymptotic distribution of  the spiked  eigenvalues of the generalized spiked covariance matrix.
By our relaxing constraints, the proposed result demonstrates several advantages as below:
First,  we remove the strict condition that the population covariance matrix  has a diagonal block independent structure given in \cite{BaiYao2012}. As known,  their  diagonal block  structure is  equivalent to require  that the spiked and non-spiked eigenvalues are generated from the independent variables, which  is difficult to reach for the huge data today. 
Second, our method  permits  the spiked eigenvalues to  be scattered into a few  bulks, any of which are larger than their related right-threshold or smaller than their related left-threshold. So our focused work is extended to a generalized case with a few pairs of thresholds.  
Furthermore, for the generalized  population covariance matrix  that satisfies the Assumption~{\bf D},
we give  a clear and universal expression for the limit distribution of the spiked  eigenvalues of generalized spiked covariance matrix, which is not involved with an unspecified "O" term and the 4th moments.
For the cases that the Assumption~{\bf D} is not met, such as the  diagonal matrix or  the  diagonal block matrix,   we also provide the corresponding result in Remark~\ref{rmk31}, which  
 performs as well as the approach in \cite{BaiYao2012}, and even better in some cases as illustrated in simulations.
Finally, the spiked eigenvalues and the population 4th moments  are not necessarily required to be bounded in our work. 
Thus the weakening
constraints make the conclusion more applicable to actual cases.


 

 The rest of our paper is arranged as follows: In Section~\ref{Pre}, the problem is described in a generalized setting, and the phase transition for the  spiked eigenvalues of generalized covariance matrix is also presented. Section~\ref{New} gives the main results of the G4MT and applies it to the CLT for the spiked  eigenvalues of the generalized spiked covariance matrix in high-dimensional setting.  
In Section~\ref{Sim}, simulations are conducted to evaluate our work comparing with the work in Bai and Yao(2012).
Then, an applications to determining the number of the spikes and  real data analysis are also discussed in Section~\ref{Apply}. 
Finally, we draw a conclusion in the Section~\ref{Con}. Important proofs are all provided in the Supplement.


\section{Problem Description and  Preliminaries}  \label{Pre}

Consider the  random samples $T_p\bX$, where 
\begin{equation*}
\bX=(\bx_1,\cdots,\bx_{n})=\left(x_{ij}\right), 1\leq i \leq p, ~ 1\leq j \leq n,
\end{equation*}
and  $T_p$ is a $p\times p$ deterministic matrix. 
Then, $T_pT_p^*=\Sigma$ is the population  covariance matrix, which can be seen as a
 general  non-definite matrix with the spectrum  arranged in descending order in  
(\ref{array}).
The population spiked eigenvalues of $\Sigma$, $\alpha_1, \cdots, \alpha_K$ with multiplicities $m_k, k=1,\cdots,K$, are  lined arbitrarily in groups among all the eigenvalues, where  $m_1+\cdots+m_K=M$ is a fixed integer. 
 
Define the  corresponding  sample covariance matrix of the observations $T_p\bX$ as 
\begin{equation}
S=T_p\left(\frac{1}{n}\bX\bX^*\right)T_p^*,
\label{S}
\end{equation}
and then the sample covariance matrix  $S$ is the so-called  generalized spiked sample covariance matrix.

Define the singular value decomposition of $T_p$ as
 \be T_p= V\left(
\begin{array}{cc}
 D_1^{1/2} & \bo    \\
 \bo & D_2^{1/2}   
\end{array}
\right)U^{*},
\label{UDU}
\ee
where $U$ and $V$ are unitary matrices,   $D_1$ is a diagonal matrix of the $M$ spiked eigenvalues
 and $D_2$ is the diagonal matrix of the non-spiked eigenvalues with bounded components. 
 Since the investigation on the limiting distribution of the spiked eigenvalues of the  sample covariance matrix depends on the basic equation $|\lambda\bI- S|=0$, it is obvious that it only involves the right singular vector matrix $U$ but not the left one.
 
 Let $J_k$ be the set of  ranks  of  $\alpha_k$ with multiplicity $m_k$ among all the eigenvalues of $\Sigma$, {\rm i.e.}
\[J_k=\{ j_k+1,\cdots, j_k+m_k\}.\]
Denote by $\{l_{j}(\bA)\}$ the eigenvalues of  a $p \times p$ matrix $\bA$.
Then,  the sample eigenvalues of the generalized spiked sample covariance matrix $S$ are sorted 
in descending order as 
\[l_{1}(S),  \cdots,  l_{j}(S), \cdots,  l_{p}(S).\]

 To consider the limiting distribution of the spiked eigenvalues of a generalized sample covariance matrix $S$, it is necessary to determine the following assumptions:  
\begin{description}
\item[ Assumption [A\!\!]] The double array $\{x_{ij},i,j = 1,2,...\} $ consist of ${\rm i.i.d.}$
random variables with mean 0 and variance 1. Furthermore, $\rE x_{ij}^2=0$ for the complex case (when both $x$'s and $T_p$ are complex). 
\item[ Assumption~[B\!\!]]  Suppose that 
\[\lim\limits_{\tau \rightarrow \infty}\tau^4 {\rm P}\left(|x_{ij}| >\tau \right)=0\] for the {\rm i.i.d.} sample  $(x_{i1}, \cdots, x_{in}),$ $ i=1,\cdots, p$, 
 where the  4th moments may   
 unnecessarily exist.
\item[ Assumption~[C\!\!]] The $p \times p$ matrix  $\Sigma=T_pT_p^*$ forms the sequence $\{\Sigma_p\}$, which  is bounded in the spectral norm. Moreover, denote the  empirical spectral distribution (ESD) of   $\Sigma$ 
as $H_n$, which  tends to a proper probability measure $H$  as $p \rightarrow \infty$.
\item[ Assumption~[D\!\!]]  Suppose that 
 \be
 \max\limits_{t,s} |u_{ts}|^2 \left({\rE|x_{11}|^4-3}\right)I(|x_{11}|<\sqrt{n}) \rightarrow 0,\label{CondU1}
 \ee
  where $U_1=\big(u_{ts}\big)_{t=1,\cdots,p; s=1,\cdots,M }$ is the first $M$ columns  of matrix $U$ defined in (\ref{UDU}). 
\end{description}
The detailed explanation of Assumption~{\bf D} can be found in the Supplementary materials.
\begin{description}
\item[ Assumption~[E\!\!]]  
Assuming that $p/n = c_n \to c >0$  and  both $n$ and $p$ go to infinity simultaneously,
the spiked eigenvalues of the matrix  $\Sigma$, $\alpha_1,\cdots, \alpha_K$  with multiplicities $m_1,\cdots,m_K$ laying out side the support of $H$,  satisfy $\phi'(\alpha_k)>0$ for $1\le k \le K$, where 
\[\phi(x)=x\left(1+c\int\frac{t}{x-t} \md H(t)\right)\]
 is detailed in the following Proposition~\ref{P1}. 
\end{description}

\subsection{Phase transition of the spiked eigenvalues of generalized covariance matrices.}  \label{Sec2.1}


In this part, an improved version of the  phase transition for each spiked eigenvalue of a generalized sample covariance matrix is detailed under our relaxed assumptions.
For each population spiked eigenvalue 
$\alpha_k$ with multiplicity $m_k$ and the associated sample eigenvalues 
$\{l_{j}(S), j \in J_k\}$, $k=1,\cdots,K$, we have following proposition
\begin{proposition}\label{P1}
 For the spiked sample covariance matrix $ S$ given in (\ref{S}), assume that $p/n =c_n \to c >0$  and both the dimensionality $p$ and the sample size $n$ grow to infinity  simultaneously. For any population spiked eigenvalue $\alpha_k, (k=1,\cdots,K)$,
 let
  \[\rho_k =
\left\{
\begin{array}{cc}
  \phi(\alpha_k),& \mbox{ if } \phi'(\alpha_k)>0,   \\
   \phi(\underline\alpha_k), &   \mbox{ if there exists } \underline\alpha_k \mbox{ such that } \phi'(\underline\alpha_k)=0
   \\ &\mbox{ and }\phi_n'(t)<0,  \mbox{ for all } \alpha_k\le t<\underline{\alpha}_k
\\
   \phi(\overline\alpha_k), &   \mbox{ if there exists } \overline\alpha_k \mbox{ such that } \phi'(\overline\alpha_k)=0\\
   & \mbox{ and }
 \phi'(s)<0, \mbox{ for all } \overline{\alpha}_k<s\le\alpha_k
\end{array}
\right.
\]
where
\be
\phi(\alpha_k)= \alpha_k\left(1+c\int\frac{t}{\alpha_k-t} \md H(t)\right). \label{phik}
\ee
Then, 
  it holds that for all $ j \in J_k$, $\{ {l_{j}}/{\rho_k}-1\}$  almost surely converges to 0.
  \label{prop1}
 \end{proposition}
 \begin{remark}
 Since the convergence of $c_n \to c$ and $H_n \to H$ may be very slow, the difference $\sqrt{n}(l_j -\phi_k)$ may not have a limiting distribution. Furthermore, from a view of statistical inference, $H_n$ can be treated as the subject population, and $c_n$ can be regarded as the ratio of dimension to sample size for the subject sample. 
 So, we usually use 
 \be \phi_n(\alpha_k)= \alpha_k\left(1+c_n\int\frac{t}{\alpha_k-t} \md H_n(t)\right),\label{phink}
 \ee
 instead of $\phi_k$ in $\rho_k$, in particular during the process of CLT. Then, 
we only require $c_n=p/n$, and both the dimensionality $p$ and the sample size $n$ grow to infinity  simultaneously, but not necessarily  in proportion. Moreover, the approximation  that
$\{ {l_{j}}/{\rho_k}-1\}$ almost surely converges to 0 still holds  for all $ j \in J_k$.
 \label{RP1}
 \end{remark}

 Note that 
 the Proposition~{\ref{prop1}} theoretically shows that the diagonal block independent assumption of  \cite{BaiYao2012} 
 can be removed,  
  and both of the spiked eigenvalues and the population 4th moment  are not necessarily required to be bounded. 
 The proof of Proposition~{\ref{prop1}}  can be easily obtained  based on the 
G4MT, which is presented  in the next section and  shows that  two samples, $\bX$ and $\bY$, from different populations satisfying Assumptions ${\bf A} \sim {\bf E}$  will lead to the same limiting distribution of  the spiked  eigenvalues of a generalized spiked covariance matrix.  By  the G4MT,  it is reasonable to  assume the Gaussian entries from $\bX$; then, Proposition~{\ref{P1}} is proved by  the almost sure convergence and the exact separation of eigenvalues 
 in \cite{BaiSilverstein1999}.
 
 In addition, by applying  the 
G4MT to  the CLT for the spiked  eigenvalues of the generalized spiked covariance matrix,
we can obtain a universal asymptotic distribution of the spiked  eigenvalues of a generalized spiked covariance matrix, which is free of the population distribution and different from the result involved with the 4th moments in   \cite{BaiYao2012}. 
By the G4MT,  the universal CLT can be also equivalently obtained by 
 \[\bY=( {y}_1,\cdots, {y}_{n})=\left(y_{il}\right), 1\leq i \leq p, ~ 1\leq j \leq n,\]
being an independent $p$-dimensional arrays from $\mathcal{N}(\bo,\bI_p)$.

 Actually, many readers have asked the same question after reading  \cite{BaiYao2012}, that is, whether the diagonalizing assumption 
\be\Sigma=
\left(
\begin{array}{cc}
\Sigma_M  &    0  \\
0  &   V_{p-M} 
\end{array}
\right)\label{B12Sig}
\ee
is  necessary. Does the result of  \cite{BaiYao2012} hold for a more general form of an arbitrary nonnegative definite matrix? Through our work, one can find that they are clever to make such an assumption, for otherwise, the limiting distribution of the normalized spiked eigenvalues  would be independent of the 4th moment of the atom variables if the condition (\ref{CondU1}) is satisfied.


\section{Main Results }  \label{New}

Our main results are two key points: First, it is the G4MT, which shows that the samples satisfying the Assumptions ${\bf A} \sim {\bf E}$  lead to
 the same asymptotic distributions of the spiked eigenvalues of a generalized spiked covariance matrix. Second, it is the CLT 
for the spiked eigenvalues of a high-dimensional generalized covariance matrix  under our  relaxed assumptions.   
For ease of reading and understanding, the G4MT is introduced during its application to the CLT for the spiked eigenvalues of a generalized covariance matrix.
The proof of G4MT will be postponed to Section~{\ref{SG4MT}} in the Supplement  for the consistency of reading. 
Before that, we also give some explanations of the truncation procedure as below.

\subsection{Truncation}
  
  Let $\hat x_{ij}= x_{ij} {\rm I}(|x_{ij}|< \eta_n \sqrt{n})$ and  $ \tilde x_{ij}=(\hat x_{ij} - \rE \hat x_{ij})/ \sigma_{n}$  with $\sigma_{n}^2=\rE \left|\hat x_{ij}- \rE  \hat x_{ij} \right|^2$.
  We can illustrate that it is  equivalent to replacing the entries of $\bX$ with the truncated and centralized ones by Assumption~$\bB$. 
  Details of the proof are presented in Supplement~\ref{app1} and the convergence rates of arbitrary moments of  $\tilde x_{ij}$ are depicted.
  
Therefore, we only need to consider the limiting distribution of the spiked eigenvalues of $\tilde S$, which is generated from the entries truncated at $\eta_n \sqrt{n}$, centralized and renormalized. For simplicity, it is equivalent to  assume that $|x_{ij}| <\eta_n \sqrt{n}$, $\rE x_{ij}=0, \rE |x_{ij}^2|=1$, and Assumption~$\bB$ is satisfied for the real case. But it cannot meet the requirement of
$\rE x_{ij}^2 =0$ for the complex case;  instead, only $\rE x_{ij}^2 =o(n^{-1})$ can be guaranteed.

\subsection{CLT for the spiked  eigenvalues of generalized covariance matrix}  \label{Sec3.1}


 As seen from the Proposition~\ref{P1}, there is a packet of $m_k$ consecutive sample eigenvalues $\{l_j(S), j \in J_k\}$ converging to a limit $\rho_k$ laying outside the support of the limiting spectral distribution (LSD), $F^{c,H}$,
of $S$. Recall the CLT for the $m_k$-dimensional vector 
\[\Big(\sqrt{n} \big(l_j(S)- \phi(\alpha_k) \big), j \in J_k\Big)\]
given in 
\cite{BaiYao2012}. 
Since the spiked eigenvalues may be allowed to tend to infinity  in our work,  and  the difference between $l_j(S)$ and $\phi(\alpha_k)$ make convergence very slow as mentioned in Remark~\ref{RP1}, we consider the renormalized random vector   
\be
\left(\sqrt{n} \Big(\frac{l_j(S)}{\phi_n(\alpha_k)}-1\Big) , j \in J_k \right).\label{mrv}
\ee
 The CLT for (\ref{mrv}) is proposed for a general case in the following theorem, 

\begin{theorem}
 Suppose that the Assumptions $\bA \sim {\bf E}$ hold.
 For each distant generalized spiked eigenvalue, the $m_k$-dimensional real vector 
\[\gamma_k=(\gamma_{kj})= \left(\sqrt{n}\Big(\frac{l_j(S)}{\phi_{n,k}}-1\Big) , j \in J_k \right)\]
converges  weakly to the joint  distribution  of the $m_k$ eigenvalues of Gaussian random matrix 
\[-\frac1{\kappa_s}\left[\Omega_{\phi_{k}}\right]_{kk}\]
where $\phi_{k}:= \phi(\alpha_k)$,  $\phi_{n,k}:= \phi_n(\alpha_k)$ in (\ref{phink}),
\be
\kappa_s =1+\phi_{k}\alpha_k\um_2(\phi_{k})+\alpha_k\um(\phi_{k}), \label{ks}
\ee
$\um,  \um_2$ are defined in (\ref{um2}).
Note that  
$\Omega_{\phi_{k}}$ is defined in Corollary~{\ref{coro1}}  and 
$\left[\Omega_{\phi_{k}}\right]_{kk}$ is  the $k$th diagonal block of  $\Omega_{\phi_{k}}$   corresponding to the indices $\{i,j \in J_k\}$. 
\label{CLT}
\end{theorem}

 \begin{proof}
 First, for the generalized spiked sample covariance matrix ${S}$, let $S_x=\displaystyle\frac{1}{n}\bX\bX^*$ be the standard sample covariance with sample size $n$,  and for the  $p \times p$ covariance matrix $\Sigma=T_pT_p^*$,  the corresponding  sample covariance matrix is 
 $S=T_pS_xT_p^*.$
 By singular value decomposition of $T_p$, we have 
\[ T_p= V\left(
\begin{array}{cc}
 D_1^{1/2} & \bo    \\
 \bo & D_2^{1/2}     
\end{array}
\right)U^{*},
\]
where $U$ and $V$ are unitary matrices,   $D_1$ is a diagonal matrix of the $M$ spiked eigenvalues
 and $D_2$ is the diagonal matrix of the non-spiked eigenvalues.
By the eigenequation
\[0=|\lambda\bI- S|=\left|\lambda\bI-V\left(
\begin{array}{cc}
 D^{1/2}_1 & \bo    \\
 \bo & D^{1/2}_2     
\end{array}
\right)
U^* S_xU
\left(
\begin{array}{cc}
 D^{1/2}_1 & \bo    \\
 \bo & D^{1/2}_2     
\end{array}
\right)V^*\right|,\]
set $Q=U^* S_xU$, and partition it in the same way as the form 
\[\left(
\begin{array}{cc}
 Q_{11} &  Q_{12}     \\
 Q_{21}  &  Q_{22}      
\end{array}
\right) :=
\left(
\begin{array}{cc}
 U_{1}^*S_x U_1&  U_{1}^*S_x U_2   \\
 U_{2}^*S_x U_1   & U_{2}^*S_x U_2    
\end{array}
\right),\]
 then we have 
$$
\begin{array}{l}
0=\left| \lambda \bI_p -\left(
\begin{array}{cc}
  D^{1/2}_1Q_{11} D^{1/2}_1 &   D^{1/2}_1Q_{12}   D^{1/2}_2   \\
  D^{1/2}_2Q_{21} D^{1/2}_1  &   D^{1/2}_2Q_{22}     D^{1/2}_2  
\end{array}
\right)\right|\\[4mm]
~=\left |  \lambda \bI_{p-M}- D^{1/2}_2Q_{22}     D^{1/2}_2 \right | \\[3mm]
\quad\cdot\! \left| \lambda \bI_M \!-\!  D^{1/2}_1\!Q_{11} D^{1/2}_1\!-\!  D^{1/2}_1\!Q_{12} D^{1/2}_2\!
(\lambda \bI_{p-M}
\!- \!D^{1/2}_2\!Q_{22}     D^{1/2}_2 )^{-1}  \!D^{1/2}_2\!Q_{21} \!D^{1/2}_1\right|.
\end{array}
$$
If  we only consider the sample spiked eigenvalues of $S$, $l_j, j \in J_k, k=1,\cdots, K$, then we have  $\left |  l_j \bI_{p-M}- D^{1/2}_2Q_{22}     D^{1/2}_2 \right | \neq 0$, but
\begin{eqnarray}
0&=&\bigg| l_j \bI_M \!-\!\frac1nD^{1/2}_1 U_1^*\bX\Big(  \bI_n + \frac1n \bX^*U_2 D^{1/2}_2
\nonumber\\
&~& (l_j \bI_{p-M}\!-\! \frac1nD^{1/2}_2U_{2}^*\bX\bX^* U_2     D^{1/2}_2 )^{-1} D^{1/2}_2U^*_2
\bX\Big) \bX^*U_1 D^{1/2}_1
\bigg|\non
&=&\bigg|
l_j\bI_M- \frac{l_j}{n} D^{1/2}_1 U^*_1\bX (l_j\bI_{n}- \frac1n \bX^*U_2 D_2U^*_2 \bX )^{-1} \bX^*U_1 D^{1/2}_1
\bigg|\label{eigeneqnplus}
\end{eqnarray}
by the identity 
\be
Z(Z'Z-\lambda \bI)^{-1}Z'=\bI+\lambda (ZZ'-\lambda \bI)^{-1}.\label{inout}
\ee

Set
\be
\Omega_M(\lambda, \bX)\!=\! \frac{1}{\sqrt{n}}\Big({\rm tr}\!\big( (\lambda \bI_{n}- \frac1n \bX^*\Gamma \bX)^{-1} \!\big) D_1\!-\!
 D_1^{\frac1{2}}U^*_1\bX (\lambda \bI_{n}\!-\! \frac1n \bX^*\Gamma \bX)^{-1} \bX^*U_1D_1^{\frac1{2}}\Big),\label{OmegaM}
\ee
where $\Gamma=U_2 D_2U^*_2$. Then, for any  sample spiked eigenvalues $l_j$, it follows from (\ref{eigeneqnplus}) that
\begin{align}
\quad 0
 &=\bigg|
l_j\bI_M- \frac{l_j}{n}{\rm tr}\Big((l_j \bI_{n}- \frac1n \bX^*\Gamma \bX )^{-1} \Big) D_1 +\frac{l_j}{\sqrt{n}}\Omega_M(l_j,\bX)
\bigg|, \non
&=\bigg|
\phi_{n,k}\bI_M- \frac{\phi_{n,k}}{n}{\rm tr}\Big((\phi_{n,k} \bI_{n}- \frac1n \bX^*\Gamma \bX )^{-1} \Big) D_1 \non
&\quad+B_1(l_j)+B_2(l_j)+\frac{\phi_{n,k}}{\sqrt{n}}\Omega_M(\phi_{n,k},\bX)
\bigg|
\label{eigeneq1}
\end{align}
where
 the involved $B_i(l_{j}), i=1,2$ are specified as below:
\begin{align}
B_1(l_j)&=(l_j-\phi_{n,k})\bI_M=\frac{1}{\sqrt{n}}\phi_{n,k}\gamma_{kj}\bI_M\label{Bcha1}\\[0.5mm]
B_{2}(l_{j})&=\frac{\phi_{n,k}}{n}D_1^{1 \over 2}U_1^*\bX\big(\phi_{n,k}\bI_{n}- \frac1n \bX^*\Gamma \bX\big)\!^{-1}\!\bX^*U_1D_1^{1 \over 2}\non
&\quad-\!
\frac{l_j}{n}D_1^{1 \over 2}U_1^*\bX\big(l_j\bI_{n}- \frac1n \bX^*\Gamma \bX\big)\!^{-1}\!\bX^*U_1D_1^{1 \over 2}\non
&=\frac{\phi_{n,k}}{n}D_1^{1 \over 2}U_1^*\bX\Big((\phi_{n,k} \bI_{n}-\frac1n \bX^*\Gamma \bX)^{-1} 
 -(l_{j} \bI_{n}-\frac1n \bX^*\Gamma \bX)^{-1}\Big)
\bX^*U_1D_1^{1 \over 2}\non
&\quad-\frac{l_{j} -\phi_{n,k}}{n}D_1^{1 \over 2}U_1^*\bX(l_{j} \bI_{n}-\frac1n \bX^*\Gamma \bX)^{-1}
\bX^*U_1D_1^{1 \over 2}\non
&=\!\frac 1{\sqrt{n}}\!\gamma_{kj}\phi^2_{n,k}\frac1{n}   \tr\big((\phi_{n,k} \bI_{n}-\frac1n \bX^*\Gamma \bX)^{-2} \big)D_1\non
&\quad
-\frac 1{\sqrt{n}}\!\gamma_{kj}\phi_{n,k}\frac1{n}   \tr\big((\phi_{n,k} \bI_{n}-\frac1n \bX^*\Gamma \bX)^{-1} \big)D_1
+o(\frac 1{\sqrt{n}})\non
&=\!\frac 1{\sqrt{n}}\!\gamma_{kj}\Big(\phi^2_{n,k}\um_2(\phi_{n,k})+\phi_{n,k} \um(\phi_{n,k})
\Big)D_1+
o(\frac 1{\sqrt{n}}) 
  \label{Bcha2}
\end{align}
and
  \be
   \underline{m}(\lambda)=\displaystyle\int \frac{1}{x-\lambda} \md \underline{F}(x),\quad \underline{m}_2(\lambda)=\displaystyle\int \frac{1}{(\lambda-x)^2} \md \underline{F}(x)\label{um2}
  \ee
with $\underline{F}(x)$  being LSD of  the matrix $\displaystyle\frac1n \bX^*\Gamma \bX$.

Furthermore, if  consider the $k$th diagonal block of   
 the item \[\phi_{n,k}\bI_M - \displaystyle\frac{\phi_{n,k}}{n}{\rm tr}\Big((\phi_{n,k} \bI_{n}- \frac1n \bX^*\Gamma \bX )^{-1} \Big) D_1\]
 in (\ref{eigeneq1}),
   $\um(\lambda)$ is  the Stieltjes transform of the matrix $\displaystyle\frac1n \bX^*\Gamma \bX$, which is 
the solution to 
 \[\lambda=-\frac{1}{\um}+c\int \frac{t}{1+t\um} \md H(t).\]
 Define the analogue $\um_n$ with $H$ substituted by the ESD $H_n$ \textcolor{blue}{and $c$ by $c_n$}, satisfying the equation
 \[\phi_{n,k}=-\frac{1}{\um_n}+c_n\int \frac{t}{1+t\um_n} \md H_n(t).\]
By the proof of Theorem~1.1 in 
\cite{BaiSilverstein2004},  it is found that 
\be
\frac1{n}\phi_{n,k}{\rm tr}(\phi_{n,k} \bI_{n}- \frac1{n} \bX^*\Gamma \bX )^{-1}  +\phi_{n,k}\um_n(\phi_{n,k})=o(\frac{1}{ \sqrt n}).\label{fact}
\ee
Then,
  by the similar derivation of (5.1) in Bai and Yao (2008), we obtain that the phase transition of $l_j$,  $\phi_{n,k}$, asymptoticly satisfies the equation 
\be \phi_{n,k}+\phi_{n,k}\um_n(\phi_{n,k})\alpha_k=0.\label{eqn0}
\ee
Therefore,
to complete the proof of Theorem~\ref{CLT}, it is needed to derive the limiting distributions of $\Omega_M(\phi_{n,k},\bX)$. So 
the theoretical tool named G4MT is proposed in the following part, which is used to prove the  limiting distributions of $\Omega_M(\phi_{n,k},\bX)$. For  the consistence of reading, we only introduce the theorem here, but postpone the proof to the Supplement~\ref{SG4MT}.

\subsection{Generalized Fourth Moment Theorem}  \label{Sec3.2}

The G4MT is established in  the following theorem,  which shows that the limiting distributions of  the spiked  eigenvalues of a generalized spiked covariance matrix is independent of the actual population distributions  provided the samples to satisfy the Assumptions $\bA \sim {\bf E}$. 
\begin{theorem}[{\bf G4MT}]\label{thm2}
Assuming that $\bX$ and $\bY$ are two sets of double arrays satisfying Assumptions $\bA \sim {\bf E}$,  then it holds that
 $ \Omega_M(\phi_{n,k},\bX) $ and $\Omega_M(\phi_{n,k}, \bY)$ have the same limiting distribution, provided one of them has.
\end{theorem}


By Theorem \ref{thm2}, we may assume that $\bX$ consists of entries of {\rm i.i.d.} standard random variables in deriving the limiting distribution of $\Omega_M(\phi_{n,k}, \bX)$. Namely, we have the following Corollary.
\begin{corollary}\label{coro1}
If $\bX$ satisfies the Assumptions $\bA \sim {\bf E}$, let
\be\theta_k=\alpha_k^2\um_2(\phi_k),\label{thetak}
\ee
 then $\Omega_M(\phi_{n,k}, \bX)$ tends to a limiting distribution of an $M\times M$ Hermitian matrix {$\Omega_{\phi_{k}}$,  where $\frac{1}{\sqrt{\theta_k}}\Omega_{\phi_{k}}$ is  Gaussian Orthogonal Ensemble (GOE) for the real case, with 
the entries above the diagonal being ${\rm i.i.d.} \mathcal{N}(0,1)$ and the entries on the diagonal being ${\rm i.i.d.} \mathcal{N}(0,2)$. 
For the complex case, the $\frac{1}{\sqrt{\theta_k}}\Omega_{\phi_{k}}$ is  GUE, whose  entries  are all ${\rm i.i.d.} \mathcal{N}(0,1)$.}
\end{corollary}
The proof  of Corollary~{\ref{coro1}} is detailed in Supplement~{\ref{SOmega}}.

\begin{remark}\label{rmk31}
Actually, for the cases where the Assumption~{\bf D} is not met, for example, the population covariance matrix is a diagonal matrix or has a diagonal block independent structure, the conclusion of Corollary~\ref{coro1}  is not valid. In this case, we need the condition that the 4th moment is finite, so that the conclusion still holds, but the variance of the element $\omega_{ij}$ of $\Omega_{\phi_{k}}$ becomes
\be
{\rm Var}(\omega_{ij})=\left\{
\begin{array}{cc}
2\theta_k+\beta_x\nu_k,  &    i=j \\
\theta_k,  &      i\neq j,
\end{array}
\right.
\label{varw}
\ee
where $\theta_k$ is define in (\ref{thetak}),
 $\beta_x=(\sum\limits_{t=1}^pu_{ti}^4\rE|x_{11}|^4-3)$ with  ${\bf u}_i=(u_{1i}, \cdots, u_{pi})'$ being the $i$th column of the matrix $U_1$ (If the covariance matrix is a diagonal matrix, then $\beta_x=(\sum\limits_{t=1}^pu_{ti}^4\rE|x_{11}|^4-3)=\rE|x_{11}|^4-3$.)
and
 $\nu_k=\alpha_k^2/\big(\phi_k(1+c{\tilde m}(\phi_k))\big)^2$ is derived in the Supplement~\ref{SOmega}. 
 In fact, $\tilde m(\phi_{k})$ is the limit of $\tilde m_p(\phi_{k})= \frac{1}{p-M} \sum\limits_{q=M+1}^p \frac{d_q}{l_q-\phi_{k}}$ with 
  $d_q$ being the $q$th diagonal element of the matrix $D_2$, and $l_q$'s are the eigenvalues of  the matrix $\displaystyle\frac1n D_2^{1 \over 2 }U_2^* \bX\bX^* U_2D_2^{1 \over 2 }$.  If $D_2=\bI_{p-M}$, then $\tilde m(\phi_{k})$ is actually the Stieltjes transform of the LSD of  the matrix  $\displaystyle\frac1n U_2^*  \bX \bX^* U_2$. 
\end{remark}

\subsection{Completing the proof of Theorem~3.1}  \label{Sec3.3}

Now, we continue to the previous proof of Theorem~\ref{CLT}.
For every sample spiked eigenvalue, $l_j, j \in J_k, k=1,\cdots, K$, it follows from equation (\ref{eigeneq1}) that
\begin{align}
\quad 0
&=\bigg|
\phi_{n,k}\bI_M- \frac{\phi_{n,k}}{n}{\rm tr}\Big((\phi_{n,k} \bI_{n}- \frac1n \bX^*\Gamma \bX )^{-1} \Big) D_1 \non
&\quad+\frac{\phi_{n,k}}{\sqrt{n}}\Omega_M(\phi_{n,k},\bX)+B_1(l_j)+B_2(l_j)
\bigg|\non
\!&=\!\bigg|
\phi_{n,k}\bI_M\!+\!\phi_{n,k}\um_n(\phi_{n,k})D_1\!+\!\frac1{\sqrt{n}}\phi_{n,k}\Omega_M(\phi_{n,k},\bX)\non
&\quad +  \frac{1}{\sqrt{n}}\phi_{n,k}\gamma_{kj}\Big(\bI_M+\!\big(\phi_{n,k}\um_2(\phi_{n,k})+\um(\phi_{n,k})
\big)D_1\Big)\!+\!o(\frac{1}{\sqrt{n}}) \bigg| 
\label{eigeneq2}
\end{align}
by the equations  (\ref{Bcha1}), (\ref{Bcha2}) and (\ref{fact}).  

By the  G4MT, we can derive  
the limiting distribution of $\Omega_M(\phi_{n,k}, \bX)$ under the assumption of  Gaussian entries.
  Details of the proof for the  limiting  distribution of $\Omega_M(\phi_{n,k}, \bX)$ is provided in Supplement~{\ref{SOmega}}. 
Therefore, 
applying  Skorokhod strong representation theorem (see \cite{Skorokhod1956}, \cite{HuBai2014}), 
we  may assume that the  convergence 
of $\Omega_M(\phi_{n,k}, \bX)$  and (\ref{eigeneq2}) are in this sense almost surely by choosing an appropriate probability space. 

To be specific, by (\ref{eigeneq2}),  it yields
\begin{seqnarray}
&&0= \left|
{\tiny \left(
\begin{array}{cccc}
\phi_{n,k}\!\big(\!1\!+\!\alpha_1\um_n(\phi_{n,k})\!\big)\!\bI_{m_1}&0& \cdots  &0  \\
 0&\ddots  && \\
 && \!\phi_{n\!,k}\!\big(\!1\!+\!\alpha_k\um_n(\phi_{n,k})\!\big)\!\bI_{m_k}&   \vdots \\
   \vdots &  && \\
   &  &\ddots& 0\\
      &  && \\
  0&\cdots &0  &\!\phi_{n,k}\!\big(\! 1\!+\!\alpha_K\um_n(\phi_{n,k})\!\big)\!\bI_{m_K}
\end{array}
\right)}\right.\\
&&+\frac{\gamma_{kj}\phi_{n,k}}{\sqrt{n}}
{\tiny\left(
\begin{array}{cccc}
\big(\!1\!+\!\alpha_1\!(\!\phi_{k}\um_2\!+\!\um\!)\!\big)
\bI_{m_1} &0& \cdots  &0  \\
 0&\ddots  && \\
 &&
\big(\!1\!+\!\alpha_k\!(\!\phi_{k}\um_2\!+\!\um\!)\!\big)\bI_{m_k} &   \vdots \\
  \vdots &  && \\
    &  &\ddots& 0\\
      &  && \\
  0&\cdots &0  & 
\big(\!1\!+\!\alpha_K\!(\!\phi_{k}\um_2\!+\!\um\!)\!\big)\bI_{m_K} 
\end{array}
\right)}\\
&& \left.+\!\frac1{\sqrt{n}}\phi_{n,k}\Omega_M(\phi_{n,k},\bX)
\!+\!o(\frac{1}{\sqrt{n}}) \right|. 
\end{seqnarray}
where $\um, \um_2$ are the simplified notations of $\um(\phi_{n,k})$ and $\um_2(\phi_{n,k})$, respectively.

For the population eigenvalues $\alpha_u$ in the $u$th diagonal block of $D_1$, 
if $u\neq k$,  $\phi_{n,k}\big(1+\alpha_u\um_n(\phi_{n,k})\big)$ keeps away from 0, 
 which means  $\phi_k(1+\alpha_u\um(\phi_k))\ne 0$ if $\alpha_u$ is fixed; or $\lim(\phi_{n,k}\um(\phi_{n,k}))=-1$ if $\alpha_u\to\infty$, then we also have $\phi_{n,k}\big(1+\alpha_u\um_n(\phi_{n,k})\big) \to \phi_k-\alpha_u \neq 0$, when $u \neq k$.  
Moreover, $\phi_{n,k}\big(1+\alpha_k\um_n(\phi_{n,k})\big) = 0$ by definition. Then, multiplying $n^{\frac{1}{4}}$ to the  $k$th  block row and $k$th block column of the above equation, by Lemma~4.1 in \cite{BaiMiaoRao1991},
 it follows that 
\[
0=\Bigg|\phi_{n,k}\left[\Omega_M(\phi_{n,k},\bX)\right]_{kk}+
{\gamma_{kj}\phi_{n,k}\big(1+\phi_{n,k}\alpha_k\um_2(\phi_{n,k})+\alpha_k\um(\phi_{n,k})\big)\bI_{m_k}}+\!o(1) \Bigg|, 
\]
where $\left[~\cdot ~\right]_{kk}$ is  the $k$th diagonal block of  a matrix  corresponding to the indices $\{i,j \in J_k\}$

Obviously, $\gamma_{kj}$ asymptotically satisfies the following equation  
\[
\Bigg|\phi_{n,k}\left[\Omega_M(\phi_{n,k},\bX)\right]_{kk}+
{\gamma_{kj}\phi_{n,k}\big(1+\phi_{n,k}\alpha_k\um_2(\phi_{n,k})+\alpha_k\um(\phi_{n,k})\big)\bI_{m_k}} \Bigg|=0, 
\]

This shows that $\gamma_{jk}$ tends to the limit of the eigenvalue of 
$-\displaystyle\frac1{\kappa_s}\left[\Omega_{\phi_{k}}\right]_{kk}$,  where 
\[\kappa_s=1+\phi_{n,k}\alpha_k\um_2(\phi_{n,k})+\alpha_k\um(\phi_{n,k})\]
and 
$\Omega_{\phi_{k}}$ is the limit of $\Omega_M(\phi_{n,k},\bX)$ defined in Corollary~{\ref{coro1}}.  Because limiting behavior keeps orders of the variables, we claim that the $m_k$ ordered variables tend to the $m_k$ ordered eigenvalues of the matrix $-\displaystyle\frac1{\kappa_s}\left[\Omega_{\phi_{k}}\right]_{kk}$.

By the strong representation theorem, we conclude that the $m_k$-dimensional real vector 
$(\gamma_{kj}, j\in J_k)$ converges  weakly to the joint distribution  of  the $m_k$ eigenvalues of the Gaussian random matrix 
\[-\frac1{\kappa_s}\left[\Omega_{\phi_{k}}\right]_{kk}\]
for each distant generalized spiked eigenvalue.
Then, the CLT for each distant  spiked eigenvalue of a
generalized covariance matrix  is obtained.
\end{proof}

In fact, some  exceptional  cases are   not included in the Theorem~{\ref{CLT}}, for example,   $T_p^*T_p$ is a diagonal matrix or a diagonal block matrix that doesn't satisfy the Assumption~{\bf D}. 
For  such special cases,  the asymptotic distribution of  the bulk of spiked eigenvalues is involved with the 4th moment of the random variable corresponding to  $\bX$. So the constraints of the  bounded 4th moments and finite  spiked eigenvalues are necessary conditions for the  assumption of the diagonal or diagonal block population. The following remark provide the asymptotic distribution of the sample spiked eigenvalues of a diagonal block covariance matrix. 
By this remark,  it also shows that the condition of diagonal  block independence 
is necessary  for  the result of 
Bai and Yao (2012).
\begin{remark}
Suppose that  $\bX$ satisfies the Assumptions $\bA, {\bf B}, {\bf C}$ and  ${\bf E}$, excluding the assumption ${\bf D}$, but the 4th moment of $\bX$  and all the spiked eigenvalues are bounded. Then
 all the conclusions of Theorem~\ref{CLT} still holds, but the limiting distribution of  $\Omega_M(\phi_{n,k}, \bX)$ turns to 
 an $M\times M$ Hermitian matrix $\Omega_{\phi_{k}}=(\omega_{st})$
defined in Remark~\ref{rmk31}.
 \label{rmk2}
\end{remark}
 This remark is used for the case of  non-Gaussian assumptions when the population covariance matrix has a diagonal or diagonal block  structure.

%
%
%
%


\section{ Simulation Study}\label{Sim}

Simulations are conducted in this section to  evaluate the performance of  our proposed method.  Two cases of the population covariance matrix  structure are considered: On one hand, the {\rm {\bf Case I}} that  $\Sigma$ is a diagonal matrix shows that our method can provide the similar result to  the one in Bai and Yao (2012) (even better for some cases under the non-Gaussian assumption), when the Assumption~{\bf D} is not satisfied; On the other hand, the {\rm {\bf Case~II}} is  provided to illustrate the priority of the proposed method to the one in Bai and Yao (2012) for the general  form of $\Sigma$. They are detailed as below:
\begin{description}
\item[ Case I:] The matrix $\Sigma={\rm diag}(4,3,3,0.2,0.2,0.1,1, \cdots,1)$ is a finite-rank perturbation of a identity matrix $\bI_p$ with  the spikes $(4,3,0.2,0.1)$ of the multiplicity $(1,2,2,1)$,  thus $K=4$ and $M=6$ as proposed in Bai and Yao (2008). 
\item[Case II:]  The matrix  $\Sigma=U_0 \Lambda U_0^*$ is a general positive definite matrix,  where   $\Lambda$ is a diagonal matrix  with
the spikes    $(4,3,0.2,0.1)$ of the multiplicity $(1,2,2,1)$ as defined in {\rm {\bf Case I}} and  $U_0$ is the matrix composed of eigenvectors
 of the following matrix
\begin{align}
\left(
\begin{array}{ccccc}
  1 & \rho & \rho^2 & \cdots & \rho^{p-1} \\
  \rho & 1 & \rho & \cdots & \rho^{p-2} \\
  \ldots & \ldots &  &  &  \\
  \rho^{p-1} & \rho^{p-2} & \cdots & \rho & 1 \\
\end{array}
\right).
\end{align}
where $\rho=0.5$.
\end{description}

 For each case,  the following  population assumptions are studied:

\begin{description}
\item[ Gaussian Assumption:]  $x_{ij}$ are ${\rm i.i.d.}$ sample from  standard Gaussian population; 
\item[ Binomial  Assumption:]  $x_{ij}$ and $y_{ij}$ are  ${\rm i.i.d.}$ samples from the binary variables valued at $\{-1, 1\}$ with equal probability $1/2$, and $\beta_x=-2$. 
\end{description} 

The simulated results are depicted as follows  with 1000 replications  at  the
 values of  $p=500, n=1000$.

 \subsection{{\rm {\bf Case~I}} under Gaussian Assumption}
 As described in {\rm \bf Case~I}, we have the spikes $\alpha_1=4$, $\alpha_2=3$, $\alpha_3=0.2$  and $\alpha_4=0.1$. 
 Assume that  the Gaussian Assumption hold and  let $l_{1},\cdots, l_{p} $ be the sorted sample eigenvalues of the matrix $S$ defined in (\ref{S}). Then by the Theorem~\ref{CLT}, we obtain the limiting results as below. 
\begin{itemize}
\item First, take the single  population  spikes $\alpha_1=4$ and $\alpha_4=0.1$ into account, and consider the largest sample eigenvalue $l_{1}$ , we have :
\[\gamma_1= \sqrt{n}\Big(\frac{l_{1}(S)}{\phi_{n,1}}-1\Big) \rightarrow N(0, \sigma_1^2)\]
where 
\[\phi_{n,1}=4.667;\quad  \sigma_1^2=1.390 .\]

Similarly, for
the least eigenvalues   $l_{p}$, we have 
\[\gamma_4= \sqrt{n}\Big(\frac{l_{p}(S)}{\phi_{n,4}}-1\Big) \rightarrow N(0, \sigma_4^2)\]
where 
\[\phi_{n,4}= 0.044;\quad  \sigma_4^2=3.950.\] 

\item Second, for  the spikes $\alpha_2=3$ with multiplicity 2,  consider the sample eigenvalue $l_{2}$ and
 $l_{3}$, we obtain that the two-dimensional random vector 
\[\gamma_2=\left(\gamma_{21}, \gamma_{22}\right)'=\left( \sqrt{n}\Big(\frac{l_{2}(S)}{\phi_{n,2}}-1\Big),
\sqrt{n}\Big(\frac{l_{3}(S)}{\phi_{n,2}}-1\Big)\right)' \]
converges to the eigenvalues of random matrix $-\frac1{\kappa_s}\left[\Omega_{\phi_{2}}\right]_{22}$,
where  $\phi_{n,2}=3.750$, $\kappa_s=1.419$ for the spike $\alpha_2=3$.  Furthermore, the
matrix 
$\left[\Omega_{\phi_{2}}\right]_{22}$
 is a $2\times 2$ symmetric matrix with the independent Gaussian entries, of which the $(i,j)$ element has mean zero and the variance given by 
 \[var(w_{ij})=
\left\{
\begin{array}{cc}
 2.266, &~ \text{if}~ i=j  \\
 1.133,  &  ~ \text{if}~ i \neq j 
\end{array}
\right.\]

Similarly, for  the spikes $\alpha_3=0.2$ with multiplicity 2, we consider the sample eigenvalue $l_{p-2}$ and
 $l_{p-1}$, we obtain that the two-dimensional random vector 
\[\gamma_3=\left(\gamma_{31}, \gamma_{32}\right)'=\left( \sqrt{n}\Big(\frac{l_{p-2}(S)}{\phi_{n,3}}-1\Big),
\sqrt{n}\Big(\frac{l_{p-1}(S)}{\phi_{n,3}}-1\Big)\right)' \]
converges to the eigenvalues of random matrix $-\frac1{\kappa_s}\left[\Omega_{\phi_{3}}\right]_{33}$,
where  $\phi_{n,3}=0.075$, $\kappa_s=1.659$ for the spike $\alpha_3=0.2$.  Furthermore, the
matrix 
$\left[\Omega_{\phi_{3}}\right]_{33}$
 is a $2\times 2$ symmetric matrix with the independent Gaussian entries, of which the $(i,j)$ element has mean zero and the variance given by 
 \[var(w_{ij})=
\left\{
\begin{array}{cc}
9.004, &~ \text{if}~ i =j  \\
 4.502,  &  ~ \text{if}~ i \neq j  
\end{array}
\right.
\]
 \end{itemize}
The simulated empirical distributions of the spiked eigenvalues from Gaussian assumption under {\rm \bf Case~I}  are drawn in Figure~\ref{fig:1} in contrast to their corresponding limiting distributions.

 \begin{figure}[htbp]
\begin{center}
\includegraphics[width = .37\textwidth]{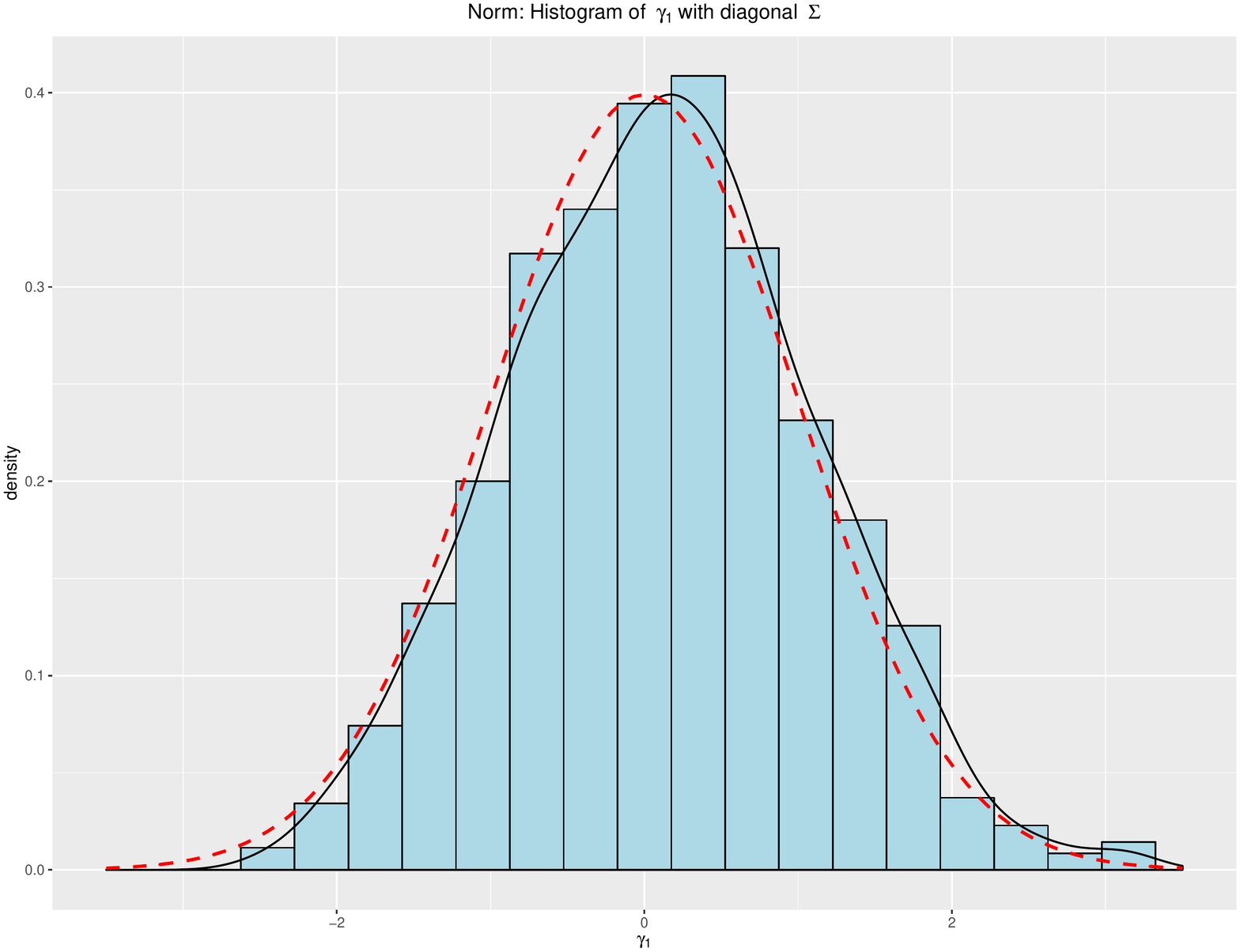}\quad\quad \includegraphics[width = .37\textwidth] {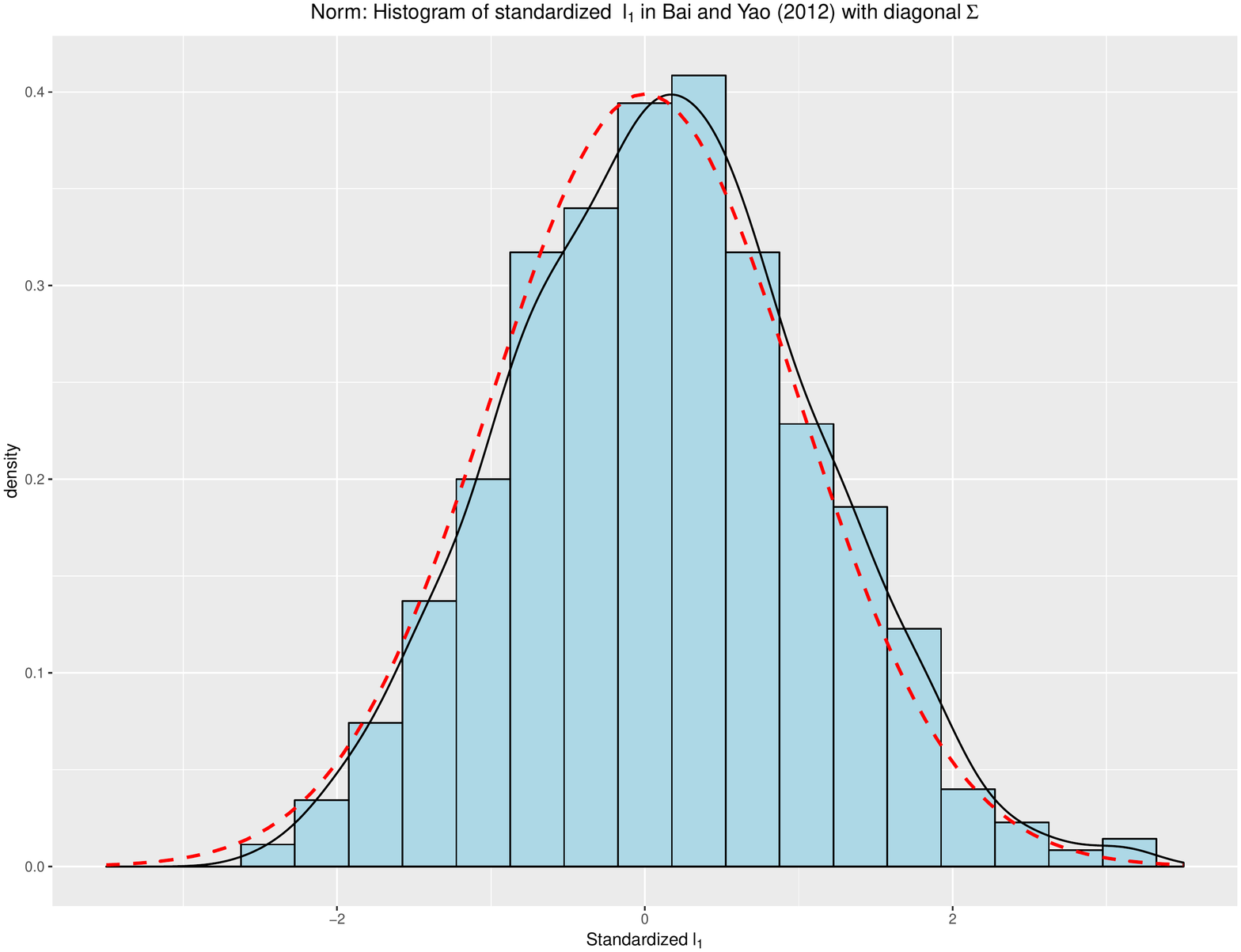}\\
\includegraphics[width = .37\textwidth]{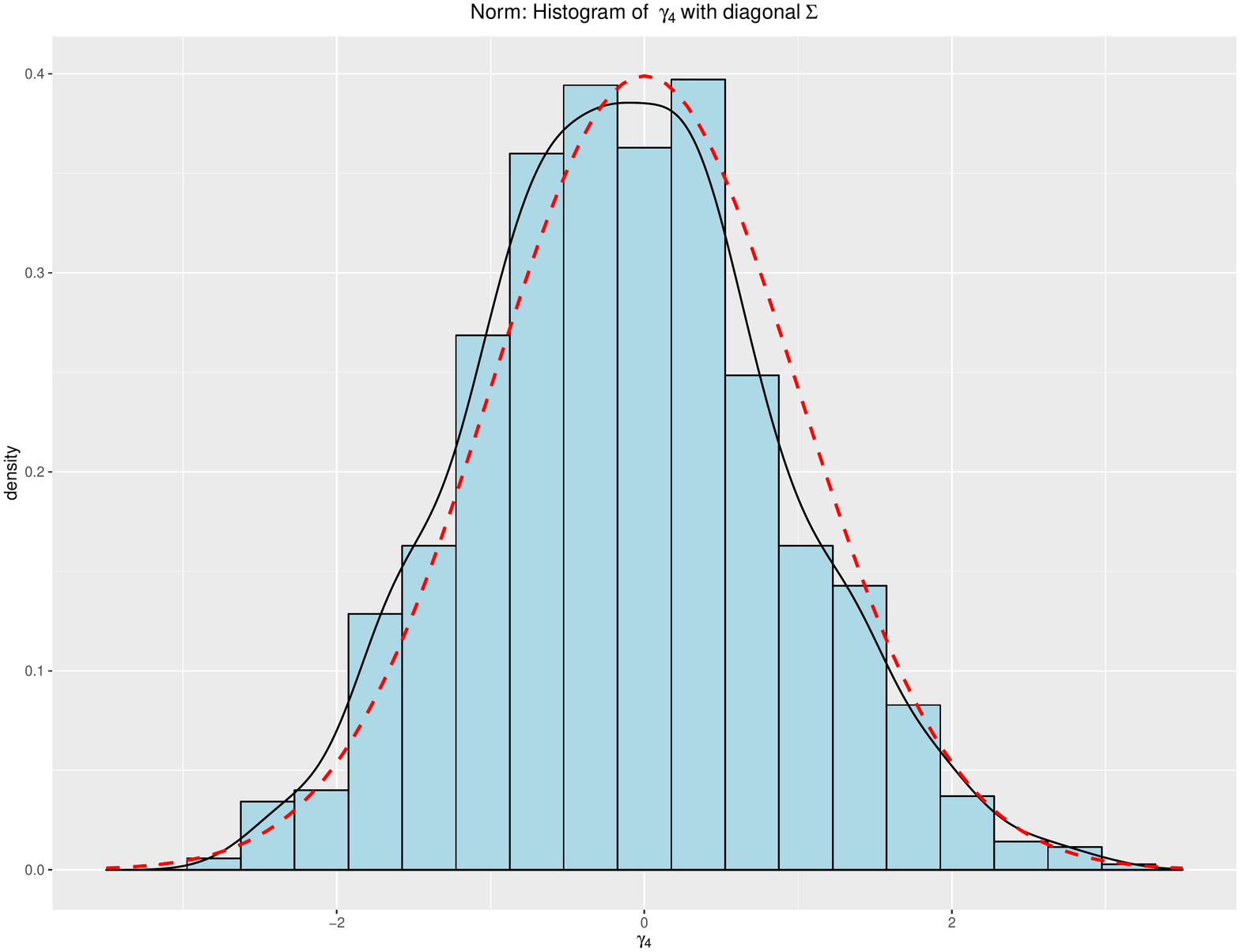}\quad\quad \includegraphics[width = .37\textwidth] {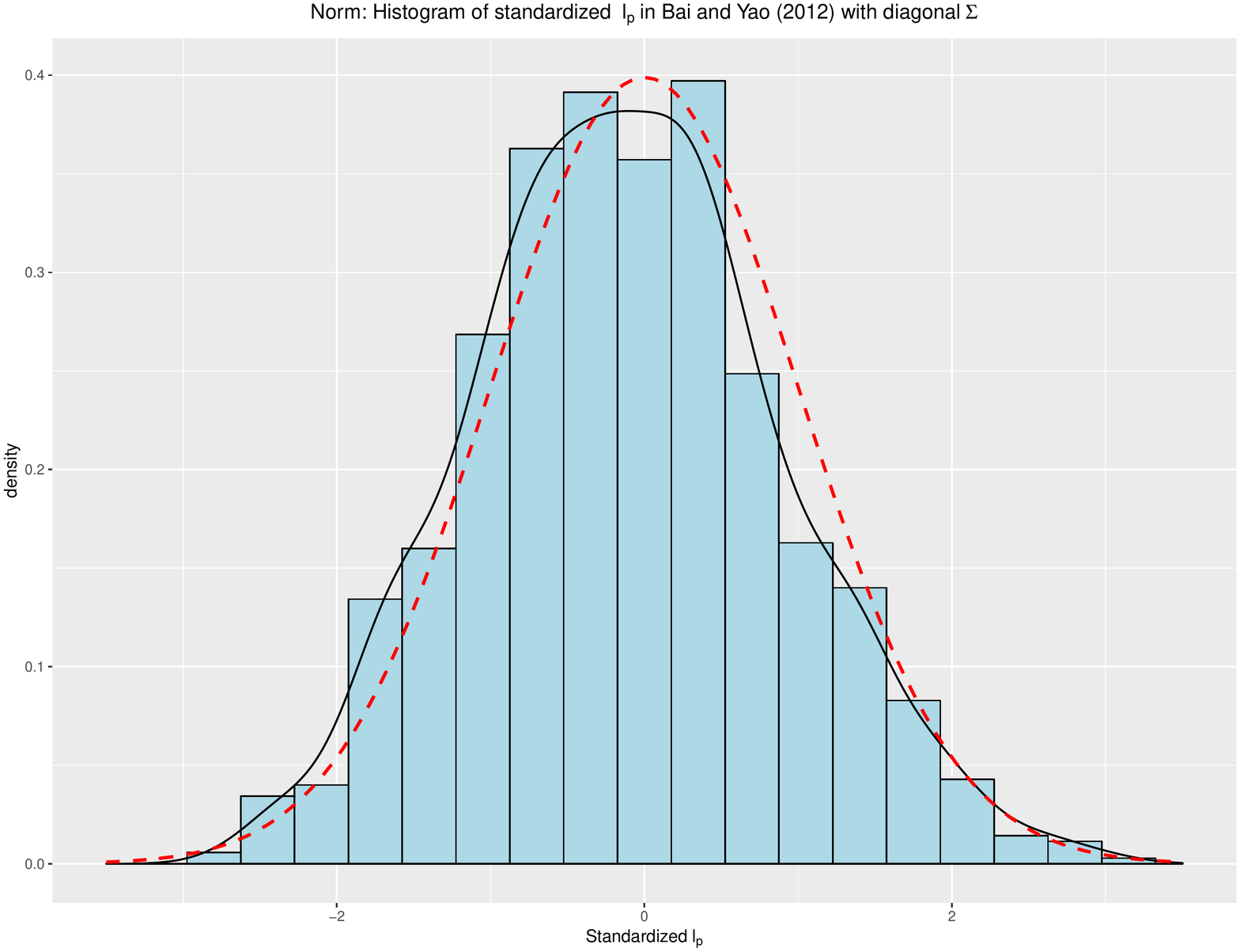}\\
\includegraphics[width = .3\textwidth]{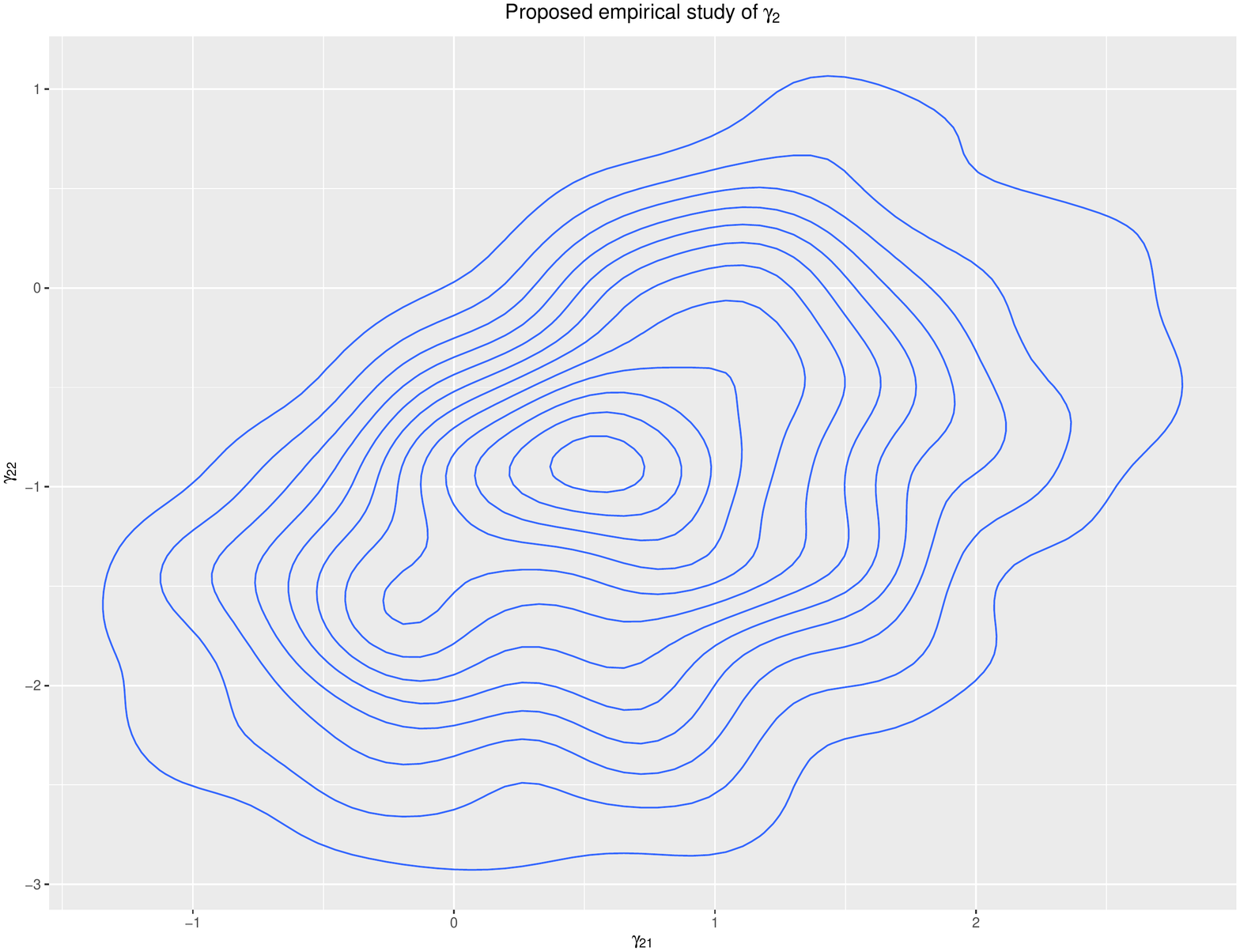}\quad \includegraphics[width = .3\textwidth] {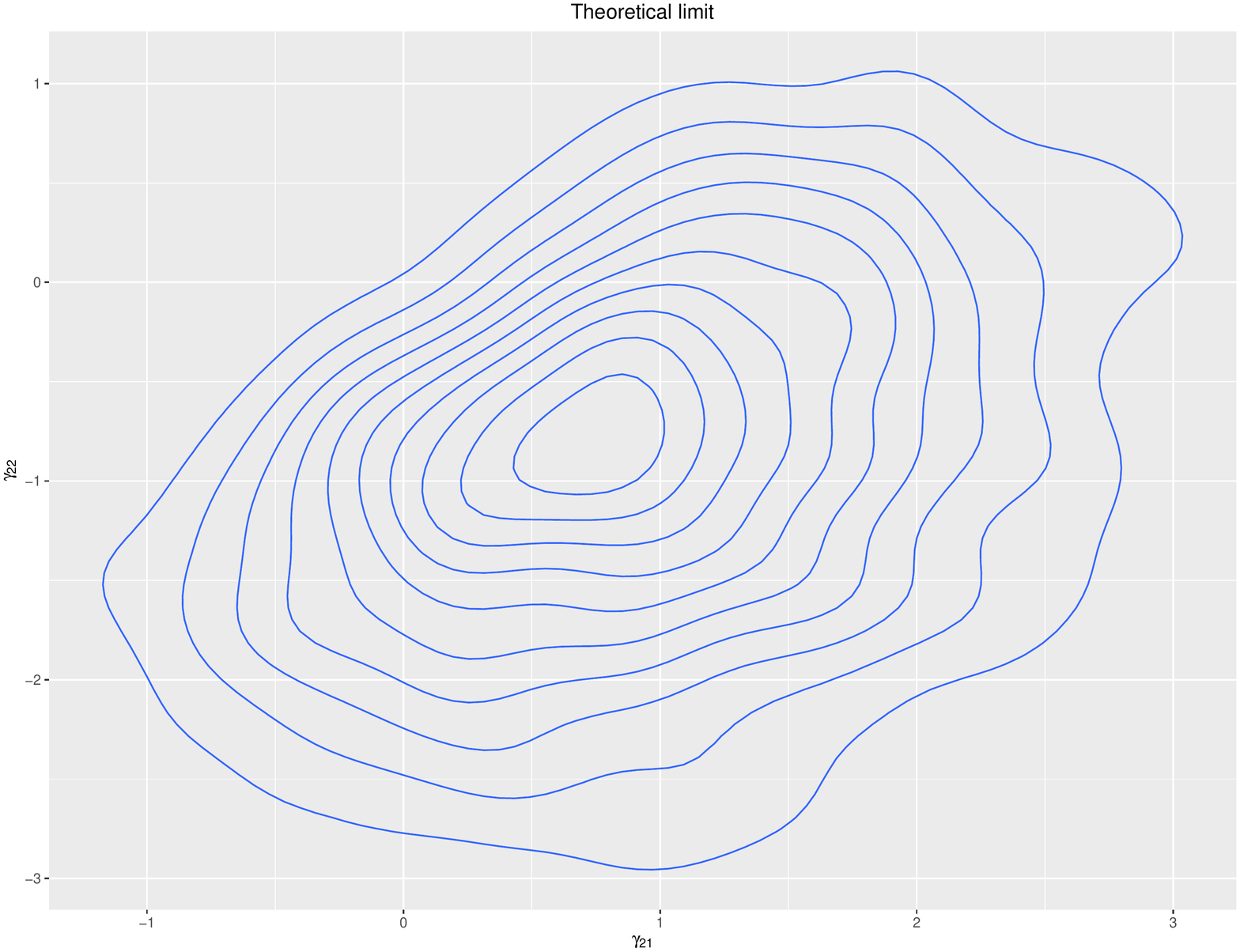}\quad \includegraphics[width = .3\textwidth] {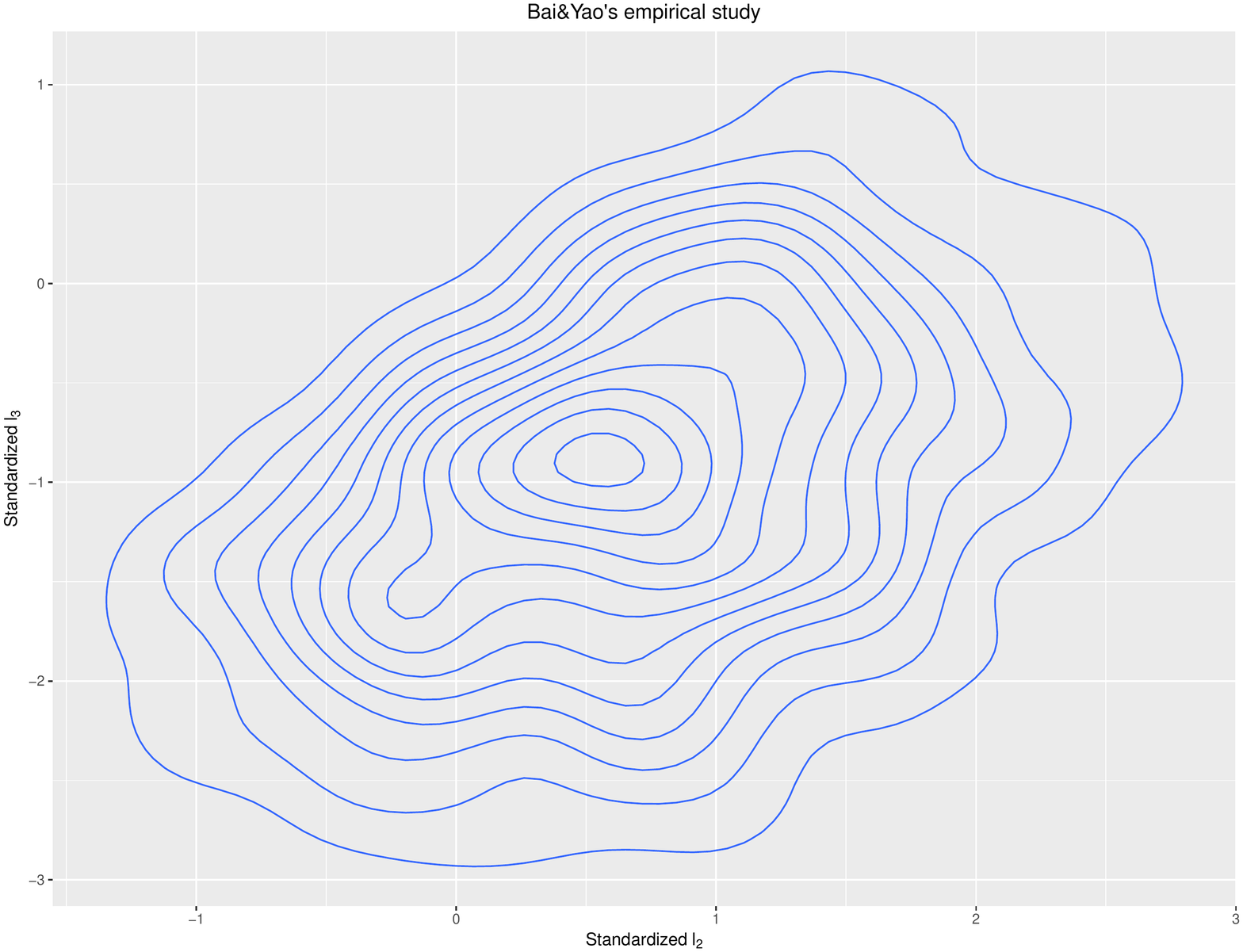}\\
\includegraphics[width = .3\textwidth]{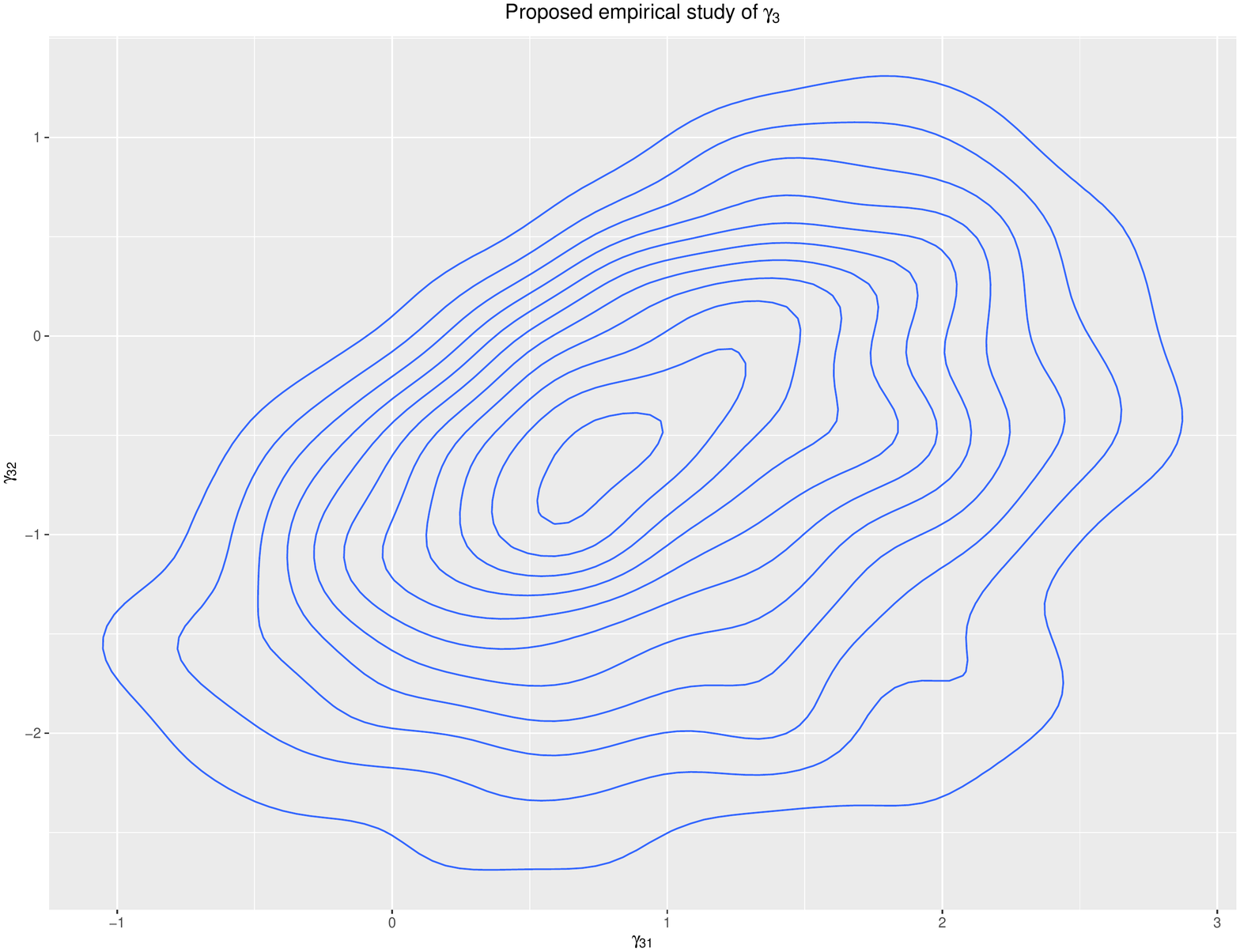}\quad \includegraphics[width = .3\textwidth] {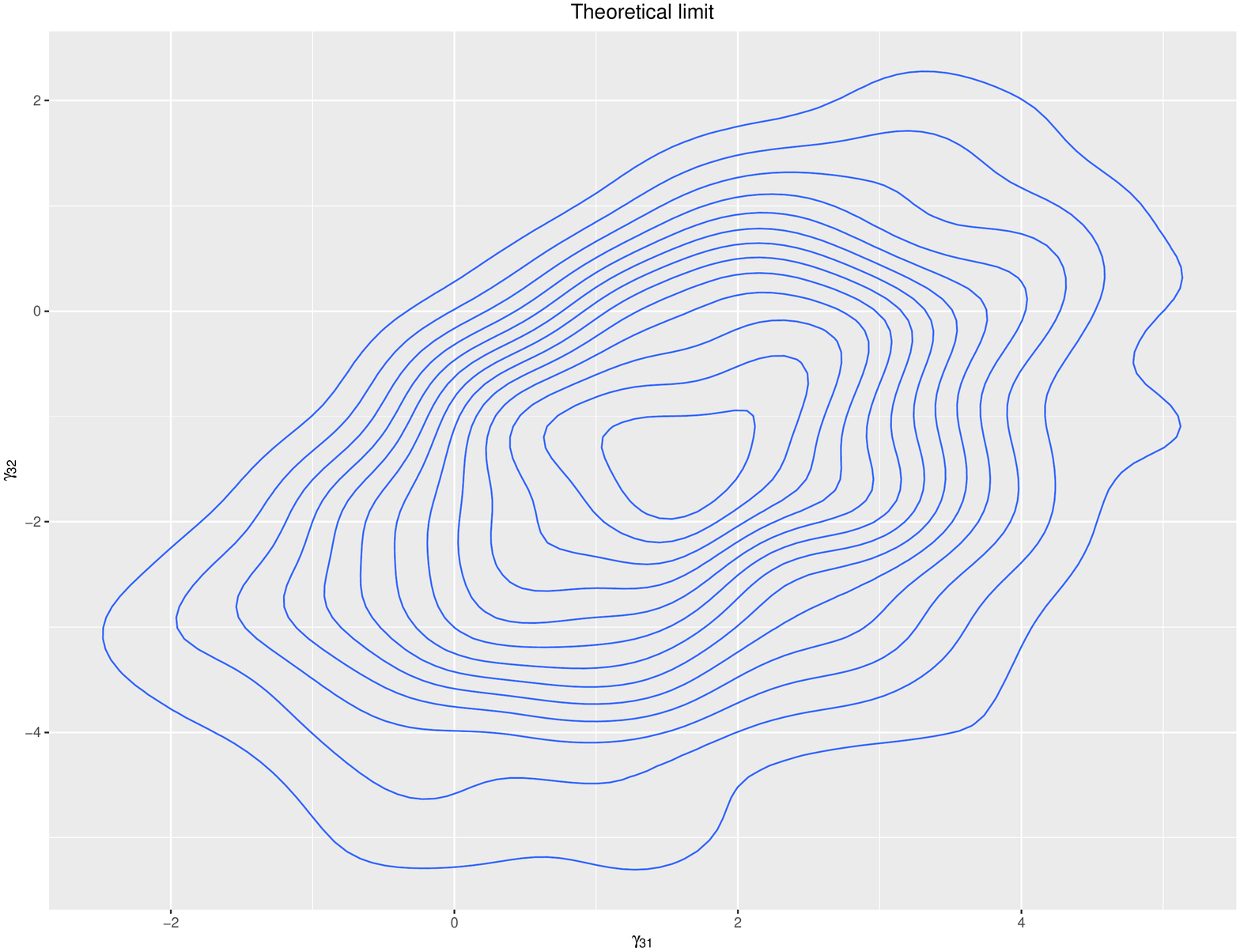}\quad \includegraphics[width = .3\textwidth] {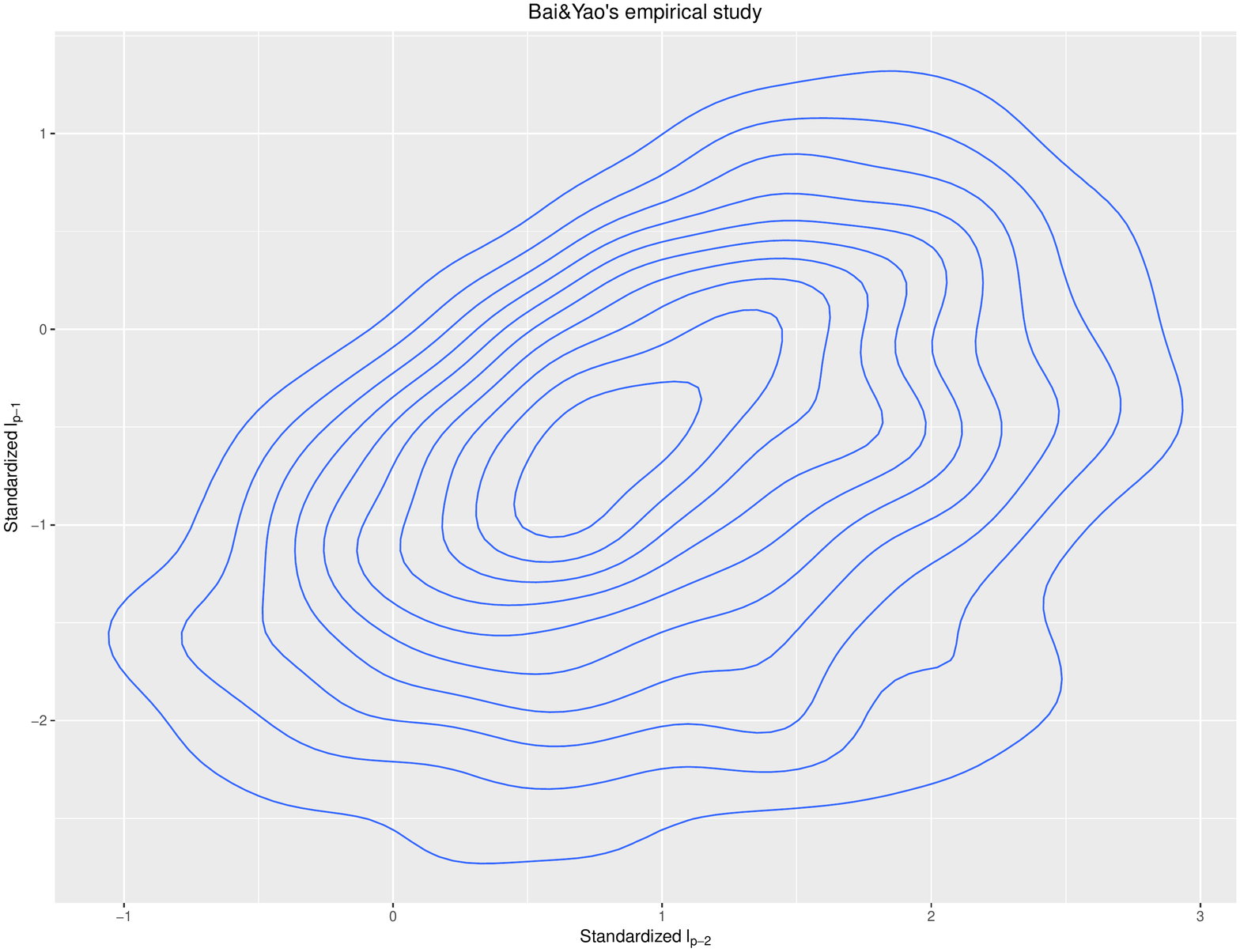}
\caption{ {\rm {\bf Case~I}} under Gaussian Assumption. Upper two panels show that the  histograms of the proposed $\gamma_1$ and $\gamma_4$ comparing to the ones of standardized $l_1$ and $l_p$ in Bai and Yao (2012) , as well as  the empirical densities (solid lines) and their Gaussian limits (dashed lines). Lower two panels show the  contour plots: the left ones are the proposed empirical joint density function of $(\gamma_{i1}, \gamma_{i2}), i=2,3.$; the middle ones are their corresponding limits; the right ones are the  empirical joint density function of
standardized $(l_2, l_3)$ and $(l_{p-2}, l_{p-1})$.
 }\label{fig:1}
\end{center} 
\end{figure}

 \subsection{{\rm {\bf Case~I}} under Binomial Assumption}
If $\{x_{ij}\}$  are from Binomial Assumption in the {\rm \bf Case~I}, then it is obtained by the Remark~\ref{rmk2}  as below: 
\begin{itemize}
\item First, for the single  population  spikes $\alpha_1=4$ and $\alpha_4=0.1$ , we have :
\[\gamma_1= \sqrt{n}\Big(\frac{l_{1}(S)}{\phi_{n,1}}-1\Big) \rightarrow N(0, \sigma_1^2)\]
where 
$\phi_{n,1}=4.667$, $\sigma_1^2=0.074$, 
and
\[\gamma_4= \sqrt{n}\Big(\frac{l_{p}(S)}{\phi_{n,4}}-1\Big) \rightarrow N(0, \sigma_4^2)\]
with
$\phi_{n,4}= 0.044,  \sigma_4^2=2.414.$ 

\item Second, for  the spikes $\alpha_2=3$ with multiplicity 2, we obtain
\[\gamma_2=\left(\gamma_{21}, \gamma_{22}\right)'=\left( \sqrt{n}\Big(\frac{l_{2}(S)}{\phi_{n,2}}-1\Big),
\sqrt{n}\Big(\frac{l_{3}(S)}{\phi_{n,2}}-1\Big)\right)' \]
converges to the eigenvalues of random matrix $-\frac1{\kappa_s}\left[\Omega_{\phi_{2}}\right]_{22}$,
where  $\phi_{n,2}=3.750$, $\kappa_s=1.417$ for the spike $\alpha_2=3$.  Furthermore, the
matrix 
$\left[\Omega_{\phi_{2}}\right]_{22}$
 is a $2\times 2$ symmetric matrix with the independent Gaussian entries, of which the $(i,j)$ element has mean zero and the variance given by 
 \[var(w_{ij})=
\left\{
\begin{array}{cc}
0.263, &~ \text{if}~ i = j  \\
 1.131,  &  ~ \text{if}~ i \neq j  
\end{array}
\right.
\]

Similarly, for  the spikes $\alpha_3=0.2$ with multiplicity 2,
\[\gamma_3=\left(\gamma_{31}, \gamma_{32}\right)'=\left( \sqrt{n}\Big(\frac{l_{p-2}(S)}{\phi_{n,3}}-1\Big),
\sqrt{n}\Big(\frac{l_{p-1}(S)}{\phi_{n,3}}-1\Big)\right)' \]
converges to the eigenvalues of random matrix $-\frac1{\kappa_s}\left[\Omega_{\phi_{3}}\right]_{33}$,
where  $\phi_{n,3}=0.075$, $\kappa_s=1.649$ for the spike $\alpha_3=0.2$.  Furthermore, the
matrix 
$\left[\Omega_{\phi_{3}}\right]_{33}$
 is a $2\times 2$ symmetric matrix with the independent Gaussian entries, of which the $(i,j)$ element has mean zero and the variance given by 
 \[var(w_{ij})=
\left\{
\begin{array}{cc}
4.481, &~ \text{if}~ i = j  \\
 6.873,  &  ~ \text{if}~ i \neq j
\end{array}
\right.
\]
 \end{itemize}
The simulated empirical distributions of the spiked eigenvalues from Binomial Assumption under {\rm \bf Case~I}  are drawn in Figure~\ref{fig:2} in contrast to their corresponding limiting distributions. 
 \begin{figure}[htbp]
\begin{center}
\includegraphics[width = .37\textwidth]{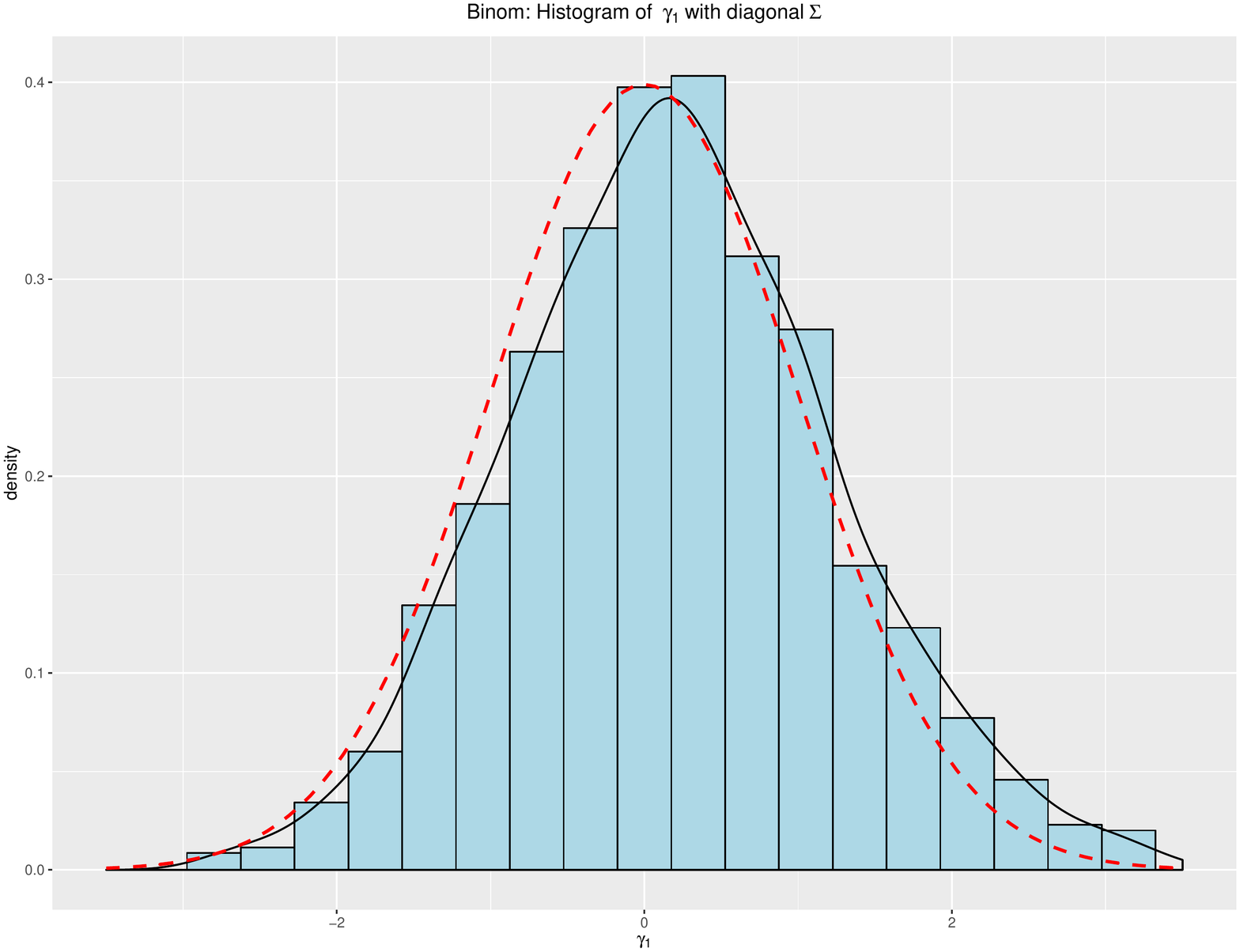}\quad\quad \includegraphics[width = .37\textwidth] {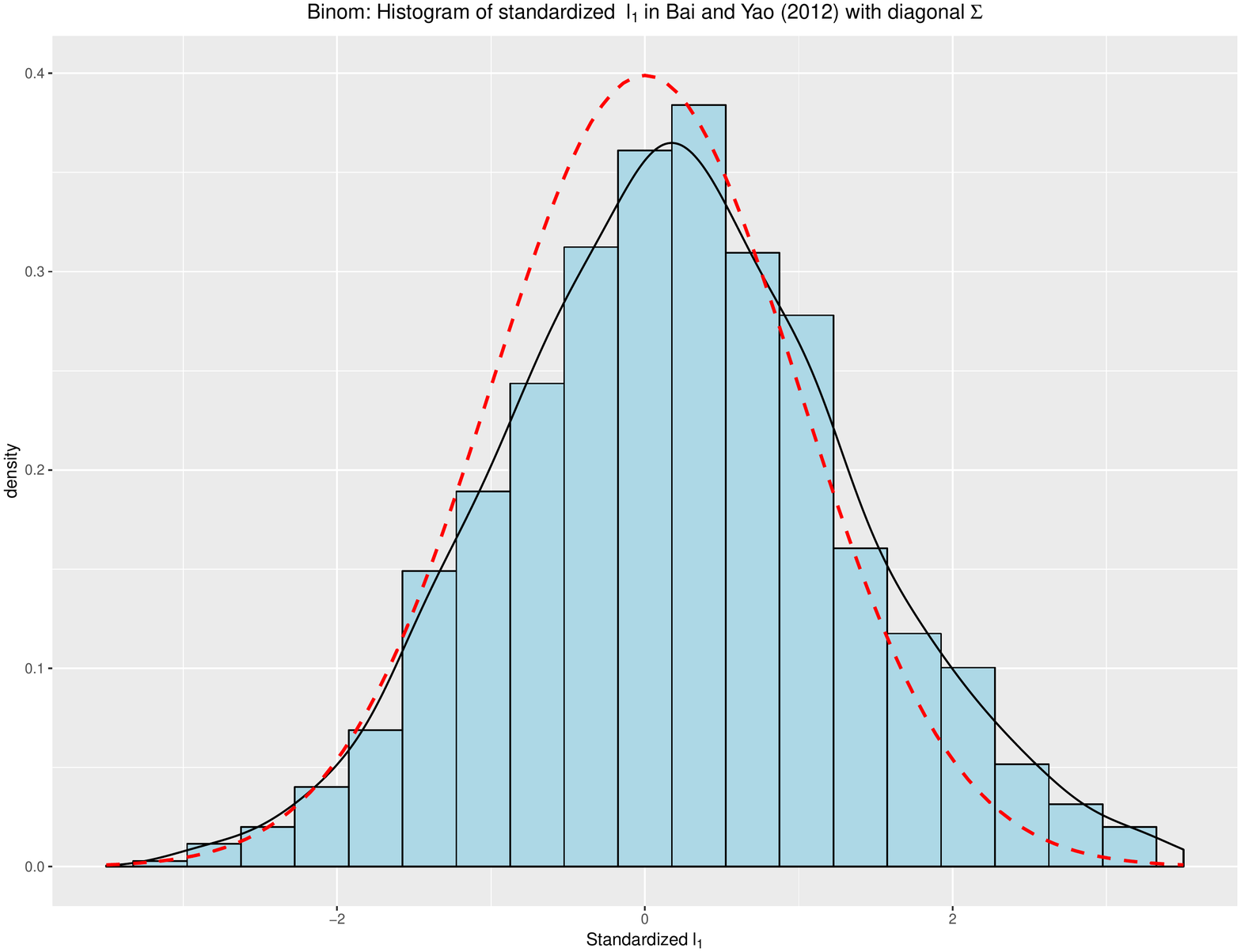}\\
\includegraphics[width = .37\textwidth]{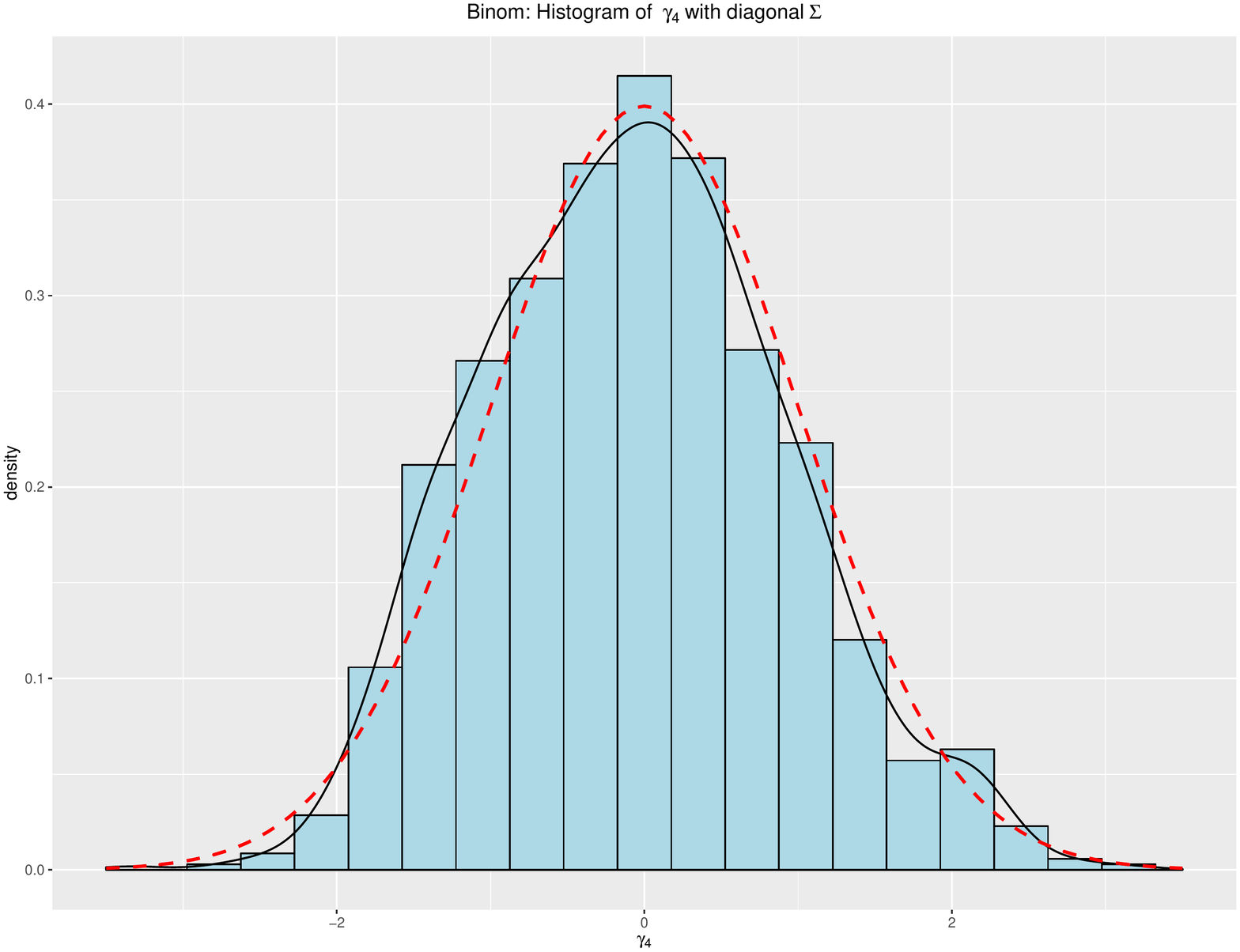}\quad\quad \includegraphics[width = .37\textwidth] {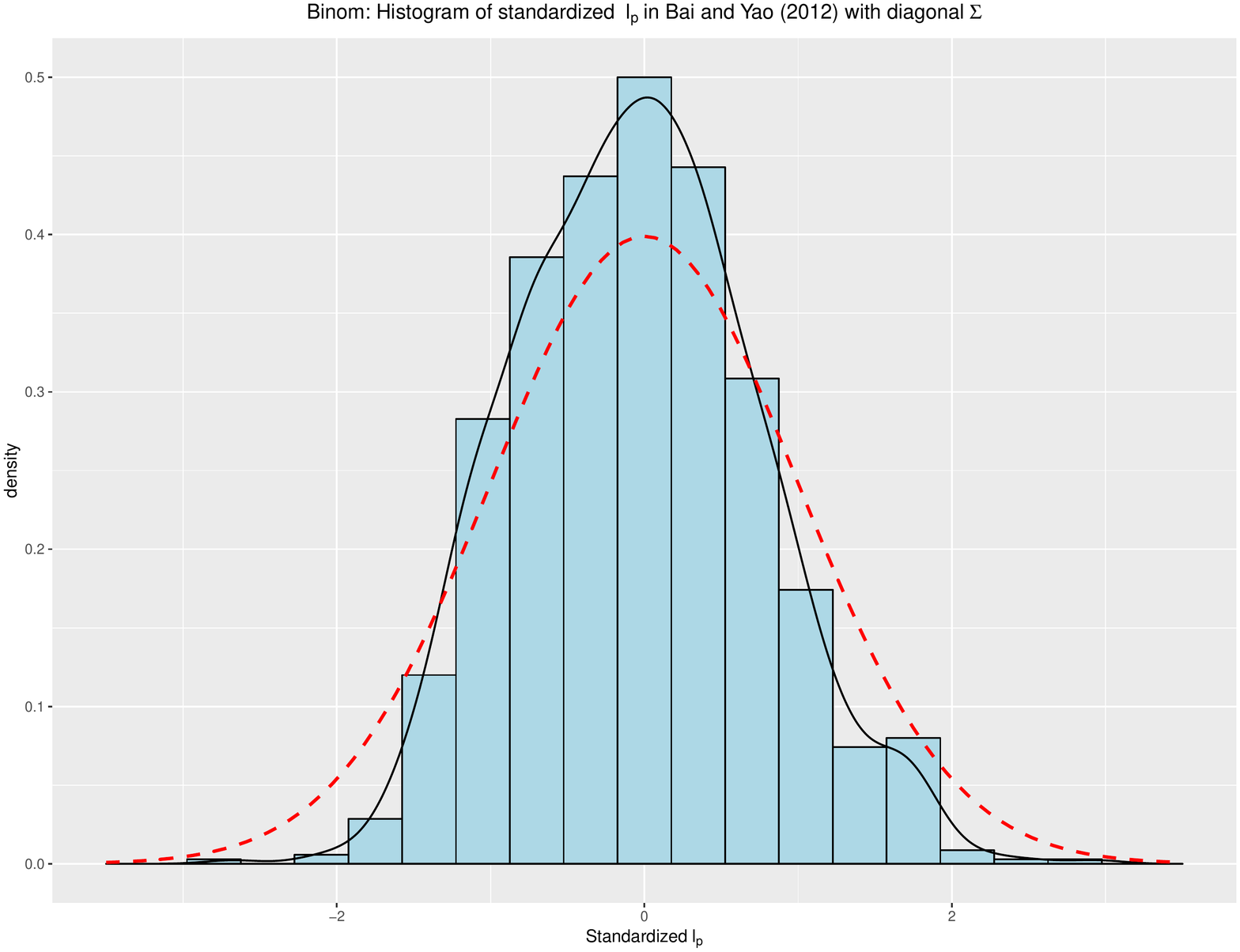}\\
\includegraphics[width = .3\textwidth]{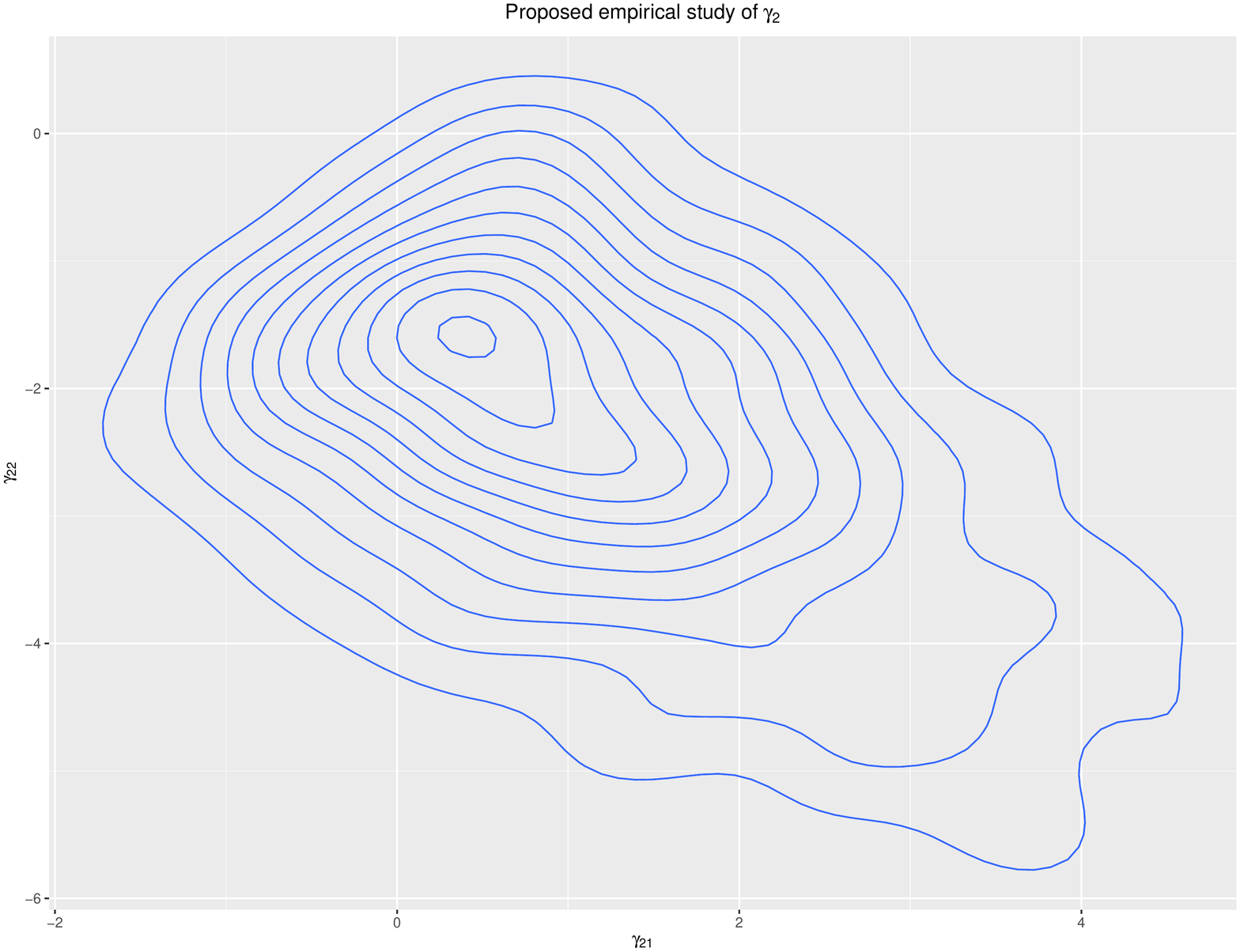}\quad \includegraphics[width = .3\textwidth] {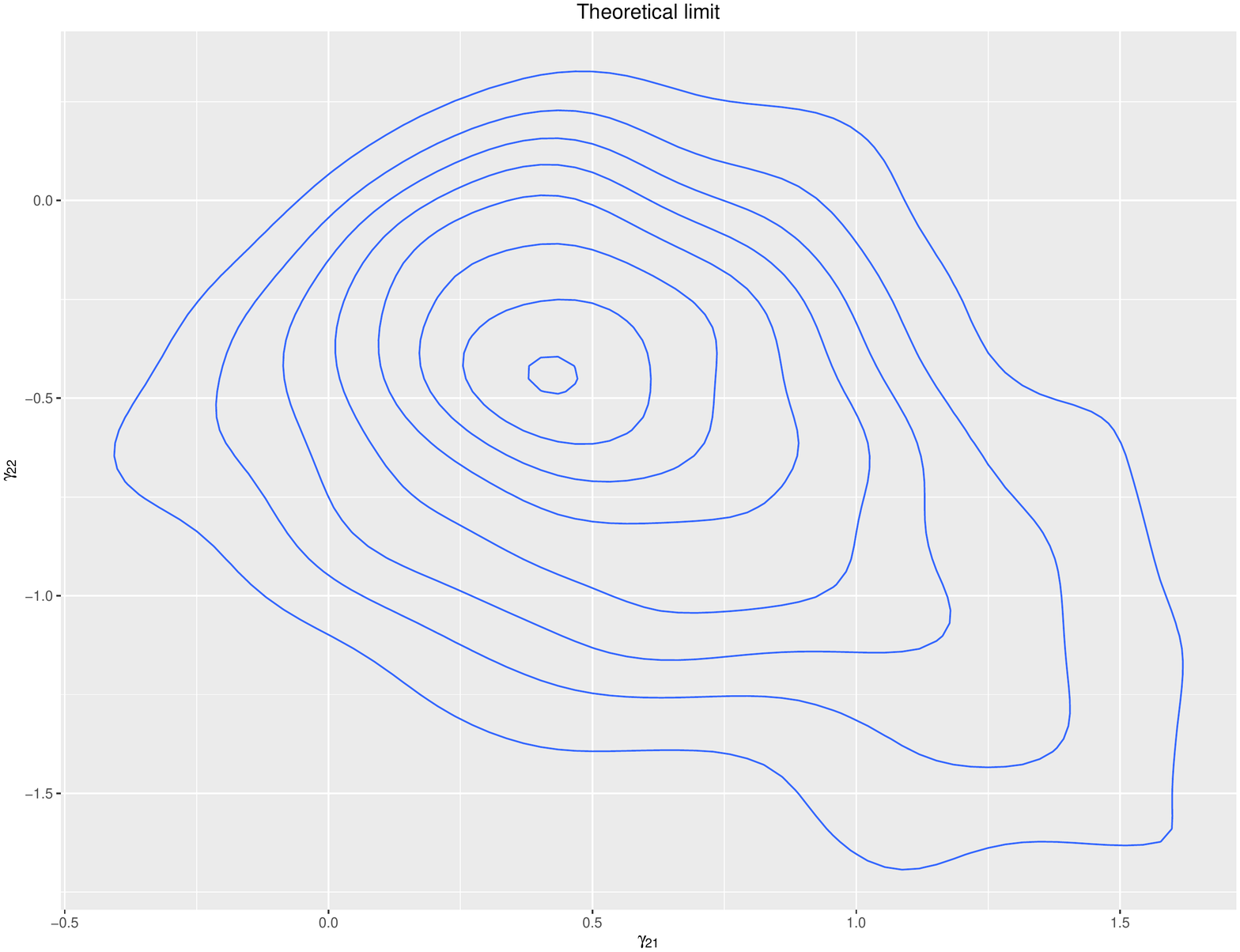}\quad \includegraphics[width = .3\textwidth] {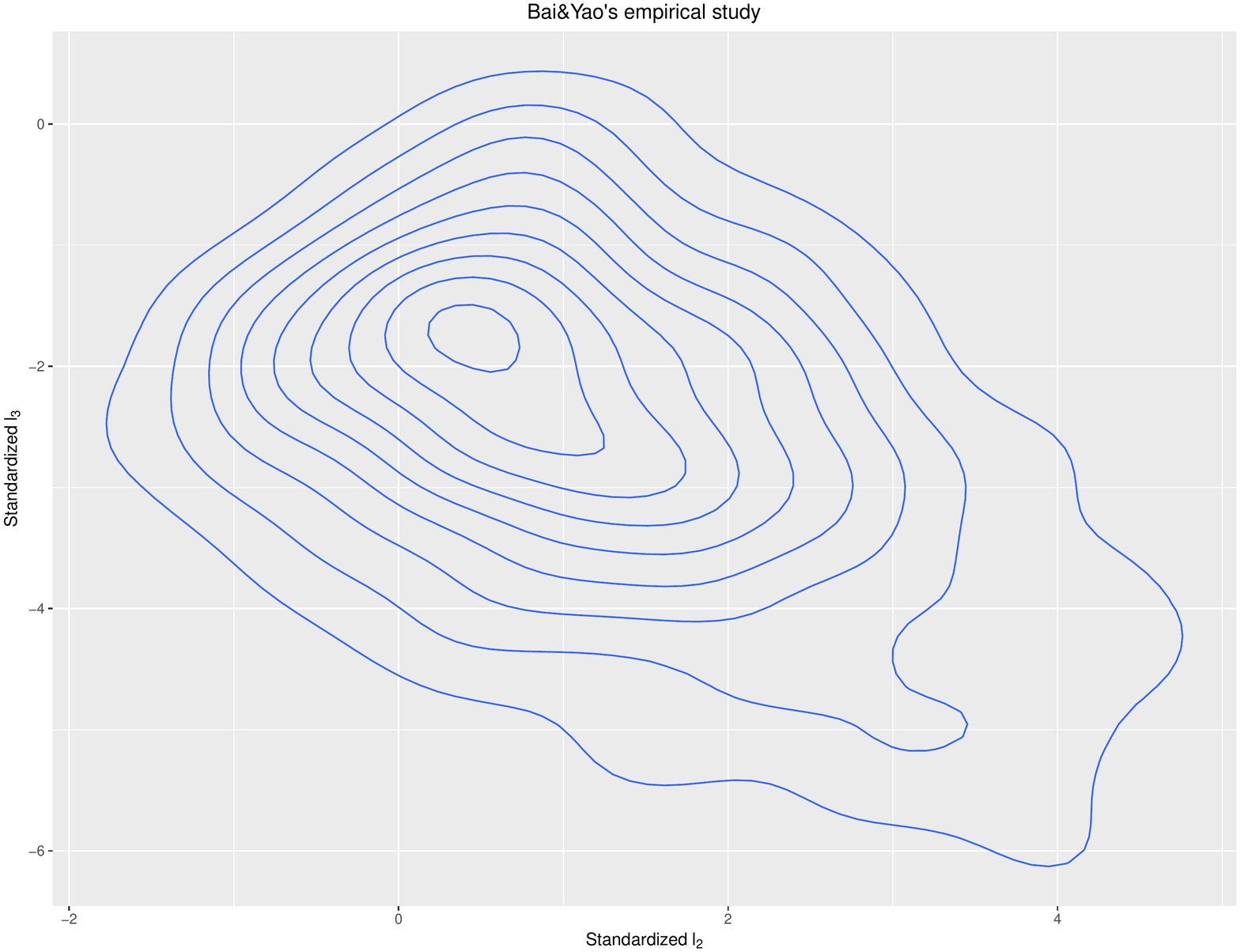}\\
\includegraphics[width = .3\textwidth]{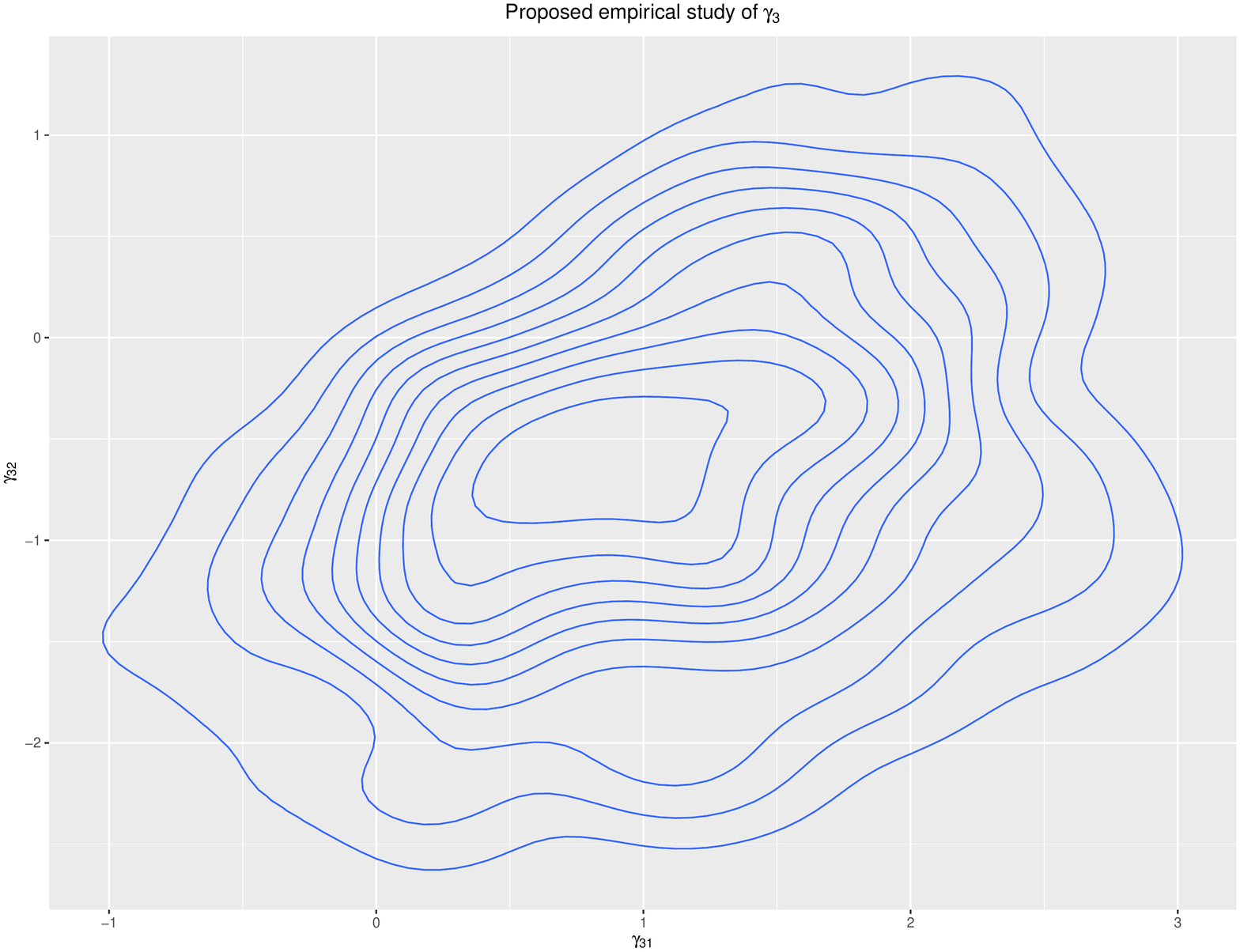}\quad \includegraphics[width = .3\textwidth] {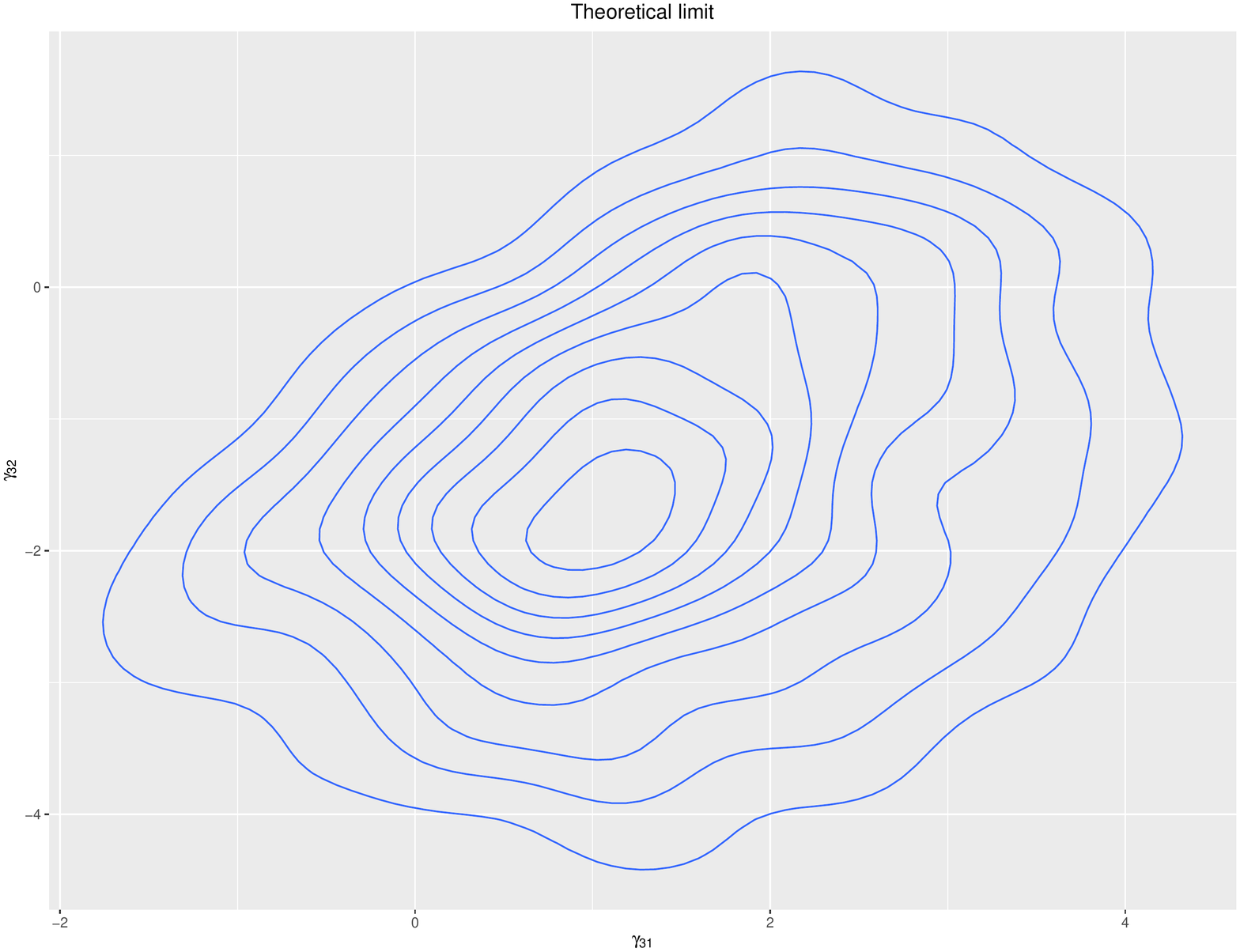}\quad \includegraphics[width = .3\textwidth] {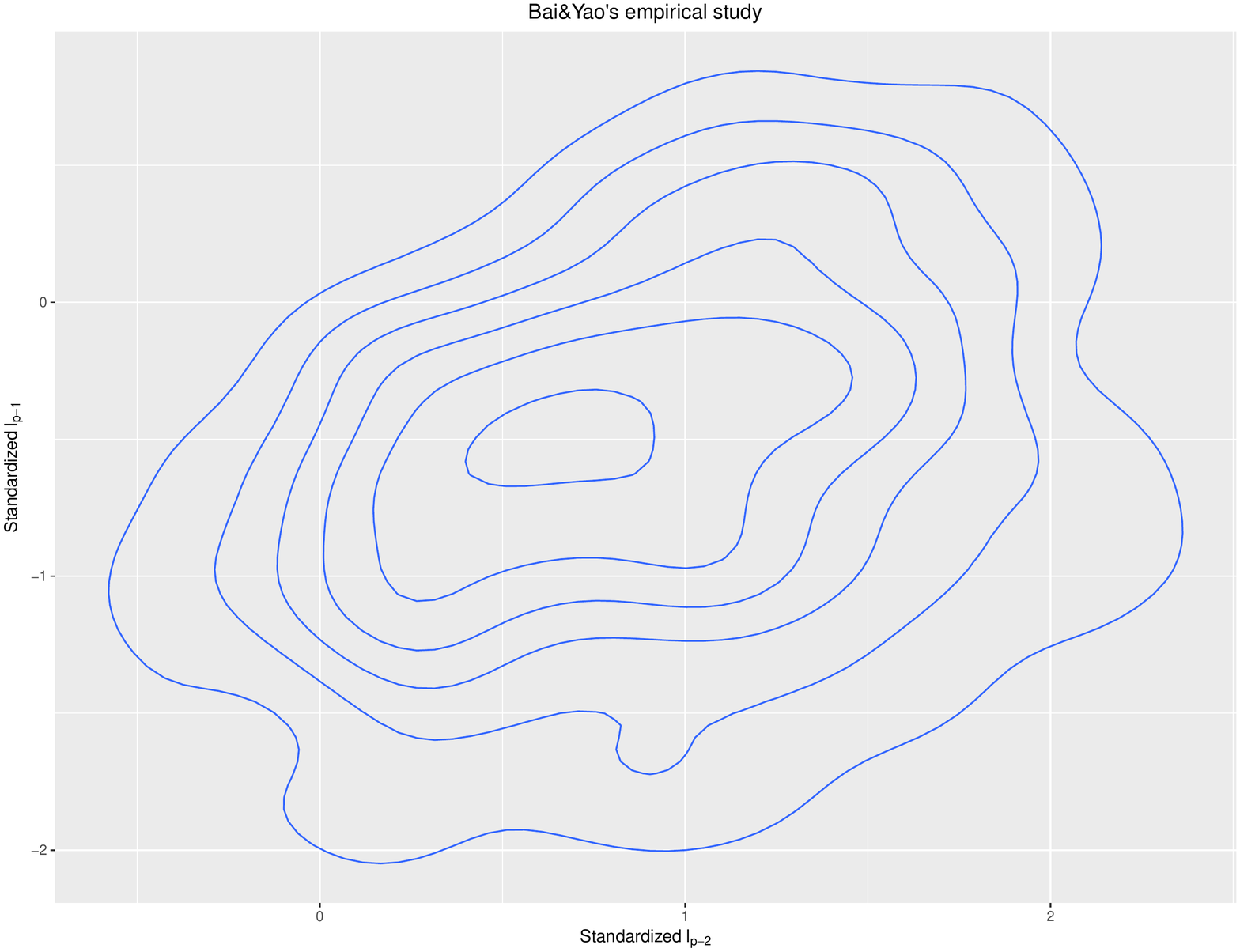}
\caption{ {\rm {\bf Case~I}} under Binomial Assumption. 
 }\label{fig:2}
\end{center} 
\end{figure}

As shown in the simulations of the {\rm \bf Case~I},
our approach provides the similar results to the ones 
 in \cite{BaiYao2012}, when the population covariance matrix has a diagonal structure. 
  Moreover,
our method performs slightly better for the non-Gaussian distribution even if  the  diagonal  independent assumption holds in the 
 {\rm \bf Case~I}.
 
 \subsection{{\rm \bf Case~II} under the both Assumptions}

For the {\rm \bf Case~II}, it is easily obtained by Theorem~{\ref{CLT} that our proposed results of the both population assumptions are the same to the one of Gaussian Assumption in {\rm \bf Case~I}, which can well fit their corresponding  limiting behaviors. However, as  shown in  the simulated results, the asymptotic distribution in \cite{BaiYao2012}, which is  involved with the 4th moment, performs not well for the non-Gaussian population assumption in the  {\rm \bf Case~II}.
Therefore, it is reasonable to theoretically remove the diagonal  independent restrictions in results of Bai and Yao (2008,2012) as illustrated in the simulations.
The simulated results of the two population assumptions in {\rm \bf Case~II} are respectively depicted in Figures~\ref{fig:4}  and \ref{fig:5}. 

 \begin{figure}[htbp]
\begin{center}
\includegraphics[width = .37\textwidth]{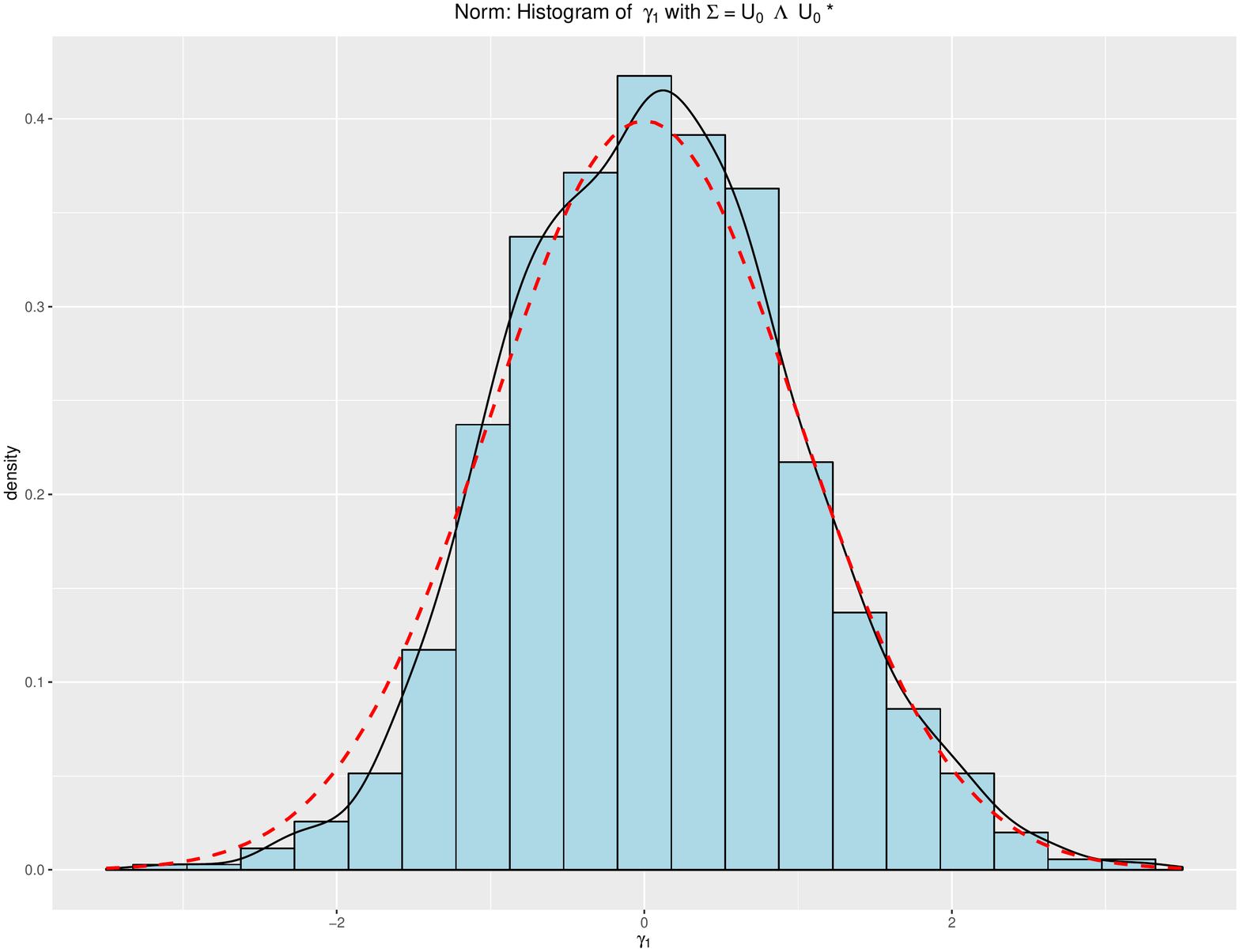}\quad\quad \includegraphics[width = .37\textwidth] {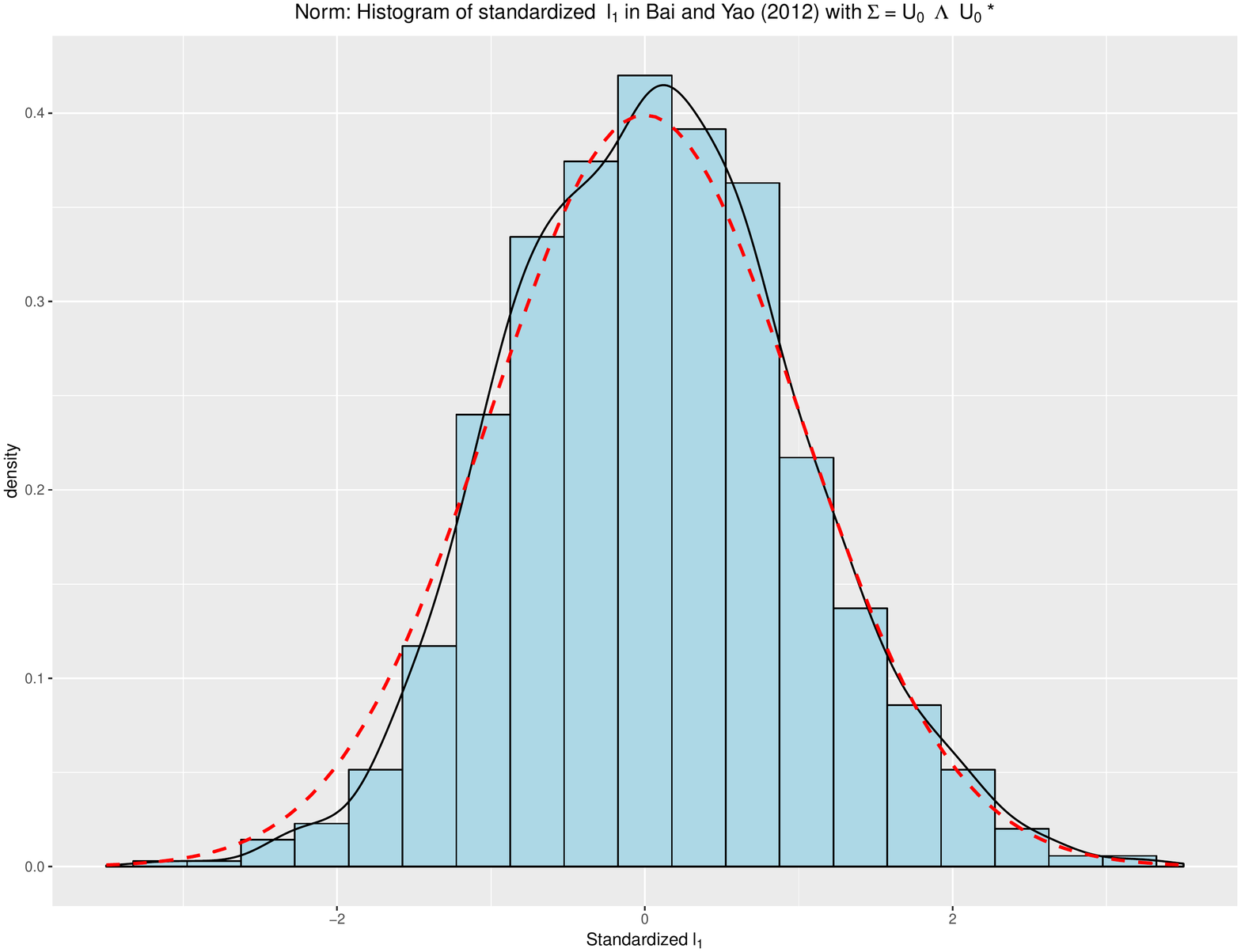}\\
\includegraphics[width = .37\textwidth]{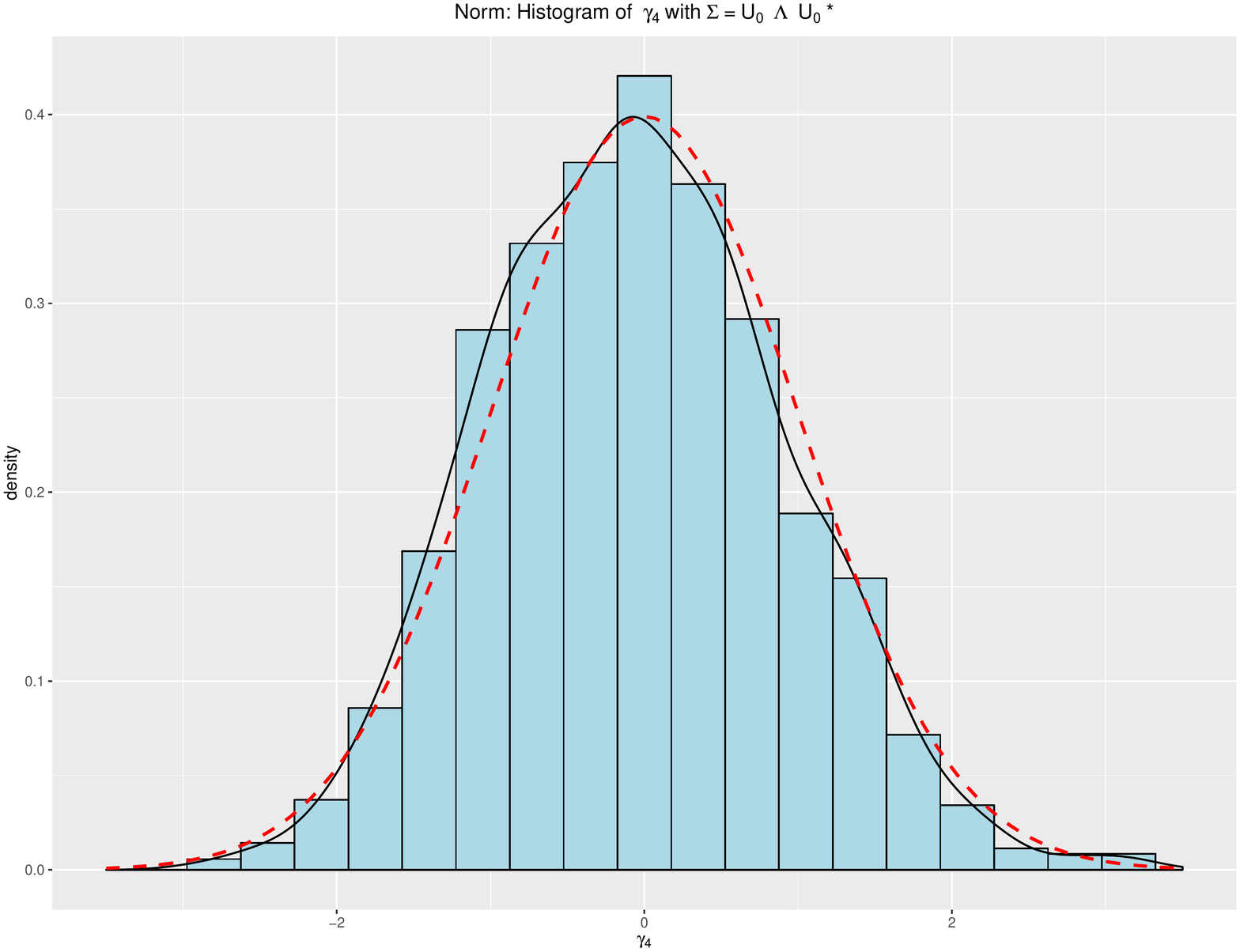}\quad\quad \includegraphics[width = .37\textwidth] {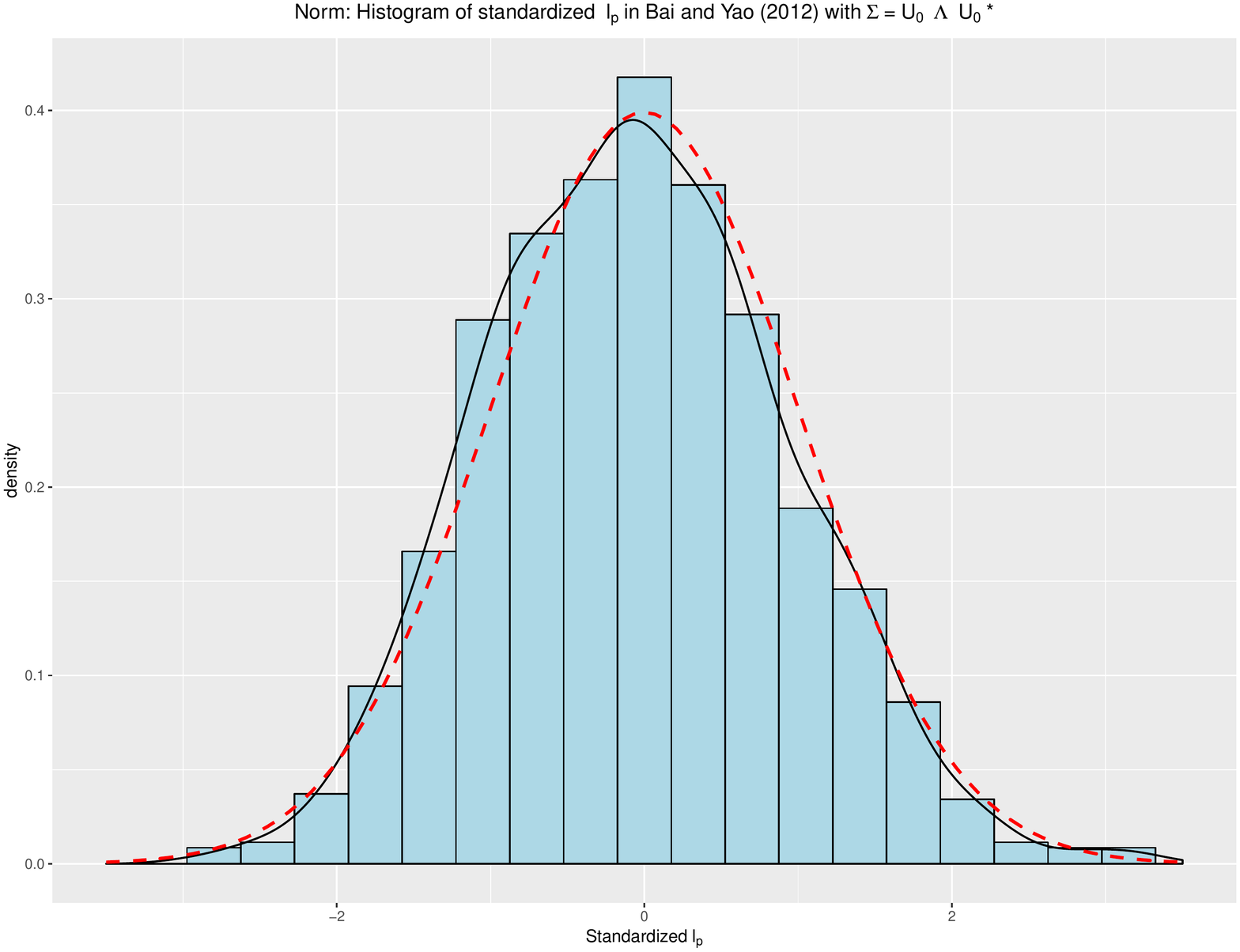}\\
\includegraphics[width = .3\textwidth]{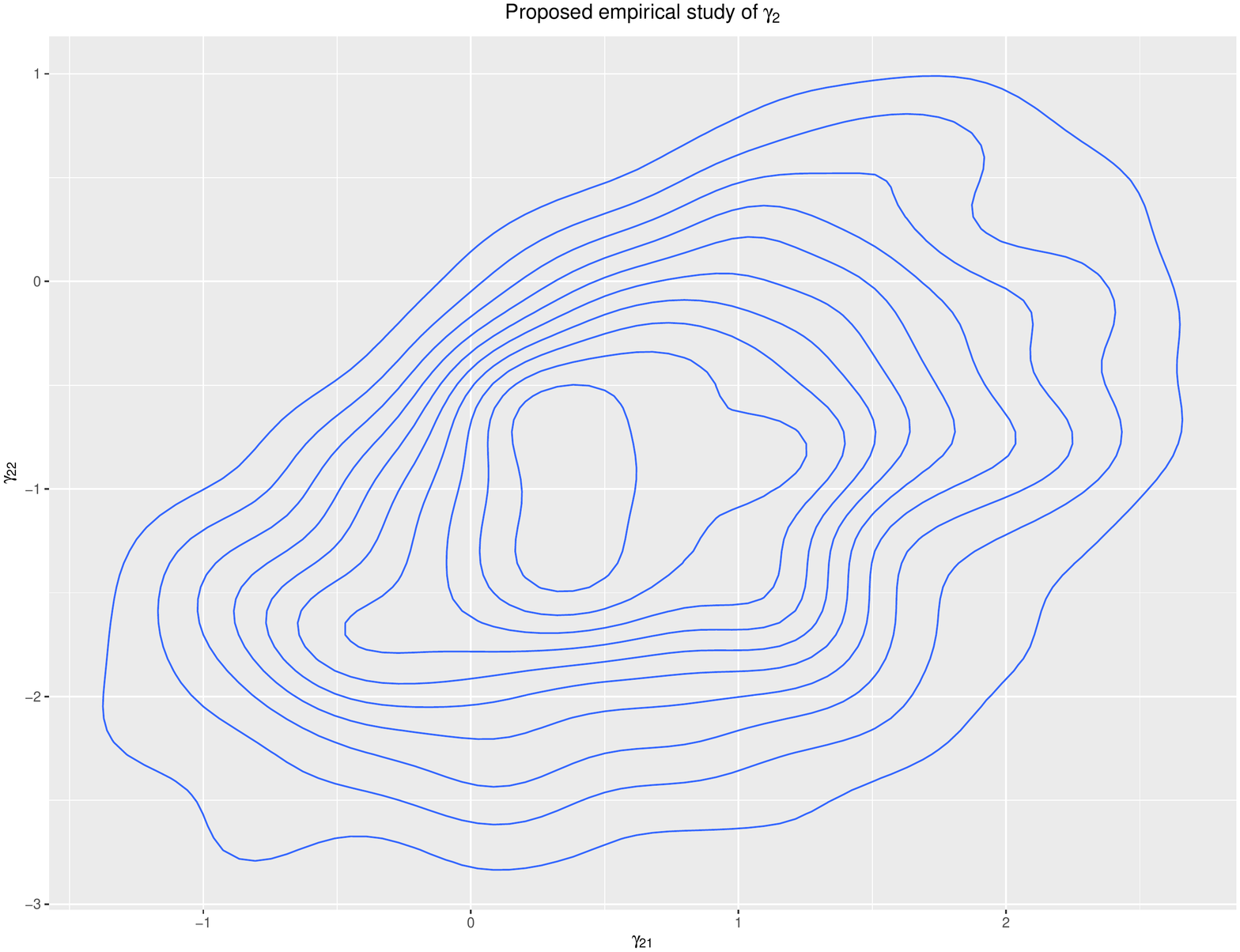}\quad \includegraphics[width = .3\textwidth] {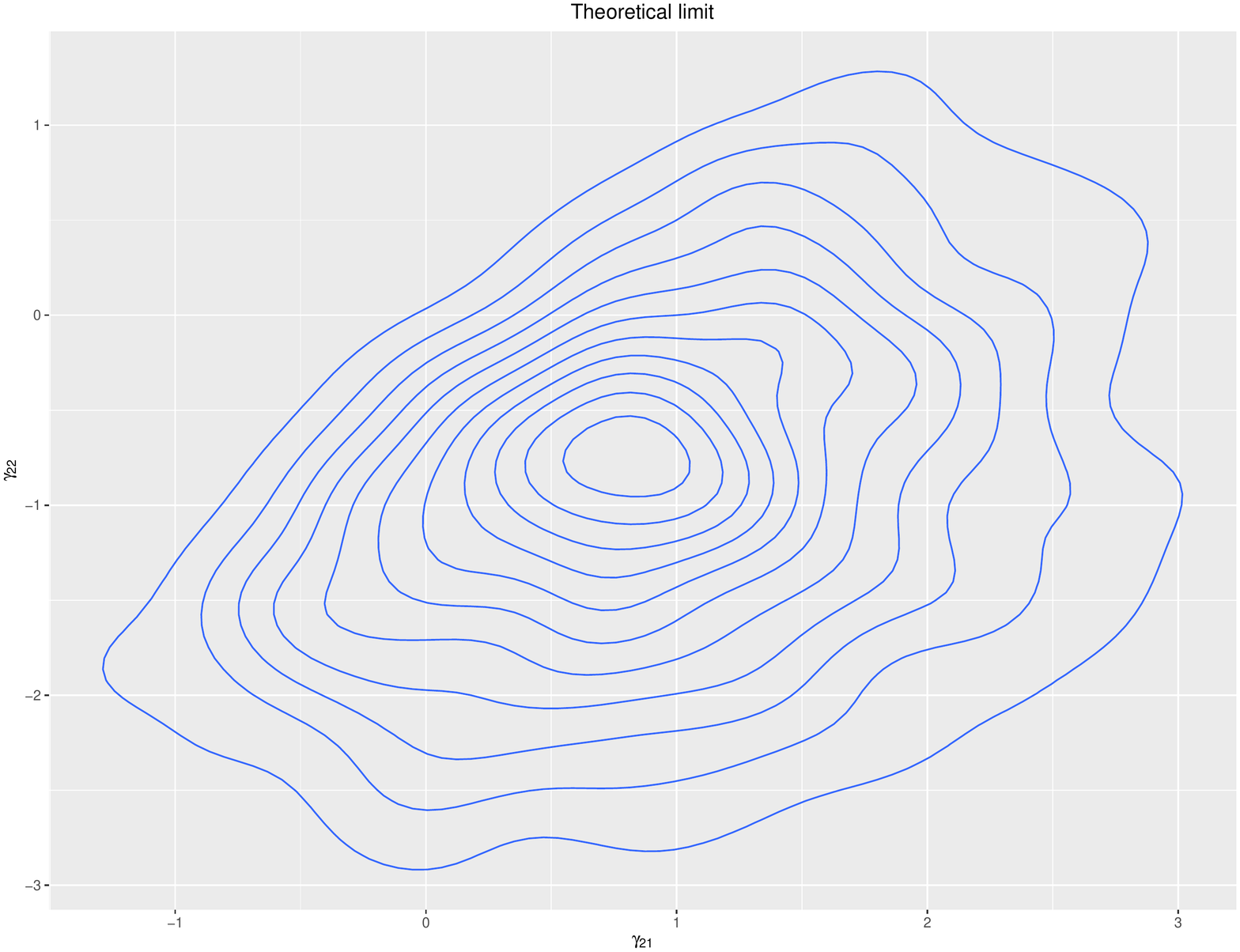}\quad \includegraphics[width = .3\textwidth] {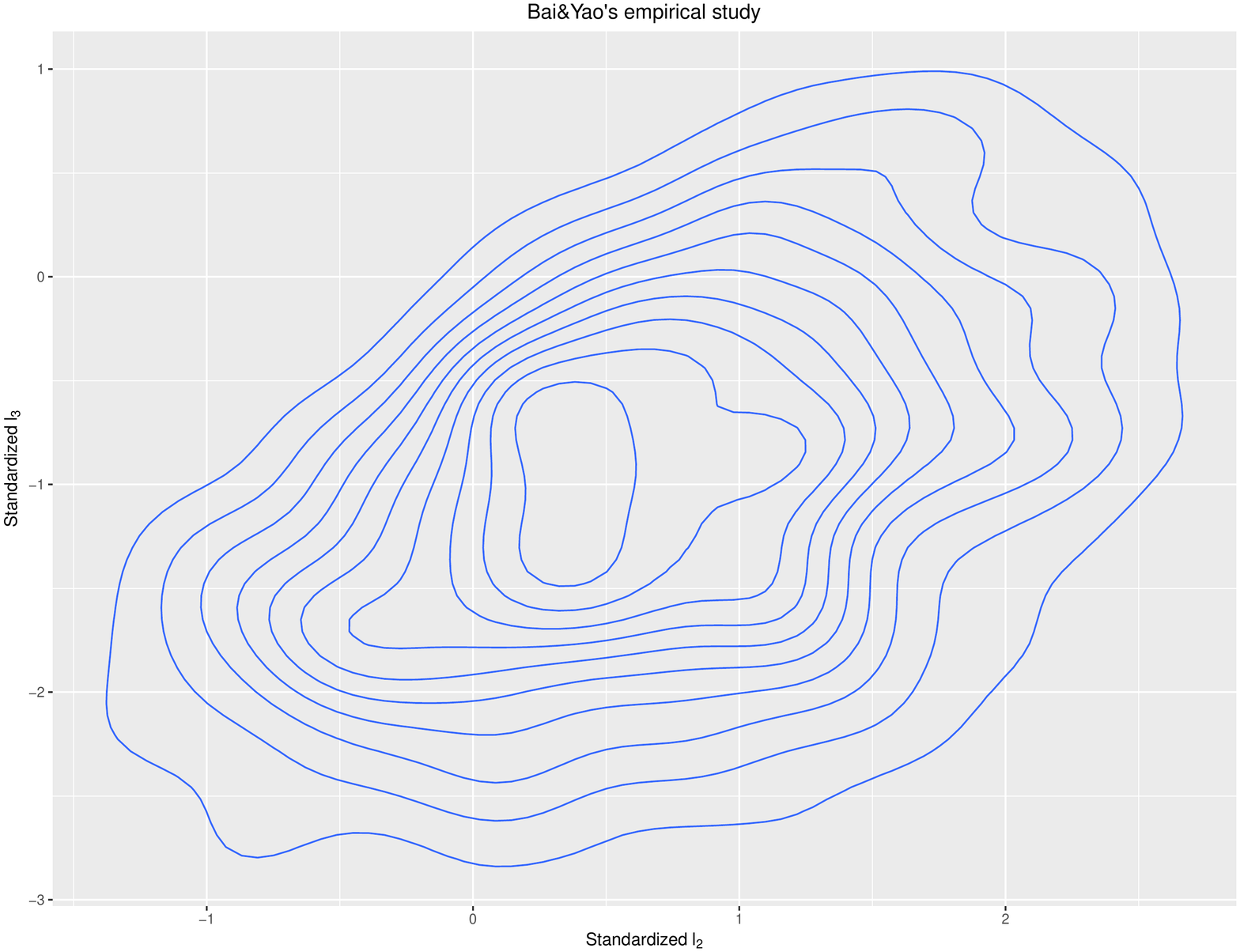}\\
\includegraphics[width = .3\textwidth]{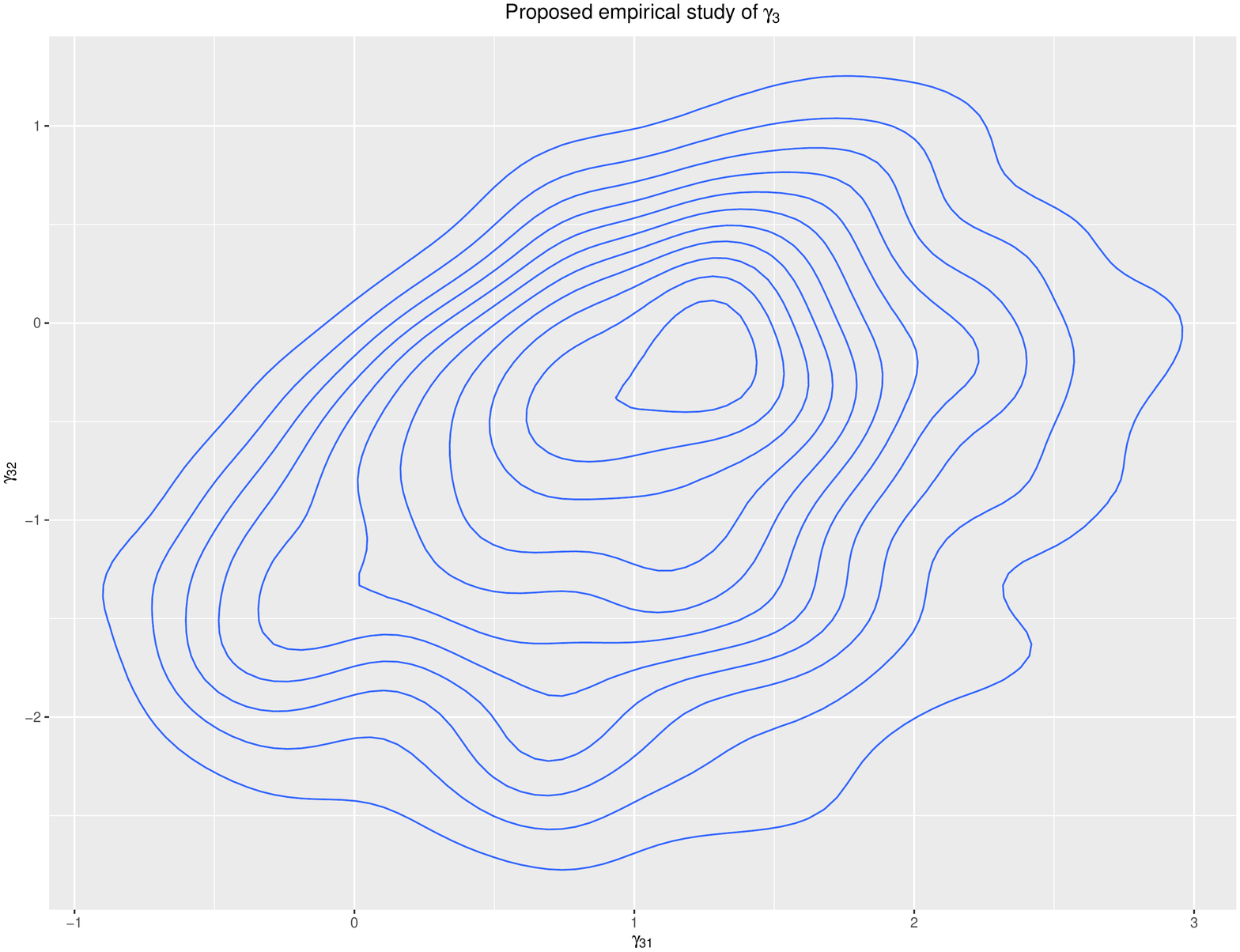}\quad \includegraphics[width = .3\textwidth] {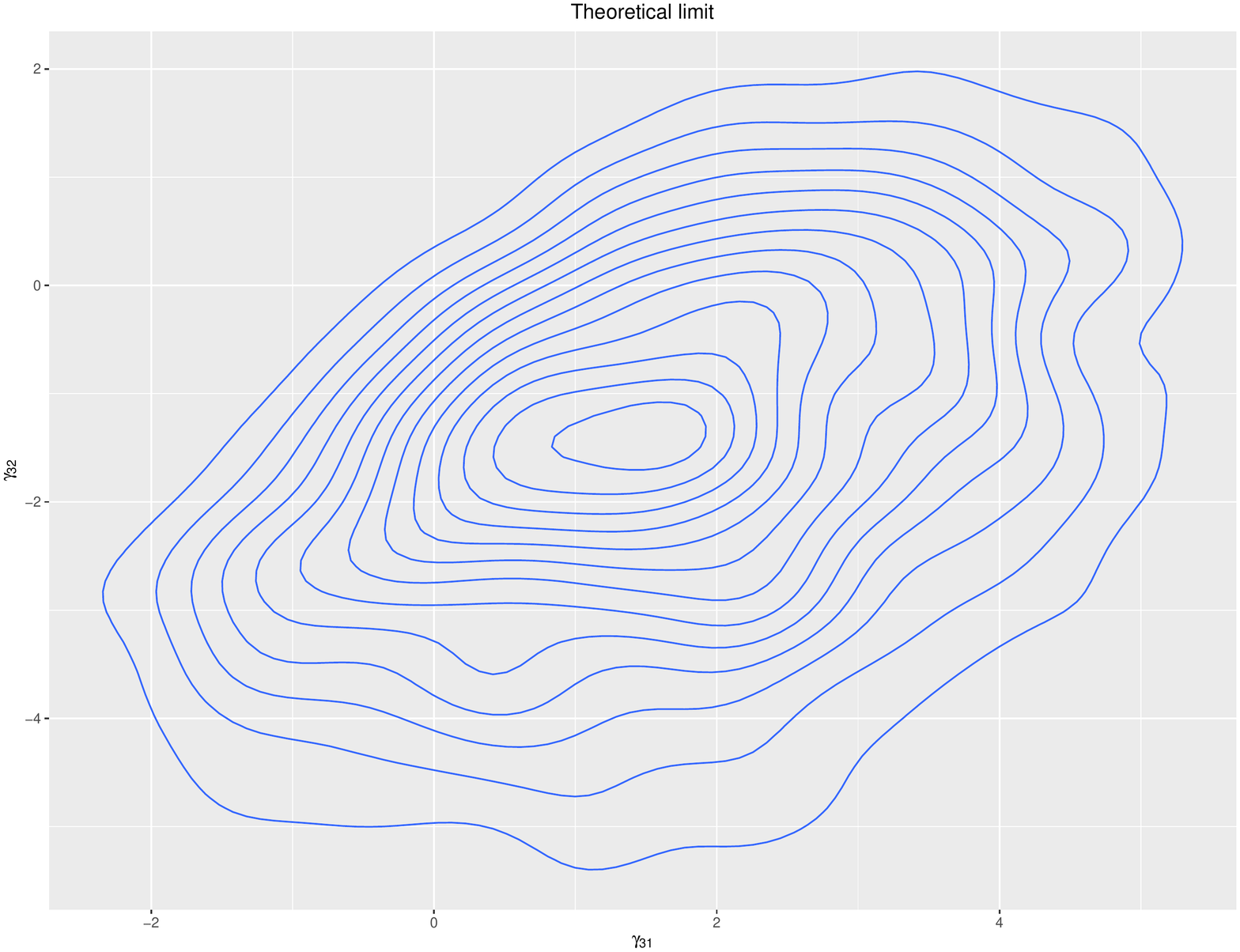}\quad \includegraphics[width = .3\textwidth] {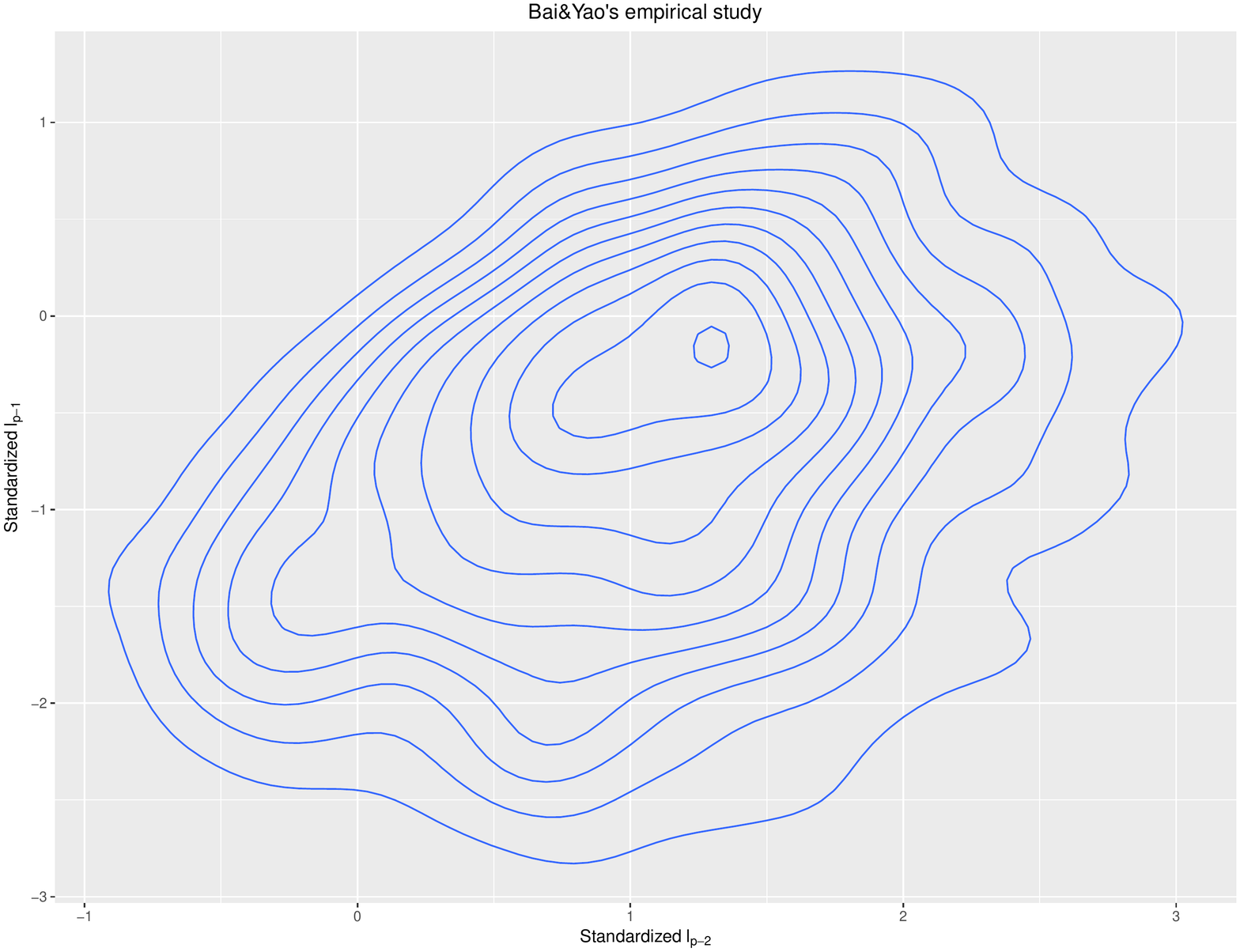}
\caption{ {\rm {\bf Case~II}} under Gaussian assumption. 
 }\label{fig:4}
\end{center} 
\end{figure}

 \begin{figure}[htbp]
\begin{center}
\includegraphics[width = .37\textwidth]{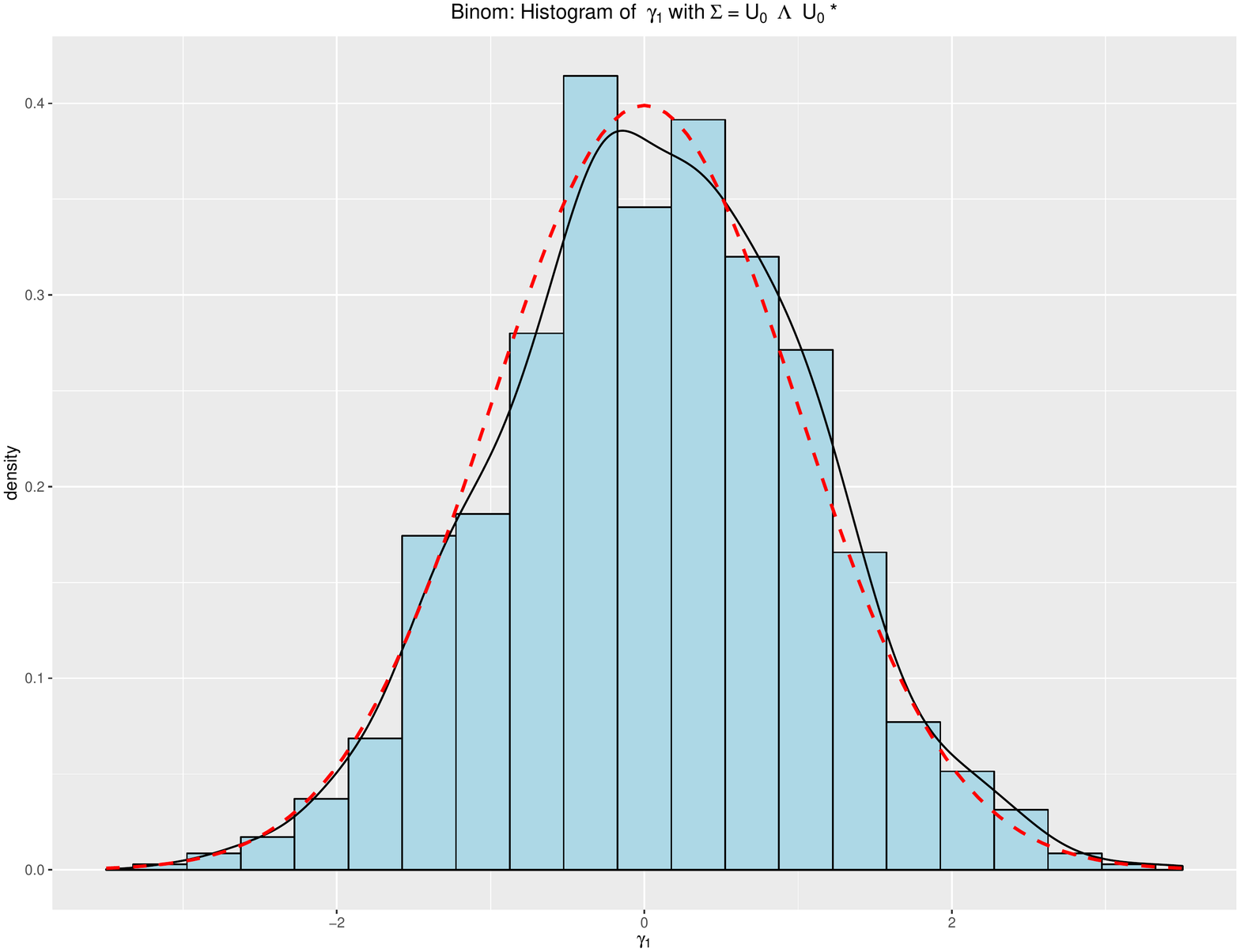}\quad\quad \includegraphics[width = .37\textwidth] {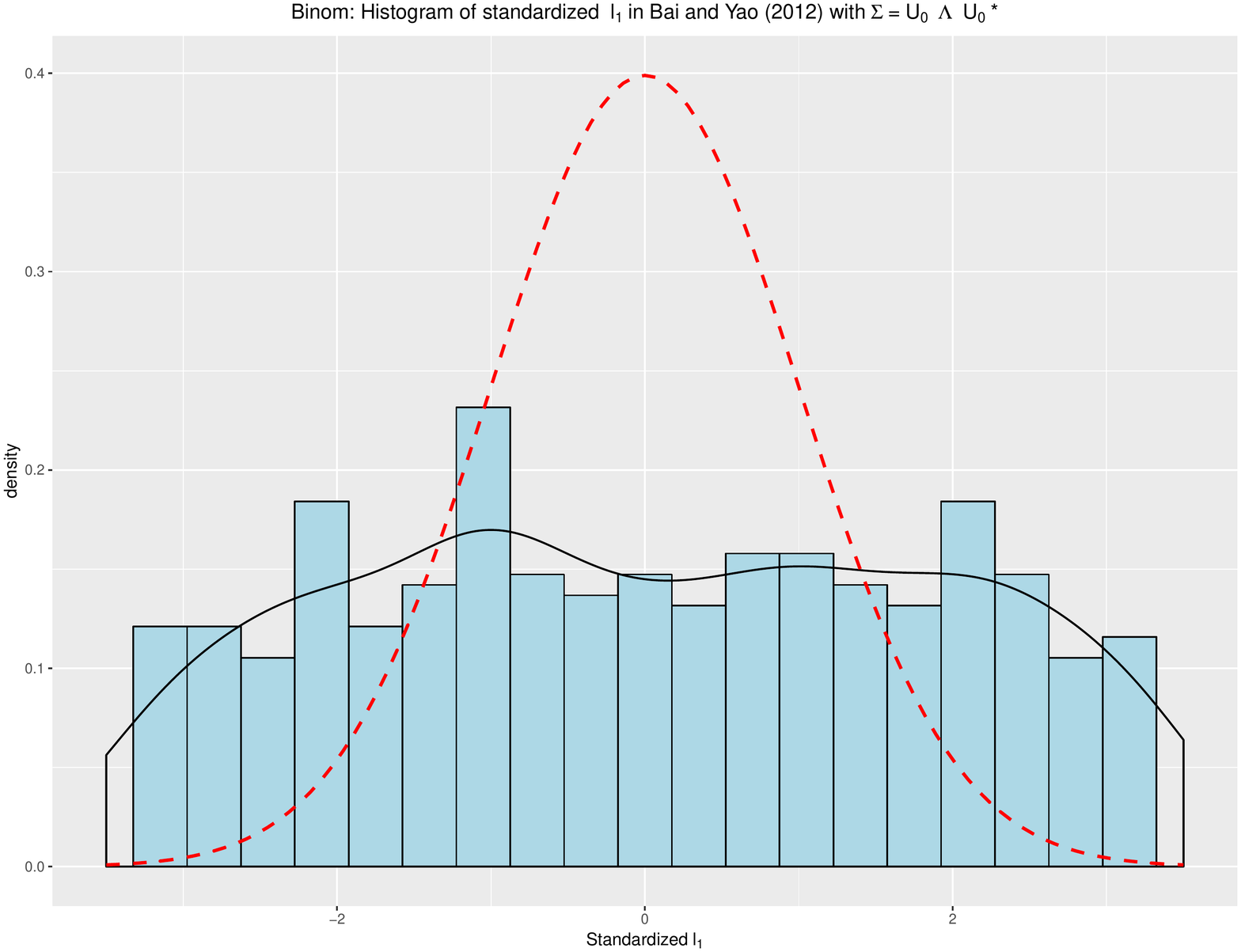}\\
\includegraphics[width = .37\textwidth]{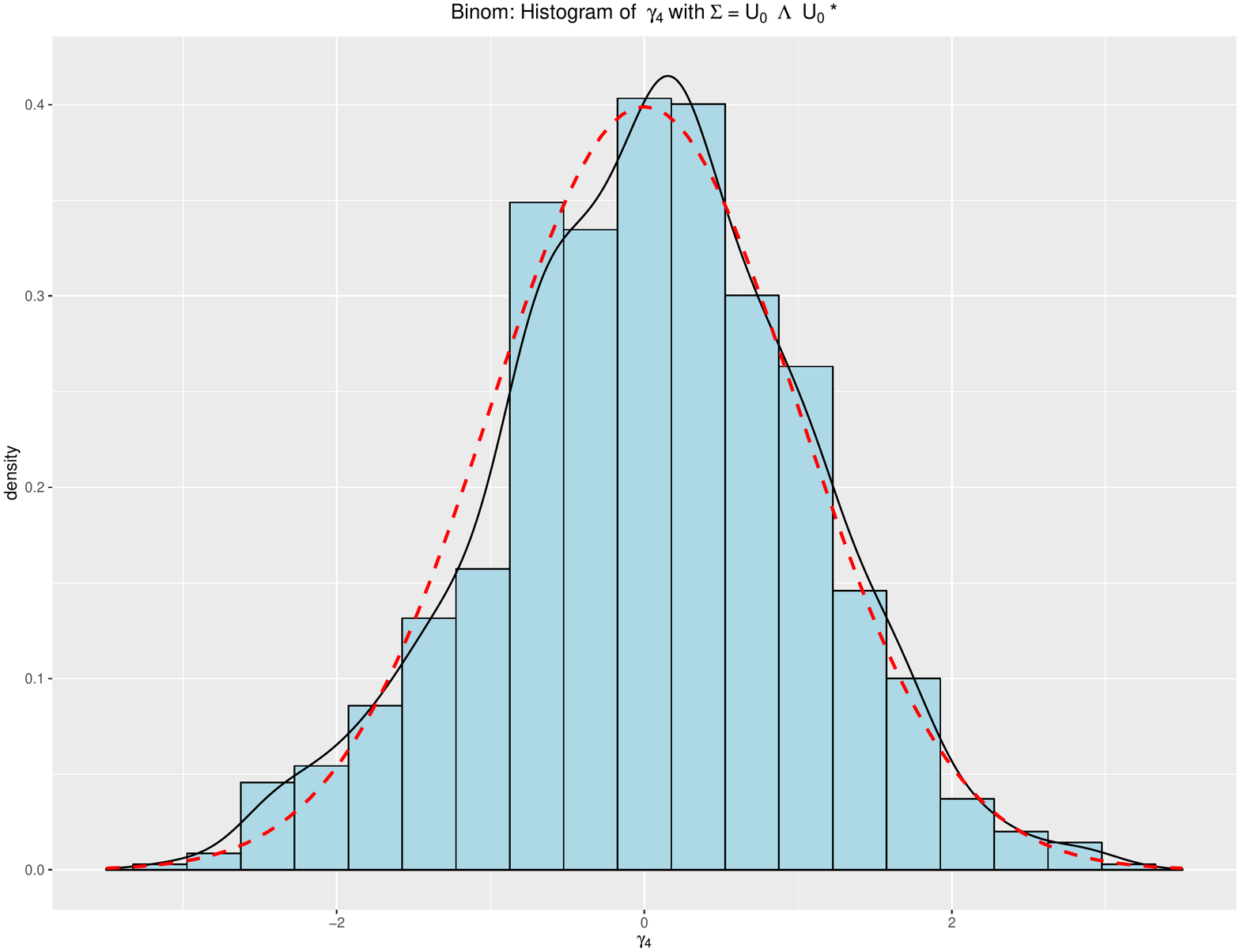}\quad\quad \includegraphics[width = .37\textwidth] {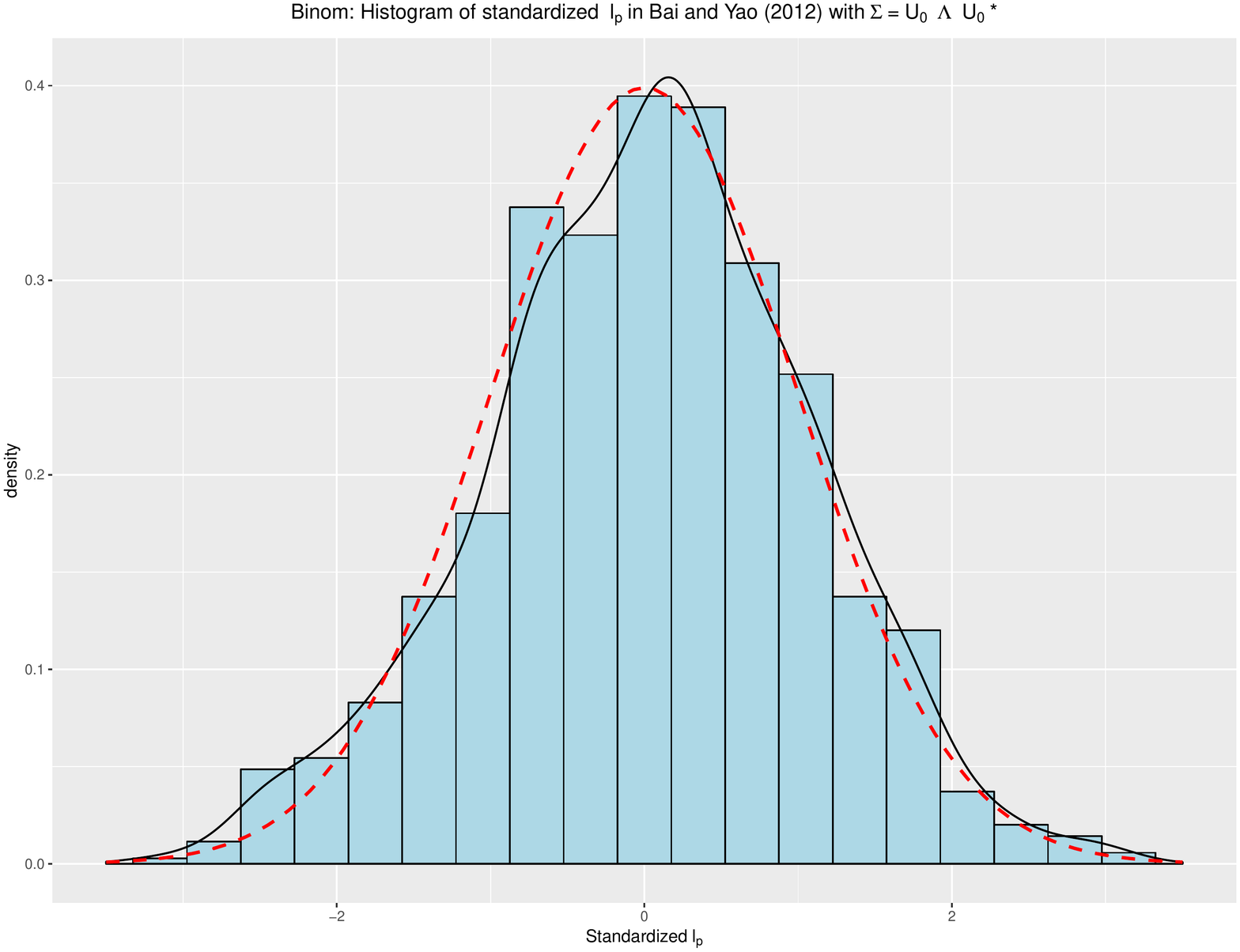}\\
\includegraphics[width = .3\textwidth]{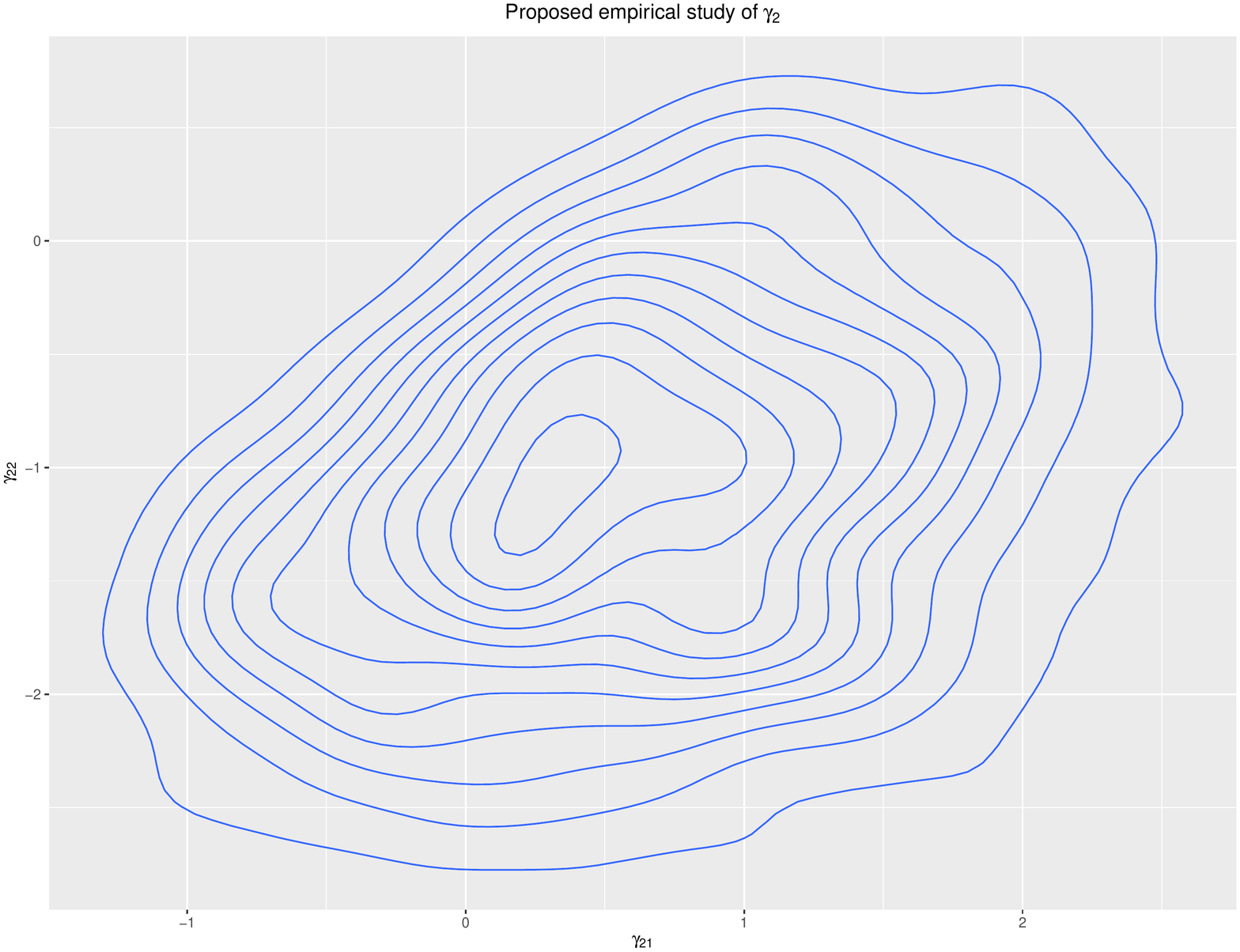}\quad \includegraphics[width = .3\textwidth] {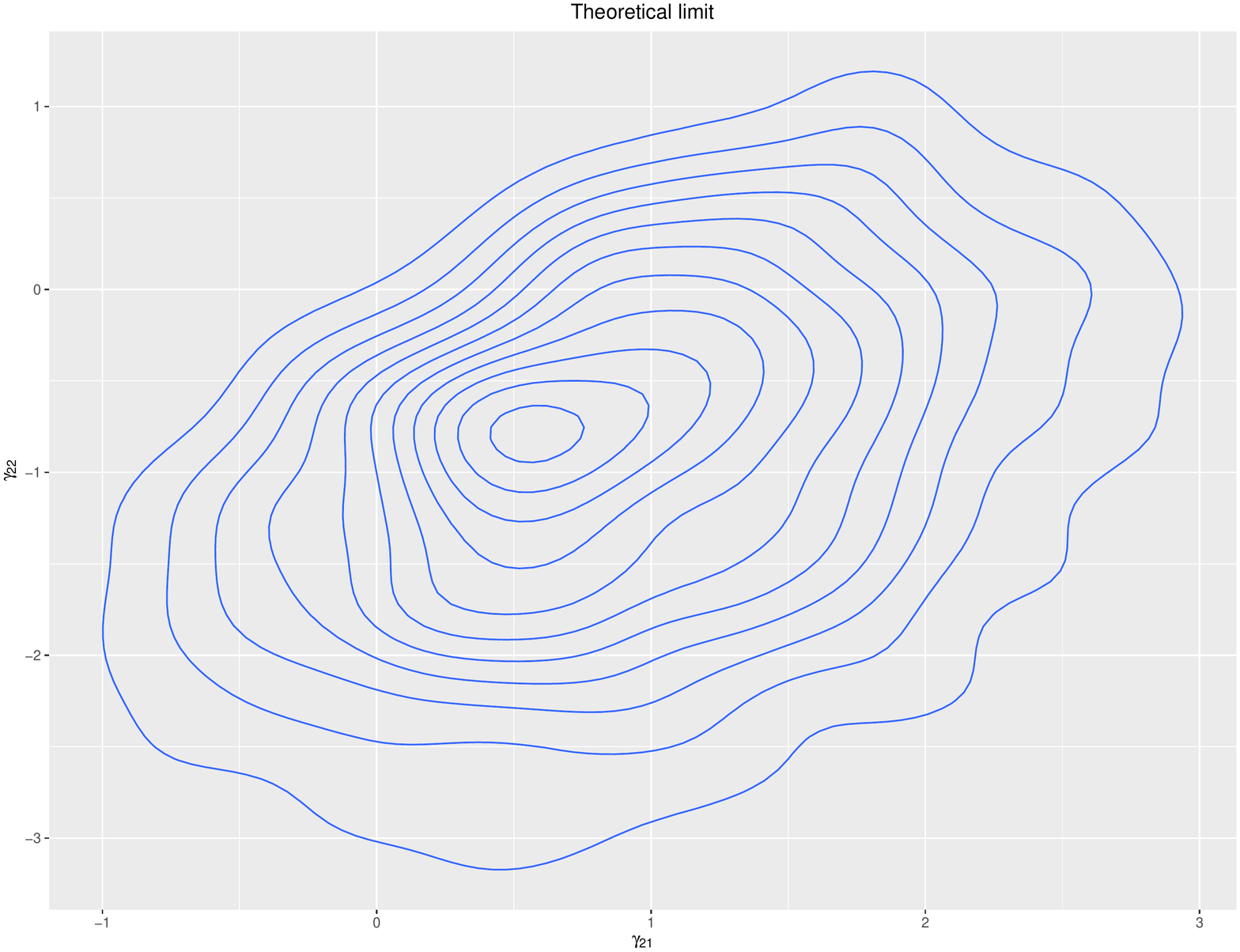}\quad \includegraphics[width = .3\textwidth] {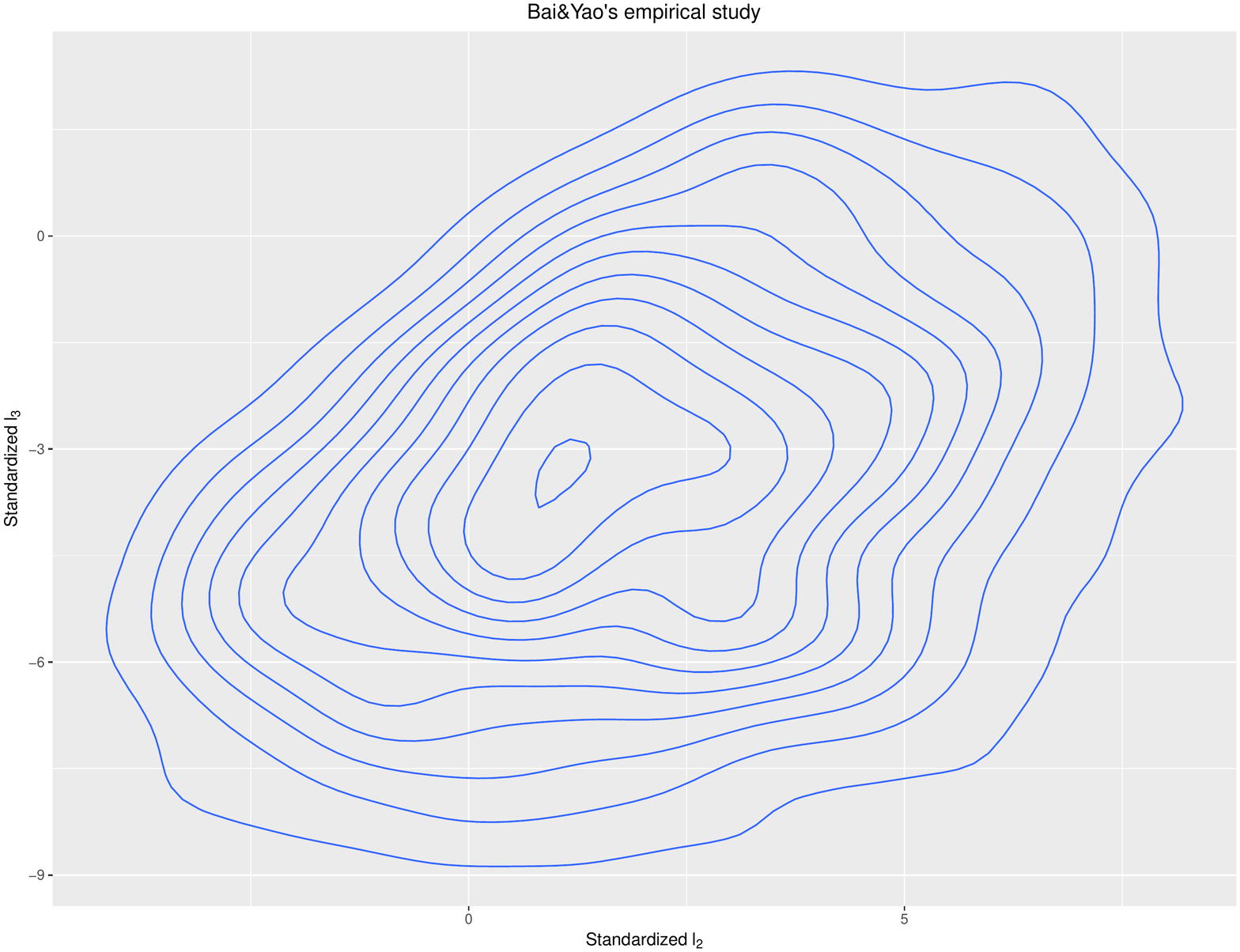}\\
\includegraphics[width = .3\textwidth]{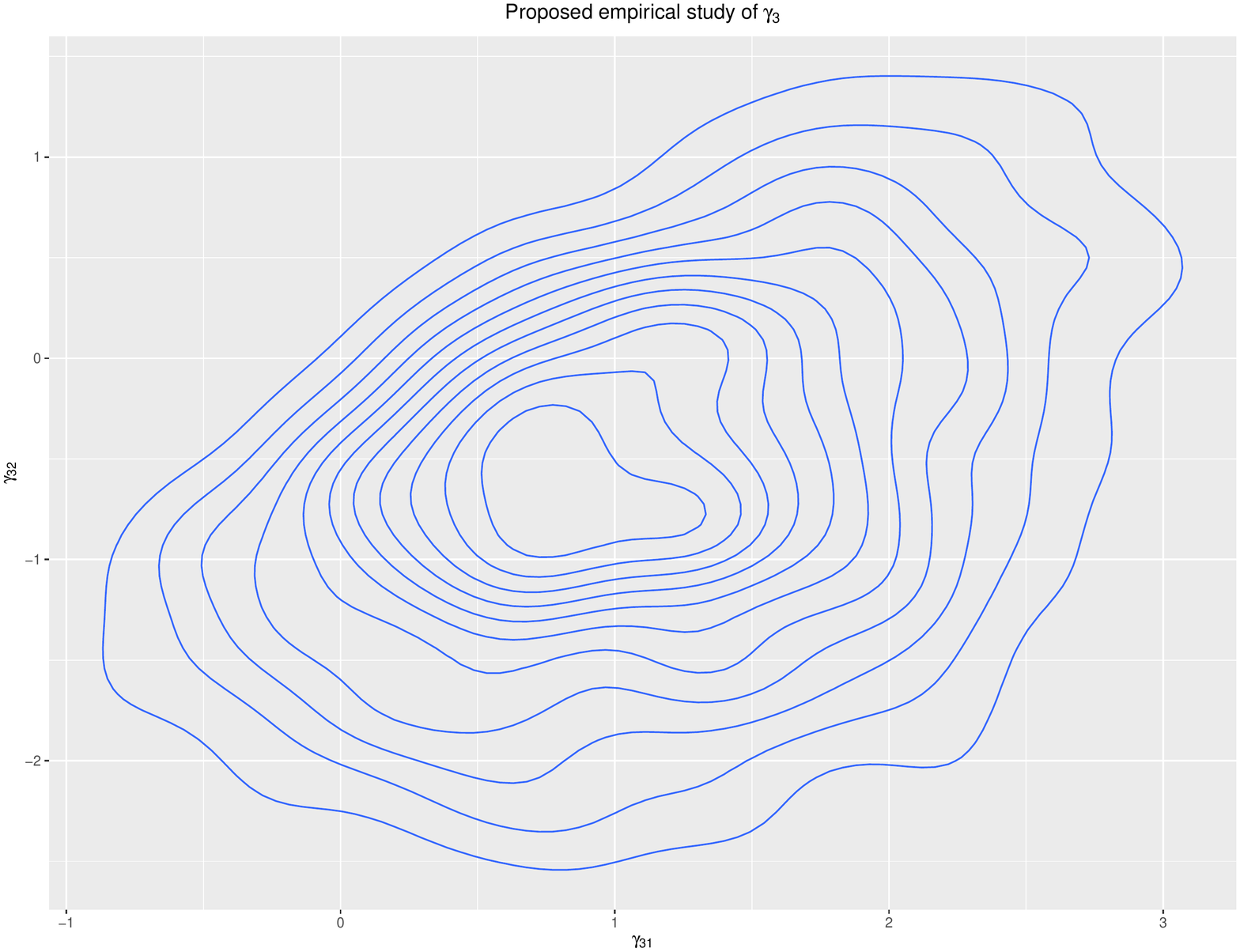}\quad \includegraphics[width = .3\textwidth] {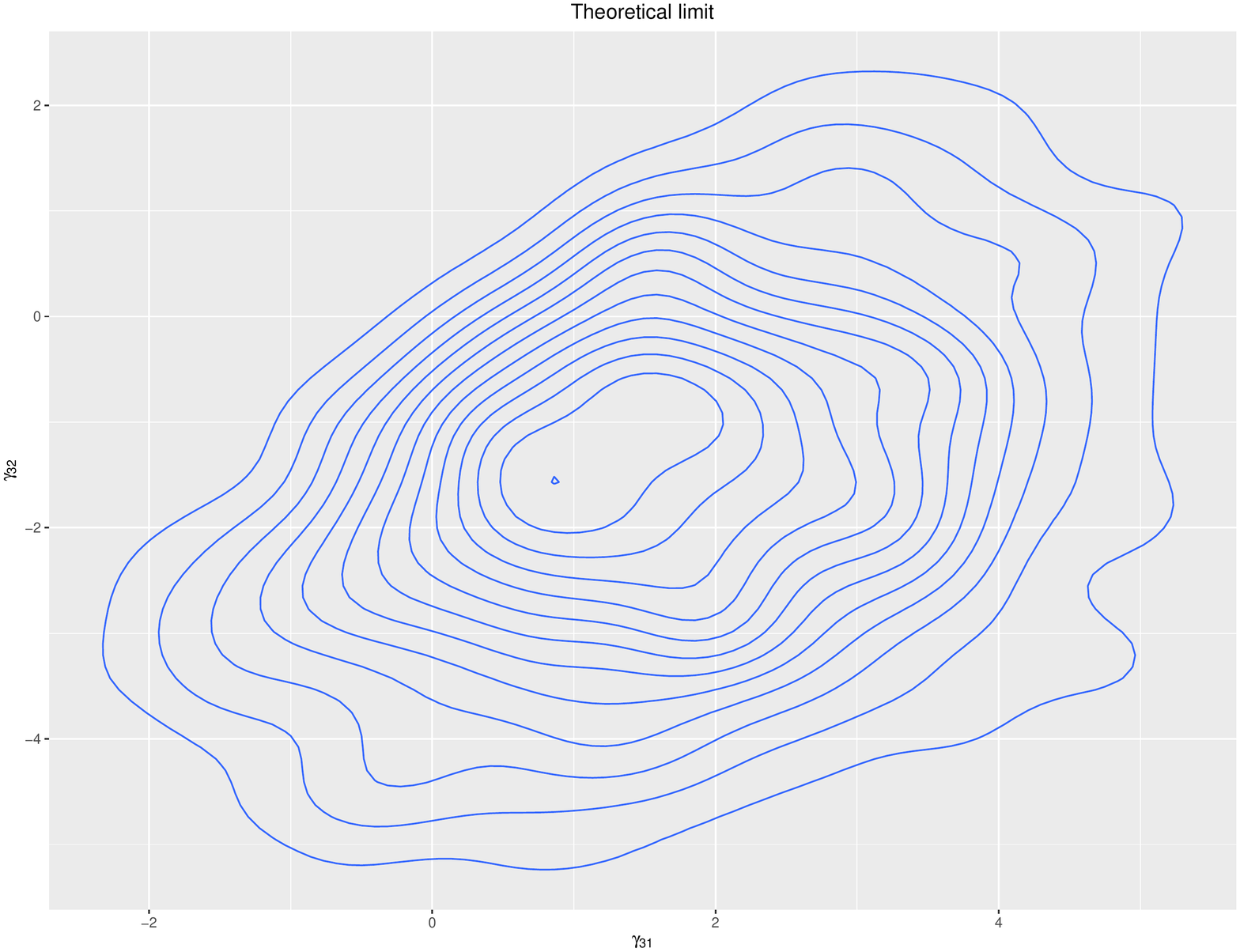}\quad \includegraphics[width = .3\textwidth] {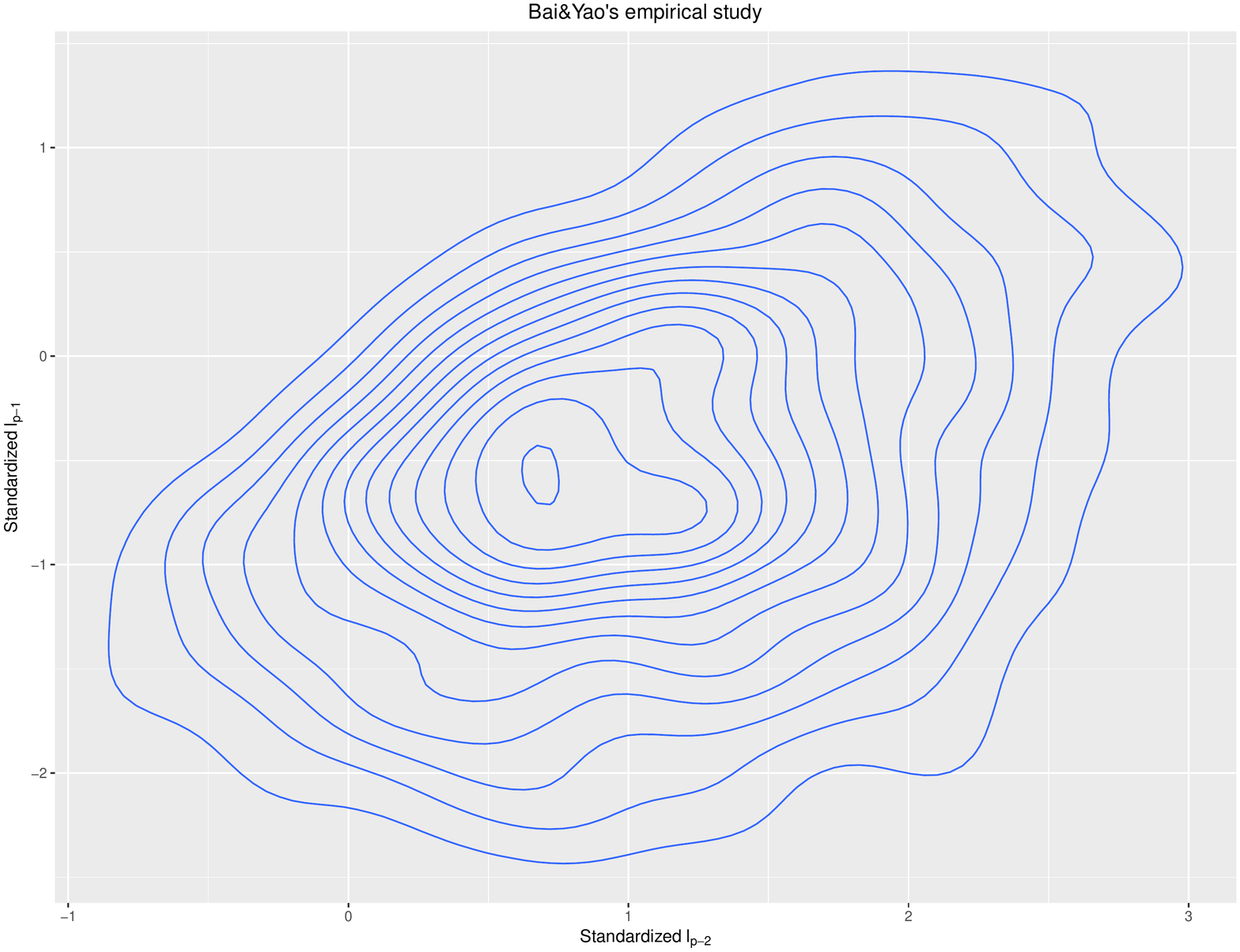}
\caption{ {\rm {\bf Case~II}} under Binomial assumption. 
 }\label{fig:5}
\end{center} 
\end{figure}


\section{Applications and real data analysis}\label{Apply}
\subsection{Application to determine the number of the spikes.}\label{sec5.1}
Since the spiked model is closely related to principal component analysis, it has important applications to the statistical inferences in many scientific fields.  For example,  to reconstruct the original signals in wireless communication,  to rebuild  the observed assets into a low-dimensional set of unobserved variables, which are the factors in economics, and so on. One of the basic but important statistical inferences in these applications is to determine the number of principal components / signals / factors, that is, the number of  spiked eigenvalues.

As mentioned in (\ref{array}), the population eigenvalues are 
\begin{equation}
\beta_{p,1}, \cdots,  \beta_{p,j},\cdots,\beta_{p,p},
\end{equation}
where $\beta_{p,j_k+1}, \cdots \beta_{p, j_k+m_k}=\alpha_k$,  $ k=1,\cdots, K$ are the spikes
 with multiplicity $m_k, k=1,\cdots,K$, and $m_1+\cdots+m_K=M$ is a fixed but unknown number.
 
 We propose to estimate the number of the spikes, $M$, by our result in Theorem~\ref{CLT}. First,
 for every sample eigenvalue $l_j$, $j\in J_k$, it follows from  Theorem~\ref{CLT} that
\[\sqrt{n}\Big(\frac{l_j(S)}{\phi_{n,k}}-1\Big) \bigg/\sigma_k \sim \mathcal{N}(0,1)\]
where $\sigma_k^2=2\theta_k/\kappa_s^2$  under our Assumption~{\bf A} $\sim$ {\bf E} and $\sigma_k^2=(2\theta_k+\beta_x\nu_k)/\kappa_s^2$ under the assumptions of the diagonal or diagonal block independence with the bounded spikes and the 4th moments. 
Then, for every sample eigenvalue $l_j$, we can  calculate  an corresponding interval 
\[C_j=\left[\big(\frac{z_{0.05}\sigma_k}{\sqrt{n}}+1\big)\phi_k, \quad \big(\frac{z_{0.95}\sigma_k}{\sqrt{n}}+1\big)\phi_k\right],\] 
where $z_{0.05}, z_{0.95}$ are the 5\% and 95\% quantiles of the standard normal distribution. 
If $l_j\in C_j$, then it is concluded that the population eigenvalues in according to $l_j$ is a spike; Otherwise, it is not a spike.
Similarly, the same procedures are  conducted  for all the sample eigenvalues, and consequently a sequence of intervals $\{C_j, j=1,\cdots, p\}$ are obtained. Therefore, we propose an estimator for  the number of the spikes, $M$, as follows
\[\hat M_0=\sum\limits_{j=1}^p I_{(l_j \in C_j)}\]
where  $I ( \cdot )$ is 
the indicator function.

However, $\phi_k$ in (\ref{phik}) and $\sigma_k^2$ calculated by Theorem~\ref{CLT} or Remark~\ref{rmk31} cannot  be directly obtained by their expressions in practice, because they are involved with the unknown population spikes $\alpha_k, k=1,\cdots,K$. Therefore, we provide some estimations to get the estimated interval $\hat C_j$, and then 
\[\hat M_0=\sum\limits_{j=1}^p I_{(l_j \in \hat C_j)}\]
which is feasible in practice.

 First, by the first equation in (\ref{eigeneq1}),  it asymptotically holds that
\[l_j+l_j\um(l_j)\alpha_k=0.\]
So we use $-\displaystyle\frac{1}{\um(l_j)}$ to estimate $\alpha_k$, where $\um(\cdot)$ defined in (\ref{um2})  is the Stieltjes transform of the LSD of  the matrix $\frac1n \bX^*\Gamma \bX$. 
Since  the number of spikes is fixed,  the LSD of $\displaystyle\frac1n \bX^*\Gamma \bX$  is approximately the same as the one of the matrix $\displaystyle\frac1n \bX^*UDU^* \bX$, where $D={\rm diag}(D_1,D_2)$. 
Therefore,  we further  define $r_{ij}=|l_i-l_j|/\max(l_i,l_j)$ and adopt 
\[\hat m(l_j)=\frac{1}{p}\sum\limits_{r_{ij}\geq 0.2; i=1}^p(l_i-l_j)^{-1},\]
which 
is a good estimator of $m(l_j)$, and $m(\cdot)$ is the Stieltjes transform of the LSD of  the matrix $\displaystyle\frac1n D_2^{1/2}U_2^*\bX\bX^*U_2 D_2^{1/2}$. 
The setting $\{i\in (1,\cdots, p): r_{ij}\geq 0.2 \}$ is  selected to avoid  the effect of multiple roots, which makes the estimations of the population spikes inaccurate. 
The constant 0.2  is a more suitable threshold value of the ratio based on our simulated results, but for other different populations, the appropriate threshold can be selected by simulation experiments, which is about 10\% to 30\%.
Moreover, by the equation 
\[\um(l_j)=-\displaystyle\frac{1-c}{l_j}+cm(l_j)\]
we obtain the estimator of $\um(l_j)$ as below
\[\hat \um (l_j)=-\displaystyle\frac{1-c}{l_j}+c \hat m(l_j)\]
Finally, we obtain  the estimator of $\alpha_k$, which is expressed as  $\hat\alpha_k=-\displaystyle\frac{1}{\hat \um(l_j)}$.
Without extra efforts, 
the following estimators are automatically  obtained that
\begin{align}
&\hat \phi_k=\phi(\hat \alpha_k);& \label{hatphik}\\
&\hat m(\phi_k)=\frac{1}{p}\sum\limits_{ i=1}^p(l_i-\hat\phi_k)^{-1};\quad &\hat \um(\phi_k)=-\displaystyle\frac{1-c}{\hat\phi_k}+c \hat m(\hat\phi_k);\\
&\hat m_2(\phi_k)=\frac{1}{p}\sum\limits_{i=1}^p(l_i-\hat\phi_k)^{-2};\quad &
\hat \um_2(\phi_k)=\displaystyle\frac{1-c}{\hat\phi_k^2}+c \hat m_2(\hat\phi_k);
\end{align}
So the estimators 
of $\sigma_k, \phi_k$ for the  renewal interval $\hat C_j$ can be expressed by the above estimations.

Through our approach, not only can we estimate the number of the spikes more accurately, but we can also give the estimations of the  population spikes, as well as the limits of the sample spiked eigenvalues. More importantly, we can also provide the specific locations of these  spikes. 

\subsection{Numerical results for Section \ref{sec5.1} } \label{sec5.2}

For the two cases of $\Sigma$ designed in Section~\ref{Sim} with $M=6$, 
we use the method provided in Section~\ref{sec5.1} to estimate the number of the population spikes
 under different population assumptions in Section~\ref{Sim}.
 
 To evaluate the performance of our approach, we shall compare it with some existing methods.
 Since the method in \cite{Onatski2009} provides a better estimator than that in \cite{BaiNg2002}, and 
 \cite{CaiHanPan2017} shows that their approach performs better than that in both of \cite{Onatski2009}  and \cite{Baietal2018}, so we only consider 
  the procedure proposed in \cite{CaiHanPan2017} and the method introduced by \cite{PYao2012},
  which are simply denoted as CHP and PY, respectively. 
  
The following tables report the estimator of the  number of the spikes and its 
corresponding  frequency by three methods. As shown in the tables, our method can give an accurate estimate of the number of the spikes in a large probability, while the other two methods fail to detect the very small spikes because they both assume that the population spikes are the larger eigenvalues, satisfying that  $\alpha_1\geq \alpha_2 \geq \cdots \geq \alpha_M \geq \beta_{1}\geq \cdots \geq \beta_{p-M}$. 
However, it makes sense  to detect all the spiked eigenvalues, including the minimal ones. For example, the original system with all the same eigenvalues has changed after  the input of some signals. If we want to test  which positions in the system have changed,  then it is equivalent to finding out all the spiked eigenvalues.
 In addition, our method has an advantage over other methods,
that is, it also presents the the estimations of the  population spikes, and the specific locations of these  spikes in the tables. 

\begin{table}
\caption{Estimations of the number of the population spikes and its  frequency.}
\begin{center}
   \begin{tabularx}{12cm}{XXXXXXXXX}   
 \multicolumn{9}{c} {~~~~{ \rm \bf{Case I}} under Gaussian Assumption}\\
\hline
\hline
&& &  \multicolumn{5}{c} {Frequency of $\hat M_0$} &\\[0.5mm]
\cline{2-9}
&$\hat M_0$&1&2&3&4&5&{\bf 6}&7\\[0.5mm]
\hline
$p$=200&Ours & 0 &0 &0 &0 &0.024 &{\bf 0.943}  &0.033\\[0.5mm]
$n$=1000&CHP &0  &0 &0 &1 &0      &{\bf 0}        &0\\[0.5mm]
&PY&0.358&0&0.642&0&0&{\bf 0}&0\\[0.5mm]
\hline
$p$=400&Ours & 0 &0 &0 &0 &0.027 &{\bf0.928} &0.045\\[0.5mm]
$n$=1000&CHP &0  &0 &0 &1 &0      &{\bf 0}        &0\\[0.5mm]
&PY&0.371&0&0.629&0&0&{\bf 0}&0\\[0.5mm]
\hline
\hline

  \end{tabularx} 
     \begin{tabularx}{12cm}{XXXXXXXXX}   
 \multicolumn{9}{c} {~~~~{ \rm \bf{Case I}} under Binomial Assumption}\\[0.5mm]
\hline
\hline
&& &  \multicolumn{5}{c} {Frequency of $\hat M_0$} &\\[0.5mm]
\cline{2-9}
&$\hat M_0$&1&2&3&4&5&{\bf 6}&7\\[0.5mm]
\hline
$p$=200&Ours & 0 &0 &0 &0 &0.054 &{\bf 0.943}  &0.003\\[0.5mm]
$n$=1000&CHP &0  &0 &0 &1 &0      &{\bf 0}        &0\\[0.5mm]
&PY&0.622&0&0.378&0&0&{\bf 0}&0\\
\hline
$p$=400&Ours & 0 &0 &0 &0.005 &0.073 &{\bf0.910} &0.012\\[0.5mm]
$n$=1000&CHP &0  &0 &0 &1 &0      &{\bf 0}        &0\\[0.5mm]
&PY&0.640&0&0.360&0&0&{\bf 0}&0\\[0.5mm]
\hline
\hline
  \end{tabularx} 
\end{center}
\label{default}
\end{table}%

%

\begin{table}[htp]
\caption{Estimations of the population spikes and its locations by our method.}
\begin{center}
   \begin{tabularx}{12cm}{XXXXXXX}   
 \multicolumn{7}{c} {~~~~{ \rm \bf{Case I}} under Gaussian Assumption}\\[0.5mm]
\hline\hline\\[-2mm]
$p$=200;~$n$=1000   &  \multicolumn{6}{c} {Estimation of the location of the population spikes} \\[0.5mm]
\cline{3-6}\\[-1mm]
&\multicolumn{6}{c}{(1,~2,~3, ~198,~199,~200)}\\[0.5mm]
&  \multicolumn{6}{c} {Estimation of the population spikes} \\[0.5mm]
\cline{3-6}\\[-1mm]
& $\hat \alpha_1$&$ \hat \alpha_2$& $\hat \alpha_3$& $\hat \alpha_4$& $\hat \alpha_5$& $\hat \alpha_6$\\[0.5mm]
&3.993 &3.207 &3.014 & 0.202 &0.198 &0.098\\[0.5mm]
\hline\\[-2mm]
$p$=400;~$n$=1000   &  \multicolumn{6}{c} {Estimation of the location of the population spikes} \\[0.5mm]
\cline{3-6}\\[-1mm]
&\multicolumn{6}{c}{(1,~2,~3, ~198,~199,~200)}\\[0.5mm]
&  \multicolumn{6}{c} {Estimation of the population spikes} \\[0.5mm]
\cline{3-6}\\[-1mm]
& $\hat \alpha_1$&$ \hat \alpha_2$& $\hat \alpha_3$& $\hat \alpha_4$& $\hat \alpha_5$& $\hat \alpha_6$\\[0.5mm]
&3.930 &3.052 &3.015 & 0.206 &0.186 &0.117\\[0.5mm]
\hline
\hline
  \end{tabularx} 
 ~\\

  \begin{tabularx}{12cm}{XXXXXXX}   
 \multicolumn{7}{c} {~~~~{ \rm \bf{Case I}} under Binomial Assumption}\\[0.5mm]
\hline\hline\\[-2mm]
$p$=200;~$n$=1000   &  \multicolumn{6}{c} {Estimation of the location of the population spikes} \\[0.5mm]
\cline{3-6}\\[-1mm]
&\multicolumn{6}{c}{(1,~2,~3, ~198,~199,~200)}\\[0.5mm]
&  \multicolumn{6}{c} {Estimation of the population spikes} \\[0.5mm]
\cline{3-6}\\[-1mm]
& $\hat \alpha_1$&$ \hat \alpha_2$& $\hat \alpha_3$& $\hat \alpha_4$& $\hat \alpha_5$& $\hat \alpha_6$\\[0.5mm]
&4.025 &3.091 &2.951 & 0.194 &0.185 &0.099\\[0.5mm]
\hline\\[-2mm]
$p$=400;~$n$=1000   &  \multicolumn{6}{c} {Estimation of the location of the population spikes} \\[0.5mm]
\cline{3-6}\\[-1mm]
&\multicolumn{6}{c}{(1,~2,~3, ~198,~199,~200)}\\[0.5mm]
&  \multicolumn{6}{c} {Estimation of the population spikes} \\[0.5mm]
\cline{3-6}\\[-1mm]
& $\hat \alpha_1$&$ \hat \alpha_2$& $\hat \alpha_3$& $\hat \alpha_4$& $\hat \alpha_5$& $\hat \alpha_6$\\[0.5mm]
&4.018 &3.008 &2.876 & 0.207 &0.194 &0.101\\[0.5mm]
\hline
\hline
  \end{tabularx} 
\end{center}
\label{default}
\end{table}%


\begin{table}[htp]
\caption{Estimations of the number of the population spikes and its  frequency.}
\begin{center}
   \begin{tabularx}{12cm}{XXXXXXXXX}   
 \multicolumn{9}{c} {~~~~{ \rm \bf{Case II}} under Gaussian Assumption}\\[0.5mm]
\hline
\hline
&& &  \multicolumn{5}{c} {Frequency of $\hat M_0$} &\\[0.5mm]
\cline{2-9}
&$\hat M_0$&1&2&3&4&5&{\bf 6}&7\\[0.5mm]
\hline
$p$=200&Ours & 0 &0 &0 &0 &0.019 &{\bf 0.950}  &0.031\\[0.5mm]
$n$=1000&CHP &0  &0 &0 &1 &0      &{\bf 0}        &0\\[0.5mm]
&PY&0.375&0&0.625&0&0&{\bf 0}&0\\[0.5mm]
\hline
$p$=400&Ours & 0 &0 &0 &0 &0.025 &{\bf0.927} &0.048\\[0.5mm]
$n$=1000&CHP &0  &0 &0 &1 &0      &{\bf 0}        &0\\[0.5mm]
&PY&0.356&0&0.644&0&0&{\bf 0}&0\\[0.5mm]
\hline
\hline
  \end{tabularx} 
     \begin{tabularx}{12cm}{XXXXXXXXX}   
 \multicolumn{9}{c} {~~~~{ \rm \bf{Case II}} under Binomial Assumption}\\[0.5mm]
\hline
\hline
&& &  \multicolumn{5}{c} {Frequency of $\hat M_0$} &\\[0.5mm]
\cline{2-9}
&$\hat M_0$&1&2&3&4&5&{\bf 6}&7\\[0.5mm]
\hline
$p$=200&Ours & 0 &0 &0 &0 &0.018 &{\bf 0.980}  &0.002\\[0.5mm]
$n$=1000&CHP &0  &0 &0 &1 &0      &{\bf 0}        &0\\[0.5mm]
&PY&0.343&0&0.657&0&0&{\bf 0}&0\\[0.5mm]
\hline
$p$=400&Ours & 0 &0 &0 &0 &0.041 &{\bf0.952} &0.007\\[0.5mm]
$n$=1000&CHP &0  &0 &0 &1 &0      &{\bf 0}        &0\\[0.5mm]
&PY&0.374&0.001&0.625&0&0&{\bf 0}&0\\[0.5mm]
\hline
\hline
  \end{tabularx} 
\end{center}
\label{default}
\end{table}%

\begin{table}[htp]
\caption{Estimations of the population spikes and its locations by our method.}
\begin{center}
   \begin{tabularx}{12cm}{XXXXXXX}   
 \multicolumn{7}{c} {~~~~{ \rm \bf{Case II}} under Gaussian Assumption}\\[0.5mm]
\hline\hline\\[-2mm]
$p$=200;~$n$=1000   &  \multicolumn{6}{c} {Estimation of the location of the population spikes} \\[0.5mm]
\cline{3-6}\\[-1mm]
&\multicolumn{6}{c}{(1,~2,~3, ~198,~199,~200)}\\[0.5mm]
&  \multicolumn{6}{c} {Estimation of the population spikes} \\[0.5mm]
\cline{3-6}\\[-1mm]
& $\hat \alpha_1$&$ \hat \alpha_2$& $\hat \alpha_3$& $\hat \alpha_4$& $\hat \alpha_5$& $\hat \alpha_6$\\[0.5mm]
&4.080 &3.122 &2.909 & 0.208 &0.191 &0.010\\[0.5mm]
\hline\\[-2mm]
$p$=400;~$n$=1000   &  \multicolumn{6}{c} {Estimation of the location of the population spikes} \\[0.5mm]
\cline{3-6}\\[-1mm]
&\multicolumn{6}{c}{(1,~2,~3, ~198,~199,~200)}\\[0.5mm]
&  \multicolumn{6}{c} {Estimation of the population spikes} \\[0.5mm]
\cline{3-6}\\[-1mm]
& $\hat \alpha_1$&$ \hat \alpha_2$& $\hat \alpha_3$& $\hat \alpha_4$& $\hat \alpha_5$& $\hat \alpha_6$\\[0.5mm]
&3.949 &3.217&2.887 & 0.231 &0.188 &0.102\\[0.5mm]
\hline
\hline
  \end{tabularx} 
 ~\\

  \begin{tabularx}{12cm}{XXXXXXX}   
 \multicolumn{7}{c} {~~~~{ \rm \bf{Case II}} under Binomial Assumption}\\[0.5mm]
\hline\hline\\[-2mm]
$p$=200;~$n$=1000   &  \multicolumn{6}{c} {Estimation of the location of the population spikes} \\[0.5mm]
\cline{3-6}\\[-1mm]
&\multicolumn{6}{c}{(1,~2,~3, ~198,~199,~200)}\\[0.5mm]
&  \multicolumn{6}{c} {Estimation of the population spikes} \\[0.5mm]
\cline{3-6}\\[-1mm]
& $\hat \alpha_1$&$ \hat \alpha_2$& $\hat \alpha_3$& $\hat \alpha_4$& $\hat \alpha_5$& $\hat \alpha_6$\\[0.5mm]
&3.838 &3.275 &2.922 & 0.216 &0.192 &0.098\\[0.5mm]
\hline\\[-2mm]
$p$=400;~$n$=1000   &  \multicolumn{6}{c} {Estimation of the location of the population spikes} \\[0.5mm]
\cline{3-6}\\[-1mm]
&\multicolumn{6}{c}{(1,~2,~3, ~198,~199,~200)}\\[0.5mm]
&  \multicolumn{6}{c} {Estimation of the population spikes} \\[0.5mm]
\cline{3-6}\\[-1mm]
& $\hat \alpha_1$&$ \hat \alpha_2$& $\hat \alpha_3$& $\hat \alpha_4$& $\hat \alpha_5$& $\hat \alpha_6$\\[0.5mm]
&4.123 &3.211 &3.001 & 0.216 &0.195 &0.096\\[0.5mm]
\hline
\hline
  \end{tabularx} 
\end{center}
\label{default}
\end{table}%

\subsection{Real data analysis}

Now we apply the procedure of determining the number of the spikes  proposed in Section \ref{sec5.1} to the actual data titled as "Early stage of Indians Chronic Kidney Disease(CKD)"\footnote{ The data is downloaded from 
\url{https://archive.ics.uci.edu/ml/datasets/Chronic_Kidney_Disease.}}.

The data came from records collected by a hospital in India over a period of about 2 months, which 
consists of 400 observations and 25 variables. The first 24 variables $X_1\cdots,X_{24}$ are independent variables,  which rerecord the various laboratory indicators and hospital records, including 
age,	blood pressure (bp), specific gravity (sg), albumin (al), sugar (su), red blood cells (rbc), pus cell (pc), pus cell clumps (pcc),  bacteria (ba),  blood glucose random (bgr), blood urea (bu), serum creatinine (sc), sodium (sod), potassium (pot), hemoglobin (hemo), packed cell volume (pcv), white blood cell count (wc), 
red blood cell count (rc), hypertension (htn), diabetes mellitus (dm), coronary artery disease (cad), appetite (appet), pedal edema (pe), anemia (ane).
The 25th variable is the dependent variable to indicate whether the patient has chronic kidney disease(ckd).

We apply our method to  determine the number of the spikes of the covariance matrix $\Sigma_0$ generated from the 
standardized data of the first 24 variables with 114 observations (For simplicity, we have only chosen 114 observations without missing values). Then, we obtain the following results  in the Table~\ref{real}.

\begin{table}[htp]
\label{real}
\caption{Estimations of the number, the sizes and the location of the population spikes by using the real data.}
\begin{center}
   \begin{tabularx}{12cm}{XXXXXXXXXX}   
\hline\hline
Number: &  \multicolumn{9}{c} {9} \\[1mm]
Location: &\multicolumn{9}{c}{(1,~2, ~18,~19,~20,~21, ~22, ~23, ~24)}\\[1mm]
Sizes: & ~~$\hat \alpha_1$&$ \hat \alpha_2$& $\hat \alpha_{3}$& $\hat \alpha_4$& $\hat \alpha_5$& $\hat \alpha_6$& $\hat \alpha_7$&$ \hat \alpha_8$& $\hat \alpha_{9}$\\[1mm]
&10.818 & 2.143 & 0.219 & 0.166 & 0.124 & 0.101 & 0.064  &  0.048 &0.009\\[1mm]
\hline
\hline
  \end{tabularx} 

\end{center}
\label{default}
\end{table}%

 As seen from the Table~\ref{real}, 
if we define the singular value decomposition of $\Sigma_0$  as $\Sigma_0=U\Lambda_0 U'$, and ${\bf u}_{i}$ is the $i$th
 column of the orthogonal matrix $U$, then
the factors generated from  independent variables $X=(X_1\cdots,X_{24})' $ can be roughly divided into three groups: one group has a greater impact with larger spiked eigenvalues, like ${\bf u}'_{1}X, {\bf u}'_{2}X$; Another group of much weaker effects, like ${\bf u}'_{i}X$, $i=18,\cdots, 24$; The last  group  that may have most of the same effects, like ${\bf u}'_{i}X$, $i=3,\cdots, 17$.
Furthermore, if we use the data with the missing values made up, the experimental results may be more accurate. To make up for missing values, one can use the missForest function in the package missForest.


\section{Conclusion}\label{Con}

In this paper, we propose a G4MT for a generalized spiked covariance matrix, which shows the universality of the asymptotic law for its spiked  eigenvalues. Through the concrete example of the CLT of normalized spiked eigenvalues, we illustrate the basic idea and procedures of the G4MT to show the universality of a limiting result related to the large dimensional random matrices. Unlike \cite{TaoVu2015}, we avoid the estimates of high-order partial derivatives of an implicit function to the entries of the random matrix, and thus, the strong condition $C_0$ of sub-exponential property is avoided.  Moreover, the  required 4th moment  condition is reduced to a tail probability in Assumption~$\bB$,  which is necessary for the existence of the largest eigenvalue limit.
Without the constraint of the existence of the  4th moment,  we only need a more regular and minor condition (\ref{CondU1}) on the elements of $U_1$. On the one hand, our result has much wider applications than \cite{BaiYao2008, BaiYao2012};  on the other hand, the result of 
Bai and Yao (2012) shows the necessity of the condition (\ref{CondU1}).

\section*{Acknowledgments}


%

\newpage


\begin{center}
{\bf Supplement to "Generalized  Four Moment  Theorem and an Application to  CLT  for Spiked Eigenvalues of high-dimensional Covariance Matrices".}
\end{center}
\renewcommand\thesection{\Alph{section}} 
\setcounter{section}{0}
 \section{The detailed explanation of Assumption~{\bf D}}\label{app0}
 
 For the Assumption~{\bf D}, we have 
\begin{remark}
 In the proof of the main theorems, this assumption is actually used as
 \[ \max\limits_{t,s} |u_{ts}|^2 \big(\rE|x_{11}|^4-3\big)I(|x_{11}|<\eta_n\sqrt{n}) \rightarrow 0, \]
where $\eta_n\to 0$ with a slow rate. In fact, because of Assumption $\bB$, the condition (\ref{CondU1}) remains the same as 
\[ \max\limits_{t,s} |u_{ts}|^2\big(\rE|x_{11}|^4-3\big)I(|x_{11}|<\eta\sqrt{n}) \rightarrow 0,\]
provided $|\log \eta|<1/4 \log n$.
 \end{remark}

\begin{remark}
 If the 4th moment of population random variable  $X$ is bounded,  only the condition 
 \[ \max\limits_{t,s} |u_{ts}|^2  \rightarrow 0 \]
  is needed; if the  4th moment does not exist, we only need  
\[\max\limits_{t,s} |u_{ts}|^2 =O({\log^{-1} n})\]
at most, since the 4th moment of the truncated variables is $o( \log n)$ by Lemma~\ref{lemmoments}.
\end{remark}
For example, assume that the random variable $X$  follows the population distribution with the density 
\[d(x)=\frac{a_0}{ \big(|x|+1\big) ^5 \log{\big(|x|+2\big)}},\]
where $a_0$ is a scaling number;
then, 
\begin{align*}
{\rm P}\left(|X| >x\right)&=O\left(|x|^{-4}  \log^{-1}{\big(|x|+2\big)}\right),
 \end{align*}
 which implies 
 \[\rE|X|^4I(|X|> \eta_n \sqrt{n})=O\big(\log(\log n)\big)\]
Then, the condition (\ref{CondU1}) 
 can be reduced to a weaker one, {\rm i.e.}
\[\max\limits_{t,s} |u_{ts}|^2=o\big(1/\log (\log n)\big).\]

 \section{Proof of the Truncation and Centralization}  \label{app1}
  
   By the Assumption~$\bB$,  
let  $\tau=\eta\sqrt{n} \rightarrow \infty$;  for every fixed $\eta>0$,  we obtain that 
\[\eta^4n^2 {\rm P} \left(|x_{ij}|>\eta\sqrt{n}\right) \rightarrow 0. \]
The limiting behavior  still performs well  by removing the fixed $\eta^4$, that is, 
 \be
 n^2 {\rm P} \left(|x_{ij}|>\eta\sqrt{n}\right) \rightarrow 0.\label{2m4}
 \ee
 Because of the arbitrariness of $\eta$ in (\ref{2m4}),  it is proved by Lemma~15 in \cite{LiBaiHu2016}
  that there exist a sequence of positive  numbers $\eta=\eta_n\to 0$ such that 
  \be
  n^2 {\rm P} \left(|x_{ij}|>\eta_n\sqrt{n}\right) \rightarrow 0.
  \label{n2p}
  \ee
  The convergence rate of the constants $\eta_n$ can be selected arbitrarily slowly, and hence, we may assume  that 
$\eta_nn^{1/5} \rightarrow \infty$.

  Then, consider the truncated samples  $\hat x_{ij}= x_{ij} {\rm I}(|x_{ij}| < \eta_n \sqrt{n})$, set 
 \[
  \br_j=\frac{1}{\sqrt{n}} T_p\bx_j,  \quad \hat \br_j= \frac{1}{\sqrt{n}} T_p \hat\bx_j
 \]
  where $T_pT_p^*=\Sigma$. 
Consequently, the generalized spiked sample covariance $S$  is expressed as 
  \[S=T_p
\left(\frac{1}{n}\bX\bX^*\right)T_p^*= \sum \limits_{j=1} ^{n}\br_j\br_j^*:= \bR_n\bR_n^*.\]
Define the matrix with truncated entries as 
    \[\hat S= \sum \limits_{j=1} ^{n} \hat \br_j\hat \br_j^*:= \hat \bR_n\hat \bR_n^*.\]
     Therefore, according to the property  (\ref{n2p}), we have 
     \begin{eqnarray*}
    && {\rm P}\left( S \neq  \hat S\right)=    {\rm P}\left( \sum \limits_{j=1} ^{n}\br_j\br_j^* \neq  \sum \limits_{j=1} ^{n} \hat \br_j\hat \br_j^*\right)\\
    && \leq \sum\limits_{i,j} {\rm P}\left(|x_{ij}| \geq \eta_n \sqrt{n} \right)=np {\rm P}\left(|x_{11}| \geq \eta_n \sqrt{n} \right)\rightarrow 0,  \quad \text{as } 
  n, p \rightarrow \infty.
     \end{eqnarray*}
    
 Next,  define the truncated and centralized sample covariance matrix as   \[\tilde S= \sum \limits_{j=1} ^{n} \tilde \br_j\tilde \br_j^*:= \tilde \bR_n\tilde \bR_n^*,\]
where $ \tilde\br_j= T_p \tilde\bx_j /\sqrt{n}$ and $ \tilde x_{ij}=(\hat x_{ij} - \rE \hat x_{ij})/ \sigma_{n}$  with $\sigma_{n}^2=\rE \left|\hat x_{ij}- \rE  \hat x_{ij} \right|^2$.
Then,  by Theorem A.46 of \cite{BaiSilverstein2010}, we have 
\begin{eqnarray*}
&&\max \left\{ l_j^{1/2}(\tilde S) -l_j^{1/2}(\hat S)\right\} \leq  \Vert \tilde \bR_n -\hat \bR_n\Vert 
= \Vert \rE \hat \bR_n \Vert+\Vert (1-\sigma_n)\tilde \bR_n\Vert\\
& \leq& n^{-1/2}|\rE \hat x_{11}|\min(p,n)+\frac{1-\sigma_n^2}{1+\sigma_n}\Vert \tilde \bR_n\Vert\\
&=&  o_{a.s}(n^{-1}),
\end{eqnarray*}
where we have used the fact that there  exists a finite constant $C_0$ such that $|\rE \hat x_{11}|\le C_0 \int^{+\infty}_{\eta_n\sqrt{n}}P(|x_{11}|\ge x) \md x=o(n^{-3/2})$, 
and that $1-\sigma_n^2=\rE x_{11}^2I(|x_{11}|\ge \eta_n\sqrt{n})+\rE x^2_{11}I(|x_{11}|<\eta_n\sqrt{n})=o(n^{-1})$ in Lemma~C.1 and $\Vert \hat \bR_n\Vert \le \Vert T_p\Vert (1+\sqrt{c}+\varepsilon), a.s.$.

Thus, it is concluded that the procedure of centralization does not have an effect on the limiting distribution of the spiked eigenvalues
 because of 
\[\sqrt{n}\left\{ \frac{l_j(\tilde S) }{\phi_{n,k} }- 1-\left(\frac{l_j(\hat S)}{\phi_{n,k}}-1\right)  \right\} \leq o(\phi_{n,k}^{-1}n^{-{\frac{1}{2}}}).\]

\section{Lemmas} \label{app3}
Some useful lemmas are provided in this section, which are needed to prove the Theorem~\ref{thm2}.
First, we investigate the arbitrary moments of  $\ \tilde x_{ij}$ and depict  their convergence rates  in the  following lemma.
\begin{lemma}\label{lemmoments}
For the entries $\{\tilde x_{ij}\}$ truncated at $\eta_n \sqrt{n}$, centralized and renormalized, it follows that 
\begin{eqnarray*}
|\rE \tilde x_{ij}^\alpha| \leq o(\sqrt{n})^{\alpha-4}, ~ \alpha>4;\\
|\rE \tilde x_{ij}^\alpha| = o(\log n), ~ \alpha=4;\\
|\rE \tilde x_{ij}^\alpha| \leq o(\sqrt{n})^{\alpha-4}, ~ \alpha<4;
\end{eqnarray*}
\end{lemma}
\begin{proof} We only estimate the inequalities above with $\tilde x_{ij}$ replaced by $\hat x_{ij}$ because the centralization only involves the third estimate with $\alpha=1$. 
For any integer $\alpha<4$, we have
  \begin{eqnarray*}
|\rE \hat x_{ij}^\alpha|
&=&\int_{\eta_n\sqrt{n}}^\infty \alpha x^{\alpha-1}P(|x_{ij}|>x)dx\nonumber\\
&\leq& \int_{\eta_n\sqrt{n}}^\infty o(x^{\alpha-5})dx=o((\eta_n\sqrt{n})^{\alpha-4}).
 \end{eqnarray*}
 Therefore, $|\rE \hat x_{ij}^\alpha| \leq o(\sqrt{n})^{\alpha-4}$,  if $\alpha<4$.
 
 For the case of  $\alpha=4$,  we have 
  \bqn
 |\rE \hat\bx_i^4|&=&
  \int_{0}^{\eta_n\sqrt{n}} x^4\md P(|x_{ij}| \leq x)
 \\&=&-\int_{0}^{\eta_n\sqrt{n}} x^4\md P(|x_{ij}| > x)
\le \int_{0}^{\eta_n\sqrt{n}} 4x^3P(|x_{ij}|>x)dx\\
&\leq &O(1)+\int_{K}^{\eta_n\sqrt{n}} 4x^3o(x^{-4}) \md x=o(\log n).
 \eqn

 For any integer $\alpha>4$,
  \bqn
 \rE |\hat\bx_i^\alpha|&=&
 \int_{0}^{\eta_n\sqrt{n}} x^\alpha\md P_x(|x_{ij}| \leq x)
 \\&\le &\alpha \int_{0}^{\eta_n\sqrt{n}} x^{\alpha-1} P(|x_{ij}|>x)\md x=o(\sqrt{n}^{\alpha-4})
 \eqn
\end{proof}

 Second, before proceeding with the proof of the G4MT numbered as Theorem~\ref{thm2},
 we  begin with some preliminary lemmas used during the process. 

\begin{lemma}\label{lembetak}
 Let ${\bX}=(\bx_1,\cdots,\bx_n)$ and ${\bf Y}=(\by_1,\cdots,\by_n)$ be two independent random matrices satisfying Assumptions $\bA\sim \bE$, and set ${\bf X}_k=(\bx_1,\cdots,\bx_k,$ $\by_{k+1},\cdots,\by_n)$ with convention 
${\bf X}={\bf X}_n$ and ${\bf Y}={\bf X}_0$.
 Denote ${\bf X}_{k0}=(\bx_1\cdots,\!\bx_{k\!-\!1},$ $\by_{k+1}\cdots \by_n)$,
\[\beta_k=1-n^{-1}\bx_k^*\Gamma^{1/2}(\lambda \bI_p-n^{-1}\Gamma^{1/2}{\bf X}_{k0}{\bf X}_{k0}^*\Gamma^{1/2})^{-1}\Gamma^{1/2}\bx_k.\]
and
\[
 \beta_{k0}=1-n^{-1}{\rm tr}\Big(\Gamma (\lambda \bI_p-n^{-1}\Gamma^{1/2}{\bf X}_{k0}{\bf X}_{k0}^*\Gamma^{1/2})^{-1}\Big).
\]
with $\Gamma=U_2 D_2U^*_2$ defined in (\ref{OmegaM}). 
Then, it is  obtained that 
\begin{eqnarray*}
\beta_{k0}&\to &-\frac1{\lambda\underline m(\lambda)}\ne 0\\
\varepsilon_k&=& \beta_k-\beta_{k0} \to 0 \label{eklimit}\\
\rE_k\varepsilon^2_k&\le&o(n^{-1}\log n)\\
\rE_k\varepsilon_k^4&=&o(n^{-1}),
\end{eqnarray*}
where $\um(\lambda)$ is the Stieltjes transform of the LSD of the matrix $\displaystyle\frac{1}{n}{\bf X}_{k0}^*\Gamma {\bf X}_{k0}$,
$\rE_0(\cdot)$ denotes the expectation and $\rE_k (\cdot)$ denotes the conditional expectation with respect to the $\sigma$-field generated by the vectors {$\bx_1,\cdots, \bx_k$}.   
\end{lemma}
\begin{proof}
 Denote    $m(\lambda,\bX_{k0})$ and  $\um(\lambda, \bX_{k0})$ as the Stieltjes transforms of the LSDs of the matrix $n^{-1}\Gamma^{1/2}{\bf X}_{k0}{\bf X}_{k0}^*\Gamma^{1/2}$ and $\displaystyle\frac{1}{n}{\bf X}_{k0}^*\Gamma {\bf X}_{k0}$, respectively.
   If no confusion, we still use the notations $m(\lambda)$ and  $\um(\lambda)$ for simplicity.

By  (3.11), (3.13) and  the limitation above (3.14) in \cite{BaiZhou2008},  we have  
\[\frac1{n}{\rm tr}\Big(\Gamma (\lambda \bI_p-n^{-1}\Gamma^{1/2}{\bf X}_{k0}{\bf X}_{k0}^*\Gamma^{1/2})^{-1}\Big)  \sim 1-\frac{1}{1-c-c\lambda m(\lambda)},\]
and it follows that 
\[
 \beta_{k0}=1-n^{-1}{\rm tr}\Big(\Gamma (\lambda \bI_p-n^{-1}\Gamma^{1/2}{\bf X}_{k0}{\bf X}_{k0}^*\Gamma^{1/2})^{-1}\Big) \sim 
-\frac1{\lambda\underline m(\lambda)}\ne 0.
\]
due to the relationship
 $\um(\lambda)=-\displaystyle\frac{1-c}{\lambda}+cm(\lambda)$.

The second conclusion $\varepsilon_k= \beta_k-\beta_{k0} \to 0$ is an easy consequence of the third,  $\rE_k|\varepsilon_k|^2\to0$.

By the formula (1.15) of \cite{BaiSilverstein2004}, we have 
\bqa
&&\rE_k|\varepsilon_k|^2=\frac{1}{n^2} \rE_k\bigg|x_k^*\Gamma^{1/2}(\lambda I_p-n^{-1}\Gamma^{1/2}{\bf X}_{k0}{\bf X}_{k0}^*\Gamma^{1/2})^{-1}\Gamma^{1/2}x_k\non
&&\quad\quad\quad\quad\quad-{\rm tr}\Big(\Gamma (\lambda \bI_p-n^{-1}\Gamma^{1/2}{\bf X}_{k0}{\bf X}_{k0}^*\Gamma^{1/2})^{-1}\Big)\bigg|^2\non
&&=\!\frac{1}{n^2}\!\left\{ \!\rE_k{\rm tr}B^2 +\rE_k{\rm tr}BB^T|\rE x_{11}^2|^2  + \sum\limits_{i}\rE_k |b_{ii}|^2(\rE |x_{ik}^4|-2-|\rE x_{11}^2|^2 )\right\}\non
&&\quad\quad\le \frac{2}{n^2} \rE_k \tr \left[B^2\right]+\frac{o(\log n)}{n^2} \sum\limits_{i}\rE_k |b_{ii}|^2 \to 0, \label{epk2}
\eqa
where $B=\left(b_{ij}\right):=\Gamma^{1/2} (\lambda \bI_p-n^{-1}\Gamma^{1/2}{\bf X}_{k0}{\bf X}_{k0}^*\Gamma^{1/2})^{-1}\Gamma^{1/2}$, and the eigenvalues of $\Gamma$ are non-spiked eigenvalues and bounded. Then, the third conclusion is proved.

Furthermore, for the conditional moments of $\varepsilon_k$, we have
\[\rE_k\varepsilon^2_k= o(n^{-1}\log n),\] 
which is automatically obtained by equation (\ref{epk2}).

Finally, 
\begin{align*}
&\rE_k\varepsilon_k^4=\frac{1}{n^4} \rE_k\bigg\{x_k^*\Gamma^{1/2}(\lambda I_p-n^{-1}\Gamma^{1/2}{\bf X}_{k0}{\bf X}_{k0}^*\Gamma^{1/2})^{-1}\Gamma^{1/2}x_k\non
&-{\rm tr}\Big(\Gamma (\lambda \bI_p-n^{-1}\Gamma^{1/2}{\bf X}_{k0}{\bf X}_{k0}^*\Gamma^{1/2})^{-1}\Big)\bigg\}^4\non
&= \!\frac{1}{n^4}\Bigg\{\big(\rE|x_{11}|^8-4\rE|x_{11}|^6+6\rE|x_{11}|^4-3\big) \sum\limits_{i}  b_{ii}^4\non
&
+\big(4\rE|x_{11}|^6-8\rE|x_{11}|^4+4\big)\sum\limits_{i\neq j} b_{ii}^2b_{ij}^2 +\!\big(\rE|x_{11}|^4 )^2-2\rE|x_{11}|^4+1\big) \sum\limits_{i\neq j} b_{ii}^2b_{jj}^2\non
&+16(\rE|x_{11}|^4 )^2\sum\limits_{i\neq j} b_{ij}^4
+\big(\rE|x_{11}|^4-1\big) \sum\limits_{i\neq i' \neq j'} b_{ii}^2b_{i'j}^2  \non
&+16\rE|x_{11}|^4\!\big(\!\!\sum\limits_{i\neq j \neq j'}\!\!\!b_{ij}^2b_{ij'}^2\!+\!\!\sum\limits_{i\neq i' \neq j'}\!\!\!b_{ij}^2b_{i'j}^2\big)\!\!+\!\!\sum\limits_{i\neq j \neq i' \neq j'}\!\! \!\!\Big(16(b_{ij}^2 +b_{i'j'}^2 )\!+\! b_{ij}b_{i'j}b_{i'j}b_{i'j'}\Big) \!\!\Bigg\}\non
&=K_0\big(o(n^{-1})+ o(n^{-2} \log n  )+o(n^{-3} \log ^2n ) \big)
=o(n^{-1}), \label{epk3}
\end{align*}
where {$K_0$} is an absolute constant and may take different values at different appearances.
\end{proof}

\begin{lemma}\label{lembb1}
Let $\bu_{1i}$, $i=1,2,\cdots,N$  be a column vector of $U_1$, $\bu_{2i}$ be 
 a unit $p$-vector orthogonal to $U_1$, and $\Gamma$ be  an nnd $p\times p$ matrix of bounded spectral norm, where $N=O(n)$. Then, there is a constant $\delta\in(0,1/4)$ such that 
\be\label{eqbj1}
\max_{i\le N}n^{-1}\lambda|\bu_{1i}^*\bX_n(\lambda \bI_n-n^{-1}\bX_n^*\Gamma \bX_n)^{-1}\bX^*_n\bu_{2i}|\le 2 n^{-\delta}, a.s.
\ee 
\end{lemma}
 
 \begin{proof}
 Note that 
\bqn
R_i&=& n^{-1}\lambda \bu_{1i}^*\bX_n(\lambda \bI_n-n^{-1}\bX_n^*\Gamma \bX_n)^{-1}\bX^*_n\bu_{2i}\non
&=&\frac1n\sum_{k=1}^n \frac{\bu_{1i}^*x_kx_k^*\bu_{2i}}{\beta_k}+\frac1n\sum_{k_1\ne k_2 } \frac{\bu_{1i}^*x_{k_1}x_{k_2}^*\bu_{2i} \beta_{k_1k_2}}{\beta_{k_1;k_2}\beta_{k_2;k_1}-\beta_{k_1k_2}^2}:=J_{1i}+J_{2i}
\eqn
where 
\bqn
\beta_{k_1k_2}&=&n^{-1}x_{k_1}^*U_2D_2^{1/2}G_{k_1,k_2}^{-1}D_2^{1/2}U_2^*x_{k_2}\\
 \beta_{k_1;k_2}&=& {\color{blue}  1-n^{-1}x_{k_1}^*U_2D_2^{1/2}G^{-1}_{\color{red}k_1, k_2}D_2^{1/2}U_2^*x_{\color{red}k_1}}\\
G_{k_1,k_2}&=&\lambda\bI_p-n^{-1}D_2^{1/2}U_2\bX_{k_1,k_2}\bX_{k_1,k_2}U_2^*D_2^{1/2}
\eqn
where $\bX_{k_1,k_2}$ is identical to $\bX$ excluding the $k_1$th and $k_2$th columns.

Define notations 
\bqn
{\color{blue} \bar\beta_{\color{red}k_1; k_2}}&=&{\color{blue}  1-n^{-1}\rE{\rm tr}G^{-1}_{\color{red}k_1, k_2}D_2^{1/2}U_2^*U_2D_2^{1/2},}\\
{\color{blue} 
\varepsilon_{k_1;k_2}}&=&{\color{blue} n^{-1}x_{k_1}^*U_2D_2^{1/2}G_{\color{red}k_1,k_2}^{-1}D_2^{1/2}U_2^*x_{k_1}-n^{-1}\rE{\rm tr}G^{-1}_{\color{red}k_1,k_2}D_2^{1/2}U_2^*U_2D_2^{1/2},}
\eqn
and events
\bqn
{\cal E}_{1k}&=&\{|\varepsilon_k|>\delta_k/2\},\ \ \ {\cal E}_{2k}=\{|n^{-1/2}\bX_{k0}|>b_+\}\\
{\cal E}_{k_1,k_2}&=&\{|\beta_{k_1k_2}|>\delta_k/2\},\\
{\cal E}&=&\bigcup_{k=1}^n{\cal E}_k\bigcup_{k_1\ne k_2}{{\cal E}_{k_1, k_2}}\{\|\bX/\sqrt{n}\|>b_+\}\bigcup\{\|G_{k_1,k_2}^{-1}\|>B^+\},
\eqn
where $\delta_k=\rE \beta_k$, $b_+>(1+\sqrt{c})^2$ and $B^+$ are large constants.

Referring to the selection of $\lambda$ and checking the proofs of \cite{BaiSilverstein1998}, 
 one may find that 
\bqn
P(\| G_{k_1,k_2}^{-1}\|>B^+)&=&o(n^{-t}),\\
P(\|\bX/\sqrt{n}\|>b^+)&=&o(n^{-t}),
\eqn
for any given $t>0$.

Employing Lemma 9.1 of \cite{BaiSilverstein2010}, one can prove that 
\bqn
P({\cal E}_{1k})&=&o(n^{-t}),\\ 
P({\cal E}_{2k})&=&o(n^{-t}),\\
 P({\cal E}_{k_1,k_2})&=&o(n^{-t}).
\eqn
Hence, we have 
\be\label{eqbf06}
P({\cal E})=o(n^{-t}).
\ee

We next prove that
\be\label{eqbj2}
\max_{i}\rE R_iI_{{\cal E}^c}\le n^{-\delta},
\ee
for some constant $\delta\in(0,1/4)$.

In fact, 
\begin{align}
\label{eqbj3}
&|\rE J_{1i}I_{{\cal E}^c}|= \left|\frac1n\sum_{k=1}^n\rE \frac{\bu_{1i}^*x_kx_k^*\bu_{2i}}{\beta_{k0}}I_{{\cal E}^c}-\frac1n\sum_{k=1}^n \rE\frac{\bu_{1i}^*x_kx_k^*\bu_{2i}\varepsilon_k}{\beta_{k0}\beta_k}I_{{\cal E}^c}\right|\nonumber\\
&=\left|\frac1n\sum_{k=1}^n\rE \frac{\bu_{1i}^*x_kx_k^*\bu_{2i}}{\beta_{k0}}-\frac1n\sum_{k=1}^n\rE \frac{\bu_{1i}^*x_kx_k^*\bu_{2i}}{\beta_{k0}}I_{{\cal E}}-\frac1n\sum_{k=1}^n \rE\frac{\bu_{1i}^*x_kx_k^*\bu_{2i}\varepsilon_k}{\beta_{k0}\beta_k}I_{{\cal E}^c}\right|\\
&\le P({\cal E})+\frac{{K_0}}n\sum_{k=1}^n\rE|\bu_{1i}^*x_kx_k^*\bu_{2i}\varepsilon_k|\nonumber\\
&\le {K_0}n^{-1/2}\log n \le  n^{-\delta},\nonumber
\end{align}
where {$K_0$} is an absolute constant independent of $i$, and may take different values at different appearances. Here, we remind the reader that the first term in (\ref{eqbj3}) is zero by  the assumption that 
$\bu_{1i}^*\bu_{2i}=0$ and the probability of ${\cal E}$ is $O(n^{-1})$ by (\ref{eqbf06}). 

Next, we have
\bqn
|\rE J_{2i}I_{{\cal E}^c}|&=& \left|\frac1n\sum_{k_1\ne k_2 } \rE\frac{\bu_{1i}^*x_{k_1}x_{k_2}^*\bu_{2i} \beta_{k_1k_2}}{\beta_{k_1;k_2}\beta_{k_2;k_1}-\beta_{k_1k_2}^2}\right|I_{{\cal E}^c}\\
&\le& \sum_{l=1}^5\left|\rE J_{2li}\right|I_{{\cal E}^c},
\eqn
where 
\bqn
J_{21i}&=&\frac1n\sum_{k_1\ne k_2 } \frac{\bu_{1i}^*x_{k_1}x_{k_2}^*\bu_{2i} \beta_{k_1k_2}}{\bar\beta^2_{\color{red}k_1; k_2}}\\ 
J_{22i}&=&\frac1n\sum_{k_1\ne k_2 } \frac{\bu_{1i}^*x_{k_1}x_{k_2}^*\bu_{2i} \beta_{k_1k_2}(\varepsilon_{k_1;k_2}+\varepsilon_{k_2;k_1})}{\bar\beta^3_{\color{red}k_1; k_2}}\\ 
J_{23i}&=&\frac1n\sum_{k_1\ne k_2 } \frac{(\bu_{1i}^*x_{k_1}x_{k_2}^*\bu_{2i}+\bu_{1i}^*x_{k_2}x_{k_1}^*\bu_{2i}) \beta_{k_1k_2}\varepsilon^2_{k_1;k_2}}{\bar\beta^3_{\color{red}k_1; k_2}\beta_{k_1;k_2}} \\
J_{24i}&=&\frac1n\sum_{k_1\ne k_2 } \frac{\bu_{1i}^*x_{k_1}x_{k_2}^*\bu_{2i} \beta_{k_1k_2}(\varepsilon^2_{k_1;k_2}+\varepsilon_{k_1;k_2}\varepsilon_{k_2;k_1}+\varepsilon^2_{k_2;k_1})}{\bar\beta^2_{\color{red}k_1; k_2}\beta_{k_1;k_2}\beta_{k_2;k_1}}\\
J_{25i}&=&\frac1n\sum_{k_1\ne k_2 } \frac{\bu_{1i}^*x_{k_1}x_{k_2}^*\bu_{2i} \beta_{k_1k_2}^3}{\beta_{k_1;k_2}\beta_{k_2;k_1}(\beta_{k_1;k_2}\beta_{k_2;k_1}-\beta_{k_1k_2}^2)}
\eqn
By elementary calculation, we have
\bqn
|\rE J_{21i}I_{{\cal E}^c}|=|-\rE J_{21i}I_{{\cal E}}|\le K_0P({\cal E})\le O(n^{-t})
\eqn
because $\rE J_{21i}=0$, which follows from $U_2^*\bu_{1i}=0$, and hence 
\bqn
\rE(\bu_{1i}^*x_{k_1}x_{k_2}^*\bu_{2i}\beta_{k_1k_2}|{\bX_{k_1,k_2}})=n^{-1}\bu_{1i}^*\mathcal{H}\bu_{2i}=0,
\eqn
where $\mathcal{H}=U_2D_2^{1/2}G^{-1}_{\color{red}k_1, k_2}D^{1/2}_2U_2^*$. 

Also, we have 
\bqn
\rE J_{22i}I_{{\cal E}^c}=o(n^{-1}\log n), \mbox{ (uniformly in $i$,)}
\eqn
by using
\bqn
&&\rE(\bu_{1i}^*x_{k_1}x_{k_2}^*\bu_{2i}\beta_{k_1k_2}\varepsilon_{k_2;k_1}|{\bX_{k_1,k_2}})\\
&=&n^{-1}\rE(x^*_{k_2}\bu_{2i}\varepsilon_{k_2;k_1}\bu_{1i}^*\mathcal{H}x_{k_2}|\bX_{k_1,k_2})=0 \mbox{ (since } \mathcal{H}\bu_{1i}=0)
\eqn
and
\bqn
&&|\rE(\bu_{1i}^*x_{k_1}x_{k_2}^*\bu_{2i}\beta_{k_1k_2}\varepsilon_{k_1;k_2}|\bX_{k_1k_2})|\\
&=&|n^{-1}\rE(\bu_{1i}^*x_{k_1}x_{k_1}^*\mathcal{H}\bu_{2i}\varepsilon_{k_1;k_2}|\bX_{k_1,k_2})|\\
&=&|n^{-2}\Big(2\bu_{1i}^*\mathcal{H}^2\bu_{2i}+o(\log n)\sum_{l=1}^pu_{1il}\bu_{2i}^*\mathcal{H}\bbe_l\mathcal{H}_{ll}\Big)|\mbox{ (since } \mathcal{H}\bu_{1i}=0)\\
&\le&o(n^{-1}\log n).
\eqn
Applying the Cauchy-Schwarz inequality, one can easily show that for $l=3,4$, 
\bqn
|\rE J_{2il}I_{{\cal E}^c}|&\le& \frac{K_0}n\sum_{k_1\ne k_2}\|\bu_{1i}\|\|\bu_{2i}\|
(\rE|\beta_{k_1k_2}^4)^{1/4}(\rE|\varepsilon_{k_1;k_2}^8|+\rE|\varepsilon_{k_2;k_1}^8|)^{1/4}\\
&\le& K_0n^{-1/4}\log n
\eqn
Furthermore, by the Holder inequality, we have
\bqn
|\rE J_{25i}I_{{\cal E}^c}|&\le& \frac{K_0}n\sum_{k_1\ne k_2}\rE|\bu_{1i}^* x_{k_1}x^*_{k_2}\bu_{2i}\beta_{k_1k_2}^3|\\
&\le& \frac{K_0}n\sum_{k_1\ne k_2}(\rE|\bu_{1i}^* x_{k_1}x^*_{k_2}\bu_{2i}|^2\rE|\beta_{k_1k_2}^6|)^{1/2}\\
&\le& K_0n^{-1/2}\log n.
\eqn
Summing the inequalities above, our assertion (\ref{eqbj2}) is proved.

Next, we shall show that
\be\label{eqbj4}
\max_{i\le N}|R_iI_{{\cal E}^c}-\rE R_iI_{{\cal E}^c}|\le n^{-\delta}, a.s.
\ee
To this end, we will employ the traditional approach of the martingale decomposition: Let $\rE_k$ denote the conditional expectation given vectors {$\bx_1,\cdots, \bx_k$}. Then, we have
\bqn
&&R_iI_{{\cal E}^c}-\rE R_iI_{{\cal E}^c}\\
&=&\sum_{k=1}^n(\rE_{k-1}-\rE_k)R_iI_{{\cal E}^c}=\sum_{k=1}^n(\rE_{k-1}-\rE_k)(R_iI_{{\cal E}^c}-R_{ik}I_{{\cal E}_{1k}^c}I_{{\cal E}_{2k}^c})\\
&=&\frac1n\sum_{k=1}^n(\rE_{k-1}-\rE_k)\gamma_{ki}I_{{\cal E}^c}+\sum_{k=1}^n(\rE_{k-1}-\rE_k)R_{ik}I_{{\cal E}}I_{{\cal E}_{1k}^c}I_{{\cal E}_{2k}^c}\\
&:=&I_{in1}+I_{in2},
\eqn
where 
$$
R_{ik}=\frac{\lambda}n \bu_{1i}^*\bX_{k0}(\lambda\bI_{n-1}-n^{-1}\bX_{k0}^*\Gamma\bX_{k0})^{-1}\bX_{k0}^*\bu_{2i}.
$$
By the inverse matrix formula, we have
\bqn
\gamma_{ki}&=&\frac1{\beta_k}\big(\bu_{1i}^*(\bI_p+n^{-1}{\bf X}_{k0}(\lambda I_{n-1}-n^{-1}{\bf X}^*_{k0}\Gamma{\bf X}_{k0})^{-1}{\bf X}^*_{k0}\Gamma\big)x_k\times\\
&& x_k^*\big(\bI_p+n^{-1}\Gamma{\bf X}_{k0}(\lambda I_{n-1}-n^{-1}{\bf X}^*_{k0}\Gamma{\bf X}_{k0})^{-1}{\bf X}^*_{k0}\big)\bu_{2i}
\eqn
By the Burkholder inequality, we have
\bqn
&&P(\max_{i}|I_{in1}|\ge n^{-\delta})\\
&\le&\sum_{i=1}^N P(|I_{in1}|\ge n^{-\delta})\le\sum_{i=1}^Nn^{l\delta}\rE|I_{in1}|^l\\
&\le&n^{l\delta}\sum_{i=1}^N\left(\sum_{k=1}^nn^{-l}\rE|\gamma_{ki}|^lI_{{\cal E}^c}+\rE\left(n^{-2}\sum_{k=1}^n\rE_{k-1}\gamma_{ni}^2I_{{\cal E}^c}\right)^{l/2}\right).
\eqn
Note that $\|n^{-1/2}\bX\|\le b_+$ implies $\|n^{-1/2}\bX_{k0}\|\le b_+$ and that $\| G^{-1}\|\le B^+$ implies $G_{k_1k_2}^{-1}\|\le B^++\delta$ with an exception of probability of $O(n^{-t})$ for any given $t$; thus, we have
\bqn
&&\rE_{k-1}\gamma_{ni}^2I_{{\cal E}^c}\\
&\le& K_0\rE_{k-1}|\beta_k\gamma_{ni}|^2I_{\{\|n^{-1/2}\bX_{k0}|\|\le b_+\}}I_{\{\|G_{k_1k_2}^{-1}\le B^++\delta\}}+O(n^{-t})\\
&\le& K_0\rE_{k-1}\|\bu_{1i}\|^2\|\bu_{2i}\|^2+O(n^{-t})\le K_0
\eqn
and one can similarly prove that
\bqn
\rE|\gamma_{ki}|^lI_{{\cal E}^c}\le K_0.
\eqn
Therefore, 
\bqa\label{eqbjj1}
P(\max_{i}|I_{in1}|\ge n^{-\delta})\le K_0N(n^{-l+1+l\delta}+n^{-l/2+l\delta})
\eqa
which is summable when $l(\frac12-\delta)>1$.

Finally, one finds that $R_{ik}$ is bounded when ${\cal E}_{1k}^c$ and ${\cal E}_{2k}^c$ happen. Thus, we have
\bqa\label{eqbjj2}
P(\max_{i}|I_{in2}|\ne 0)\le \sum_{i=1}^N\sum_{k=1}^n K_0P({\cal E})=o(n^{-t}),
\eqa
which is summable.

 Combining 
(\ref{eqbjj1}) and (\ref{eqbjj2}), we complete the proof of (\ref{eqbj4}). Then, the lemma follows from (\ref{eqbj2}), (\ref{eqbj4}) and the fact that $I({\cal E})\to 0$, a.s.
\end{proof}

\section{Proof of Theorem \ref{thm2}}\label{SG4MT}

  Following the notations in Lemma~\ref{lembetak}, we still use ${\bf X}=(\bx_1,\cdots,\bx_n)$ and ${\bf Y}=(\by_1,\cdots,\by_n)$ to be two independent random matrices satisfying Assumptions $\bA\sim \bE$ and denote ${\bf X}_k=(\bx_1,\cdots,\bx_k,$ $\by_{k+1},\cdots,\by_n)$ with convention 
${\bf X}={\bf X}_n$ and ${\bf Y}={\bf X}_0$.  ${\bf X}_{k0}=(\bx_1\cdots,\bx_{k-1},\by_{k+1}\cdots \by_n)$ is the overlapping part between $\bX$ and $\bY$,  and $\Omega_M(\phi_{n,k}, {\bf X})$ is simply defined as $\Omega_M({\bf X})$ if no confusion. Note that the difference between ${\bf X}_k$ and ${\bf X}_{k0}$ lies in the $k$th column, that is, $\bx_k$ in ${\bf X}_k$, and the difference between ${\bf X}_{k-1}$ and ${\bf X}_{k0}$ is also in the $k$th column, that is, $\by_k$ in ${\bf X}_{k-1}$.

By applying the inverse matrix formula, we have 
\bqn
&&\left(\Omega_M({\bf X}_{k})-\Omega_M({\bf X}_{k0})\right)\nonumber\\
&=&\frac1{\phi_k\beta_k\sqrt{n}}D_1^{\frac{1}{2}}\Big(\big(1+n^{-2}\bx_k^*\Gamma {\bf X}_{k0}
(\phi_{n,k} \bI_{n-1} -n^{-1}{\bf X}_{k0}^*\Gamma{\bf X}_{k0})^{-2}{\bf X}^*_{k0}\Gamma \bx_k\big) \bI_M\nonumber\\
&&-U_1^*(\bx_k+n^{-1}{\bf X}_{k0}(\phi_{n,k} \bI_{n-1}-n^{-1}{\bf X}^*_{k0}\Gamma{\bf X}_{k0})^{-1}{\bf X}^*_{k0}\Gamma \bx_k)\nonumber\\
&&(\bx_k^*+n^{-1}\bx_k^*\Gamma{\bf X}_{k0}(\phi_{n,k} \bI_{n-1}-n^{-1}{\bf X}^*_{k0}\Gamma{\bf X}_{k0})^{-1}{\bf X}^*_{k0})U_1\Big)D_1^{\frac{1}{2}},
\eqn
where $\Gamma=U_2 D_2U^*_2$, and $\beta_k,\beta_{k0}, \varepsilon_k$ are defined in Lemma~\ref{lembetak}.

Denote
\begin{eqnarray}
\tau_{k0}&=&\Big(1+n^{-2}{\rm tr}\big(\Gamma {\bf X}_{k0}
(\phi_{n,k} \bI_{n-1} -n^{-1}{\bf X}_{k0}^*\Gamma{\bf X}_{k0})^{-2}{\bf X}^*_{k0}\Gamma\big)\Big)\bI_M\nonumber\\
&&- U_1^*
\big(\bI_p+n^{-1}{\bf X}_{k0}(\phi_{n,k} \bI_{n-1}-n^{-1}{\bf X}^*_{k0}\Gamma{\bf X}_{k0})^{-1}{\bf X}^*_{k0}\Gamma\big)\nonumber\\
&&\big(\bI_p+n^{-1}\Gamma{\bf X}_{k0}(\phi_{n,k} \bI_{n-1}-n^{-1}{\bf X}^*_{k0}\Gamma{\bf X}_{k0})^{-1}{\bf X}^*_{k0} \big)U_1\non
\tau_k&=&\Big(n^{-2}\big(\bx_k^*\Gamma {\bf X}_{k0}
(\phi_{n,k} \bI_{n-1} -n^{-1}{\bf X}_{k0}^*\Gamma{\bf X}_{k0})^{-2}{\bf X}^*_{k0}\Gamma \bx_k\non
&&-{\rm tr}(\Gamma {\bf X}_{k0}
(\phi_{n,k} \bI_{n-1} -n^{-1}{\bf X}_{k0}^*\Gamma{\bf X}_{k0})^{-2}{\bf X}^*_{k0}\Gamma)\big)\Big) \bI_M\nonumber\\
&&-U_1^*\big(\bI_p+n^{-1}{\bf X}_{k0}(\phi_{n,k} I_{n-1}-n^{-1}{\bf X}^*_{k0}\Gamma{\bf X}_{k0})^{-1}{\bf X}^*_{k0}\Gamma\big)(\bx_k\bx_k^*-\bI_p)\nonumber\\
&&\big(\bI_p+n^{-1}\Gamma{\bf X}_{k0}(\phi_{n,k} I_{n-1}-n^{-1}{\bf X}^*_{k0}\Gamma{\bf X}_{k0})^{-1}{\bf X}^*_{k0}\big)U_1
\label{tauk}
\end{eqnarray}

Then,  we have
\begin{align*}
&\big(\Omega_M({\bf X}_{k})-\Omega_M({\bf X}_{k0})\big)\\
=&\frac{1}{\phi_{n,k}}\!D_1^{1 \over 2}\!\left(\!\frac{1}{\beta_{k0}\sqrt{n}}\!\big(\!\tau_{k0}\!+\!\tau_k\!\big)\!-\!
\frac{1}{ \beta_{k0}^2\sqrt{n}}\!\big(\!\tau_{k0}\!+\!\tau_k\!\big)\varepsilon_k\!+\!
\frac{1}{ \beta_{k0}^2\beta_k\sqrt{n}}\!\big(\!\tau_{k0}\!+\!\tau_k\!\Big)\varepsilon_k^2\!\right)\!D_1\!^{1 \over 2}.
\end{align*}
Similarly, we have
\begin{align*}
&\big(\Omega_M({\bf X}_{k-1})-\Omega_M({\bf X}_{k0})\big)\\
=&\frac{1}{\phi_{n,k}}\!D_1^{1 \over 2}\!\!\left(\!\frac{1}{\beta_{k0}\sqrt{n}}\!\big(\!\tau_{k0}\!+\!\tau_{ky}\!\big)\!-\!
\frac{1}{ \beta_{k0}^2\sqrt{n}}\!\big(\!\tau_{k0}\!+\!\tau_{ky}\!\big)\!\varepsilon_{ky}\!+\!
\frac{1}{ \beta_{k0}^2\beta_{ky}\sqrt{n}}\!\big(\!\tau_{k0}\!+\!\tau_{ky}\!\Big)\!\varepsilon_{ky}^2\!\right)\!\!D_1\!^{1 \over 2},
\end{align*}
where $\beta_{ky}$, $\tau_{ky}$ and $\varepsilon_{ky}$ are similarly defined as $\beta_k$, $\tau_{k}$ nand $\varepsilon_k$ with $\bx_k$ replaced by $\by_k$. 

For  any $M\times M$ symmetric matrix $\bf W$,
a proposition about  $\rE (\tr \bW \tau_k)^2$  is  formulated  in the following lemma, which
is used  in the process of proof for 
the G4MT.

\begin{lemma}
Under the Assumptions $\bA \sim \bE$, for  $\tau_k$ defined in (\ref{tauk}) and any $M\times M$ symmetric matrix $\bf W$, we have 
\be
\rE\big(\tr(\bW\tau_k)\big)^2 =2{\rm tr} (\bDe^2)+  o(1)
\ee
\label{lemtauk}
where $\bDe$ is defined in (\ref{bDe}).
\end{lemma}
\begin{proof} 
 By the definition (\ref{tauk}),  let 
  \bqn
\bH_1&=&\frac{1}{n}{\bf X}_{k0}(\phi_{n,k} I_{n-1}-n^{-1}{\bf X}^*_{k0}\Gamma{\bf X}_{k0})^{-1}{\bf X}^*_{k0}\Gamma\\
\bH_2&=&\frac{1}{n}\Gamma {\bf X}_{k0}(\phi_{n,k} I_{n-1}-n^{-1}{\bf X}^*_{k0}\Gamma{\bf X}_{k0})^{-2}{\bf X}^*_{k0}\Gamma,
\eqn
 and then  we have 
 \begin{eqnarray*}
\tau_k&=&\frac{1}{n}\big(\bx_k^*\bH_2 \bx_k-{\rm tr}( \bH_2)\big) \bI_M-U_1^*(\bI_p+\bH_1)(\bx_k\bx_k^*-\bI_p)(\bI_p+ \bH_1^*)U_1\\
& :=& \tau_{k1} -\tau_{k2} 
\end{eqnarray*}

 For any $M\times M$ symmetric matrix $\bf W$,
 \be
\rE\big(\tr(\bW\tau_k)\big)^2= \rE\big(\tr(\bW\tau_{k1})\big)^2+ \rE\big(\tr(\bW\tau_{k2})\big)^2-2\rE\big(\tr(\bW\tau_{k1})\tr(\bW\tau_{k2})\big).
\ee
Then, by the equation (1.15) in \cite{BaiSilverstein2004}, we have 
\bqn
\rE\big(\tr(\bW\tau_{k1})\big)^2
&=&\frac{1}{n^2}\big( (\rE|x_{11}|^4-\rE|x_{11}^2|^2-2)\sum\limits_{t}  h_{2,tt}^2+2{\rm tr} (\bH_2^2)\big)({\rm tr} \bW)^2\\
&=&o(n^{-1}\log n)
\eqn
where $h_{2,ts}$'s  the $(t,s)$ element of the matrix $\bH_2$.

Let 
\be
\bDe=(\bI_p+ \bH_1^*)U_1\bW U_1^*(\bI_p+\bH_1),\label{bDe}
\ee
then 
 \bqn
\big(\tr(\bW\tau_{k2})\big)^2&=&\big(\bx_k^*\bDe \bx_k-{\rm tr}( \bDe)\big)^2\\
&=& (\rE|x_{11}|^4-\rE|x_{11}^2|^2-2)\sum\limits_{t}  \Delta_{tt}^2+2{\rm tr} (\bDe^2),
\eqn
where $\Delta_{ts}$'s  the $(t,s)$ element of the matrix $\bDe$.

By the Assumption~$\bD$ and Lemma~\ref{lembb1},
 we have
\begin{align*}
\rE|x_{11}|^4\sum\limits_{t}  \Delta_{tt}^2 & \leq \rE|x_{11}|^4\sum\limits_{t}  (2 \bbe_t^* U_1\bW U_1^* \bbe_t + 
2 \bbe_t^*\bH_1^* U_1\bW U_1^*\bH_1 \bbe_t)^2\\
& \leq \rE|x_{11}|^4 \sum\limits_{t}  8 \Big( ( \bbe_t^* U_1\bW U_1^* \bbe_t )^2+ 
( \bbe_t^*\bH_1^* U_1\bW U_1^*\bH_1 \bbe_t)^2\Big)\\
&\leq o(1)+o( n^{-4\delta}\log n)
\end{align*}
where $\bbe_t$ is a $p$-dimensional unit vector with the $t$th element  equal to 1 and others equal to 0. 
By similar techniques, we also obtain 
\begin{align*}
{\rm tr} (\bDe^2)&=\sum\limits_{t,s}  \Delta_{ts}^2  \leq \sum\limits_{t,s}  (2 \bbe_t^* U_1\bW U_1^* \bbe_s + 
2 \bbe_t^*\bH_1^* U_1\bW U_1^*\bH_1 \bbe_s)^2\\
& \leq \sum\limits_{t,s}  8 \Big( ( \bbe_t^* U_1\bW U_1^* \bbe_s )^2+ 
( \bbe_t^*\bH_1^* U_1\bW U_1^*\bH_1 \bbe_s)^2\Big)\\
&\leq 8\Big( 1+o(n^{-4\delta})\Big)\tr (\bW\bW^*)\\
&= O(1) 
\end{align*}
Thus, 
\[\big(\tr(\bW\tau_{k})\big)^2= 2{\rm tr} (\bDe^2)+o(1).\]
\end{proof}

Now, we are in position to complete the proof of the G4MT. To this end, we only need to show that the difference in the
 characteristic functions tends to zero. For any $M\times M$ symmetric matrix $\bf W$, 
 $\rE e^{i{\rm tr}\left({\bf W}\Omega_M({\bf X}) \right)}-\rE e^{i{\rm tr}({\bf W}\Omega_M({\bf Y}))}\to 0$ is proved in the following part. Using the notations we introduced above, we have
 \begin{align}
& \rE e^{i{\rm tr}({\bf W}\Omega_M({\bf X}))}-\rE e^{i{\rm tr}({\bf W}\Omega_M({\bf Y}))}\nonumber\\
=&\sum_{k=1}^n\Big( \rE e^{i{\rm tr}\big({\bf W}\Omega_M({\bf X}_{k})\big)}-\rE e^{i{\rm tr}\big({\bf W}\Omega_M({\bf X}_{k-1})\big)}\Big) \nonumber\\
=&\sum_{k=1}^n \!\rE e^{i{\rm tr}\big({\bf W}\Omega_M({\bf X}_{k0})\big)}\!\Big(\rE_k e^{i{\rm tr}\big({\bf W}(\Omega_M({\bf X}_{k})\!-\!\Omega_M({\bf X}_{k0}))\big)} \!-\!\rE_k e^{i{\rm tr}\big({\bf W}(\Omega_M({\bf X}_{k\!-\!1})\!-\!\Omega_M({\bf X}_{k0}))\big)}\!\Big) \nonumber\\
=&\sum_{k=1}^n\rE e^{i{\rm tr}\!\Big(\!{\bf W}\!\big(\!\Omega_M({\bf X}_{k0})\!+\!\frac{D_1^{\frac{1}{2}}\tau_{k0} D_1^{\frac{1}{2}}}{ \phi_{n,k}\beta_{k0}\sqrt{n}}\!\big)\!\Big)\!}\!\bigg(\!\rE_k e^{i{\rm tr}\Big(\frac{{\bf W} D_1^{\frac{1}{2}}}{\phi_{n,k}}\big(\frac{\tau_k}{ \beta_{k0}\sqrt{n}}-\frac{(\tau_{k0}+\tau_k)\varepsilon_k}{\beta_{k0}^2\sqrt{n}}\!+\!\frac{(\tau_{k0}+\tau_k)\varepsilon^2_k}{\beta_{k0}^2\beta_k\sqrt{n}}\big)D_1^{\frac{1}{2}}\Big)}\non
\quad& -\rE_k e^{i{\rm tr}\Big(\frac{{\bf W} D_1^{\frac{1}{2}}}{\phi_{n,k}}\big(\frac{\tau_{ky}}{ \beta_{k0}\sqrt{n}}-\frac{(\tau_{k0}+\tau_{ky})\varepsilon_{ky}}{\beta_{k0}^2\sqrt{n}}+\frac{(\tau_{k0}+\tau_{ky})\varepsilon^2_{ky}}{\beta_{k0}^2\beta_{ky}\sqrt{n}}\big)D_1^{\frac{1}{2}}\Big)}\bigg)\non
=&\sum_{k=1}^n\rE e^{i{\rm tr}\!\Big(\!{\bf W}\!\big(\!\Omega_M({\bf X}_{k0})\!+\!\frac{D_1^{\frac{1}{2}}\tau_{k0} D_1^{\frac{1}{2}}}{ \phi_{n,k}\beta_{k0}\sqrt{n}}\!\big)\!\Big)\!}\non
\quad&\bigg(\!\rE_k e^{i{\rm tr}\!\Big(\!\frac{{\bf W} D_1^{\frac{1}{2}}}{\phi_{n,k}}\!\big(\!\frac{\tau_k}{ \beta_{k0}\sqrt{n}}-\frac{(\tau_{k0}+\tau_k)\varepsilon_k}{\beta_{k0}^2\sqrt{n}}\!\big)\!D_1^{\frac{1}{2}}\!\Big)} \!-\! \rE_k e^{i{\rm tr}\!\Big(\!\frac{{\bf W} D_1^{\frac{1}{2}}}{\phi_{n,k}}\!\big(\!\frac{\tau_{ky}}{ \beta_{k0}\sqrt{n}}-\frac{(\tau_{k0}+\tau_{ky})\varepsilon_{ky}}{\beta_{k0}^2\sqrt{n}}\!\big)\!D_1^{\frac{1}{2}}\!\Big)\!}\bigg)\nonumber\\
&+\sum_{k=1}^n\rE e^{i{\rm tr}\!\Big(\!{\bf W}\!\big(\!\Omega_M({\bf X}_{k0})\!+\!\frac{D_1^{\frac{1}{2}}\tau_{k0} D_1^{\frac{1}{2}}}{ \phi_{n,k}\beta_{k0}\sqrt{n}}\!\big)\!\Big)\!}\non
& \Bigg(\rE_k e^{i{\rm tr}\Big(\frac{{\bf W} D_1^{\frac{1}{2}}}{\phi_{n,k}}\big(\frac{\tau_k}{ \beta_{k0}\sqrt{n}}-\frac{(\tau_{k0}+\tau_k)\varepsilon_k}{\beta_{k0}^2\sqrt{n}}\big)D_1^{\frac{1}{2}}\Big)}\Big(e^{i{\rm tr}\big(\frac{{\bf W} D_1^{\frac{1}{2}}(\tau_{k0}+\tau_k)D_1^{\frac{1}{2}}\varepsilon^2_k}{\phi_{n,k}\beta_{k0}^2\beta_k\sqrt{n}}\big)}-1\Big)\non
\quad &
 -\rE_k e^{i{\rm tr}\Big(\frac{{\bf W} D_1^{\frac{1}{2}}}{\phi_{n,k}}\big(\frac{\tau_{ky}}{ \beta_{k0}\sqrt{n}}-\frac{(\tau_{k0}+\tau_{ky})\varepsilon_{ky}}{\beta_{k0}^2\sqrt{n}}\big)D_1^{\frac{1}{2}}\Big)}
\Big(e^{i{\rm tr}\big(\frac{{\bf W} D_1^{\frac{1}{2}}(\tau_{k0}+\tau_{ky})D_1^{\frac{1}{2}}\varepsilon^2_{ky}}{\phi_{n,k}\beta_{k0}^2\beta_{ky}\sqrt{n}}\big)}-1\Big)\Bigg).
\label{eqb5}
 \end{align}
where $\rE_k=\rE(\cdot|{\bf X}_{k0})$. 

By the Lemma~{\ref{lembetak}}, it follows that $\beta_{k0}\to -\frac 1{ \lambda m(\lambda)}\neq 0$ and {$ \varepsilon_k\to 0 $}, and we conclude that the last two terms are $o(1)$. As an example,
\begin{align}
&\Bigg|\sum_{k=1}^n\rE e^{i{\rm tr}\!\Big(\!{\bf W}\!\big(\!\Omega_M({\bf X}_{k0})\!+\!\frac{D_1^{\frac{1}{2}}\tau_{k0} D_1^{\frac{1}{2}}}{ \phi_{n,k}\beta_{k0}\sqrt{n}}\!\big)\!\Big)\!}\non
&\quad\cdot \rE_k e^{i{\rm tr}\Big(\frac{{\bf W} D_1^{\frac{1}{2}}}{\phi_{n,k}}\big(\frac{\tau_k}{ \beta_{k0}\sqrt{n}}-\frac{(\tau_{k0}+\tau_k)\varepsilon_k}{\beta_{k0}^2\sqrt{n}}\big)D_1^{\frac{1}{2}}\Big)}\Big(e^{i{\rm tr}\big(\frac{{\bf W} D_1^{\frac{1}{2}}(\tau_{k0}+\tau_k)D_1^{\frac{1}{2}}\varepsilon^2_k}{\phi_{n,k}\beta_{k0}^2\beta_k\sqrt{n}}\big)}-1\Big)\Bigg|\nonumber\\
&\le\sum_{k=1}^n\rE\left(\rE_k(2I(|\varepsilon_k|\ge\delta_k))+\frac{K_0}{\sqrt{n}}\rE_k|{\rm tr}\bW(\tau_{k0}+\tau_k)\varepsilon_k^2|\right)\label{cheb}\\
&\le\sum_{k=1}^nK_0\big(\rE\varepsilon_k^4+\frac1{\sqrt{n}}(\rE|({\rm tr}\bW(\tau_{k0}+\tau_k)|^2\rE\varepsilon_k^4))^{1/2}\big)=o(1),\nonumber
\end{align} 
where we have used the facts from Lemma \ref{lembetak} that 
\begin{eqnarray*}
&&\rE_k\varepsilon_k^4\le
o(n^{-1}),\\
&&\rE_k|{\rm tr}\bW(\tau_{k0}+\tau_k)|^2\le K_0
\end{eqnarray*}
with  $K_0$ being a suitable constant valued different at different appearances. 

 Here, $\delta_k$ should not be too small, like the half of the non-zero limit of $\beta_{k0}$; then, we have 
\[|\beta_k|\geq|\beta_{k0}| -|\varepsilon_k| \] 
and moreover $|\beta_k|\geq 1/2|\beta_{k0}|$, which is bounded from below. Thus,  we give the  partition as (\ref{cheb}),  and use the Chebyshev's inequality when $|\varepsilon_k|\ge\delta_k$ and then apply the Taylor expansion and Cauchy-Schwartz  inequality to the case of  $|\varepsilon_k|<\delta_k$. Similarly, we can show the other term is $o(1)$.

Therefore, we have 
 \begin{align}
& \rE e^{i{\rm tr}{\bf W}\Omega_M({\bf X})}-\rE e^{i{\rm tr}{\bf W}\Omega_M({\bf Y})}\nonumber\\
=&\sum_{k=1}^n\rE e^{i{\rm tr}\!\Big(\!{\bf W}\!\big(\!\Omega_M({\bf X}_{k0})\!+\!\frac{D_1^{\frac{1}{2}}\tau_{k0} D_1^{\frac{1}{2}}}{ \phi_{n,k}\beta_{k0}\sqrt{n}}\!\big)\!\Big)\!}\!\bigg(\!\rE_k e^{i{\rm tr}\Big(\frac{{\bf W} D_1^{\frac{1}{2}}}{\phi_{n,k}}\big(\frac{\tau_k}{ \beta_{k0}\sqrt{n}}-\frac{(\tau_{k0}+\tau_k)\varepsilon_k}{\beta_{k0}^2\sqrt{n}} \big)D_1^{\frac{1}{2}}\Big)}\non
\quad& -\rE_k e^{i{\rm tr}\Big(\frac{{\bf W} D_1^{\frac{1}{2}}}{\phi_{n,k}}\big(\frac{\tau_{ky}}{ \beta_{k0}\sqrt{n}}-\frac{(\tau_{k0}+\tau_{ky})\varepsilon_{ky}}{\beta_{k0}^2\sqrt{n}} \big)D_1^{\frac{1}{2}}\Big)}\bigg)+o(1)\nonumber
\end{align}

By the same approach, one may show that 
 \begin{align}
& \rE e^{i{\rm tr}{\bf W}\Omega_M({\bf X})}-\rE e^{i{\rm tr}{\bf W}\Omega_M({\bf Y})}\nonumber\\
=&\sum_{k=1}^n\rE e^{i{\rm tr}\!\Big(\!{\bf W}\!\big(\!\Omega_M({\bf X}_{k0})\!+\!\frac{D_1^{\frac{1}{2}}\tau_{k0} D_1^{\frac{1}{2}}}{ \phi_{n,k}\beta_{k0}\sqrt{n}}\!\big)\!\Big)\!}\!\bigg(\!\rE_k e^{i{\rm tr}\Big(\frac{{\bf W} D_1^{\frac{1}{2}}}{\phi_{n,k}}\big(\frac{\tau_k}{ \beta_{k0}\sqrt{n}}-\frac{\tau_{k0}\varepsilon_k}{\beta_{k0}^2\sqrt{n}} \big)D_1^{\frac{1}{2}}\Big)}\non
\quad& -\rE_k e^{i{\rm tr}\Big(\frac{{\bf W} D_1^{\frac{1}{2}}}{\phi_{n,k}}\big(\frac{\tau_{ky}}{ \beta_{k0}\sqrt{n}}-\frac{\tau_{k0}\varepsilon_{ky}}{\beta_{k0}^2\sqrt{n}} \big)D_1^{\frac{1}{2}}\Big)}\bigg)+o(1)\label{eqb6}
\end{align}
%

Since
\begin{eqnarray*}
&&\bigg({\rm tr}\big(\frac{{\bf W}\tau_{k}}{ \beta_{k0}\sqrt{n}}-\frac{{\bf W}\tau_{k0}\varepsilon_{k}}{\beta_{k0}^2\sqrt{n}}\big)
\bigg)^2\\
&=&\frac{1}{n}\bigg( \frac{1}{\beta_{k0}^2} \big(\tr(\bW\tau_k)\big)^2-\frac{2}{\beta_{k0}^3}\tr(\bW\tau_k)\tr(\bW\tau_{k_0}\varepsilon_k)+\frac{1}{\beta_{k0}^4}\big(\tr(\bW\tau_{k_0})\big)^2\varepsilon_k^2\bigg)
\end{eqnarray*} 
and noted that
\[\rE \Big( \tr(\bW\tau_k)\tr(\bW\tau_{k_0}\varepsilon_k)\Big)=o(1)\]
and
\[\rE\bigg(\big(\tr(\bW\tau_{k_0})\big)^2\varepsilon_k^2\bigg)=o(n^{-1}\log n).\]

Then, by  Lemma~\ref{lemtauk}, we have 
\begin{align*}
&\bigg({\rm tr}\big({\bf W}\frac{\tau_{k}}{ \beta_{k0}\sqrt{n}}-\frac{\tau_{k0}\varepsilon_{k}}{\beta_{k0}^2\sqrt{n}}\big)
\bigg)^2= \frac{1}{n\beta_{k0}^2} \big(2{\rm tr} (\bDe^2)+o(1)\big) 
\end{align*}

Because $\bX$ and $\bY$ satisfy the Assumptions $\bA \sim \bE$; then, 
${\rm tr} (\bDe^2)$ 
are identical with $\by_k$ instead of $\bx_k$. Furthermore, $\alpha_i/\phi_{n,k}, i=1,\cdots, K$ are bounded. 
Therefore, 
by (\ref{eqb6}), we have  
\begin{align}
& \rE e^{i{\rm tr}{\bf W}\Omega_M({\bf X})}-\rE e^{i{\rm tr}{\bf W}\Omega_M({\bf Y})}\nonumber\\
\leq &\Bigg|\sum_{k=1}^n\rE e^{i{\rm tr}\!\Big(\!{\bf W}\!\big(\!\Omega_M({\bf X}_{k0})\!+\!\frac{D_1^{\frac{1}{2}}\tau_{k0} D_1^{\frac{1}{2}}}{ \phi_{n,k}\beta_{k0}\sqrt{n}}\!\big)\!\Big)\!}\non
&K_0\Bigg(\rE_k \Big(e^{i{\rm tr}\big({\bf W}(\frac{\tau_k}{ \beta_{k0}\sqrt{n}}-\frac{\tau_{k0}\varepsilon_k}{\beta_{k0}^2\sqrt{n}})\big)}\!\!-\!1\!-\!\!i{\rm tr}\Big({\bf W}\big(\frac{\tau_k}{ \beta_{k0}\sqrt{n}}-\frac{\tau_{k0}\varepsilon_k}{\beta_{k0}^2\sqrt{n}}\big)\Big)
+\frac{1}{n\beta_{k0}^2} {\rm tr} (\bDe^2)
\Big)\nonumber\\
& \!\!-\!\rE_k \Big(e^{i{\rm tr}\big({\bf W}(\frac{\tau_{ky}}{ \beta_{k0}\sqrt{n}}-\frac{\tau_{k0}\varepsilon_{ky}}{\beta_{k0}^2\sqrt{n}})\big)}\!\!-\!1\!-\!\!i{\rm tr}\Big({\bf W}\big(\frac{\tau_{ky}}{ \beta_{k0}\sqrt{n}}\!-\!\frac{\tau_{k0}\varepsilon_{ky}}{\beta_{k0}^2\sqrt{n}}\big)\Big)\!+\!\frac{1}{n\beta_{k0}^2} {\rm tr} (\bDe^2)\!\Big)\!\Bigg)\!\Bigg| \non
\le & \tilde{K_0}~o(1),\nonumber
\end{align}
where $K_0$ is a suitable bounded constant  taking different values at different appearances. 

The proof of Theorem~{\ref{thm2}} is completed.

\section{Proof of Corollary~\ref{coro1}: Limiting distribution of $\Omega_M(\phi_{n,k},\bX)$}
\label{SOmega}
\begin{proof}
According to the Theorem~\ref{thm2}, we can derive the limiting distribution of $\Omega_M(\phi_{n,k},\bX)$ under the Gaussian assumption of the entries from $\bX$.
Define  ${\bm \xi_1}=U_1^*\bX$ and  ${\bm \xi_2}=U_2^*\bX$, where $U=(U_1,U_2)$ is defined in (\ref{UDU}).
Then, ${\bm \xi_1}$ and ${\bm \xi_2}$ are independent random sample matrixes with the elements from $\mathcal{N}(0,1)$.
Further,  by the expression of  $\Omega_M(\phi_{n,k},\bX)$, we have 
\begin{align*}
&\Omega_M(\phi_{n,k}, {\bm \xi})\\
&= \frac{1}{\sqrt{n}}D_1^{\frac{1}{2}}\Big({\rm tr}\big( (\phi_{n,k} \bI_{n}- \frac1n {\bm \xi_2}^*D_2 {\bm \xi_2})^{-1} \big)  \bI_M-{\bm \xi_1} (\phi_{n,k} \bI_{n}- \frac1n {\bm \xi_2}^*D_2 {\bm \xi_2})^{-1} {\bm \xi_1}^*\Big)D_1^{\frac{1}{2}}. 
\end{align*}
Let ${\bm \xi^*_{1,i}}$ be the $i$th row of ${\bm \xi_1}$;  
then,  the $(i,j)$ element of $\Omega_M(\phi_{n,k}, {\bm \xi})$ is 
defined as  
\begin{align*}
&\omega_{ij}
\!=\! \frac{\alpha_k}{\sqrt{n}}\!\Big(\!{\rm tr}\big( (\phi_{n,k} \bI_{n}\!-\! \frac1n {\bm \xi_2}^*D_2 {\bm \xi_2})^{-1} \big) \delta_{i,j}\!-\!{\bm \xi^*_{1,i}} \!(\!\phi_{n,k} \bI_{n}\!-\! \frac1n {\bm \xi_2}^*D_2 {\bm \xi_2})^{-1} \!{\bm \xi_{1,j}}\!\Big)\!,
\end{align*}
where $\delta_{i,j}=0$ if $i \neq j$ and $\delta_{i,i}=1$.

{
By the classical limiting theory,  it is easily obtained that 
\[\rE(\omega_{ij})\to 0 ~\text{and}~  (\rE|\omega_{ij}|)^2 \to 0\]
Furthermore, by the formula (1.15) of \cite{BaiSilverstein2004} and Gaussian assumption, we have 
 \[\rE(\omega_{ii}^2) \to 2\alpha_{k}^2\um_2(\phi_k) ~\text{and}~\rE(\omega_{ij}^2) \to \alpha_{k}^2\um_2(\phi_k), ~\text{if}~i \neq j\]
for the real case;
and 
\[\rE(\omega_{ij}^2) \to \alpha_{k}^2\um_2(\phi_k), ~i \leq j\]
for the complex case.

Let $\theta_k=\alpha_{k}^2\um_2(\phi_k)$, 
then it is concluded that $\Omega_M(\phi_{n,k},\bX)$
converges weakly to  an $M\times M$ Hermitian matrix $\Omega_{\phi_{k}}=(\Omega_{ij})$,  where $\frac{1}{\sqrt{\theta_k}}\Omega_{\phi_{k}}$ is GOE for the real case, with the entries above the diagonal being ${\rm i.i.d.} \mathcal{N}(0,1)$ and the entries on the diagonal being ${\rm i.i.d.} \mathcal{N}(0,2)$. 
For the complex case, the $\frac{1}{\sqrt{\theta_k}}\Omega_{\phi_{k}}$ is  GUE, whose  entries  are all ${\rm i.i.d.} \mathcal{N}(0,1)$.}

\begin{itemize}
\item  The cases  involved with the 4th moment.
\end{itemize}

For the cases  where the Assumption~{\bf D} is not met and the 4th moment is bounded, we reconsider  
\begin{align*}
&\omega_{ij}
\!=\! \frac{\alpha_k}{\sqrt{n}}\!\Big(\!{\rm tr}\big( (\phi_{n,k} \bI_{n}\!-\! \frac1n {\bm \xi_2}^*D_2 {\bm \xi_2})^{-1} \big) \delta_{i,j}\!-\!{\bm \xi^*_{1,i}} \!(\!\phi_{n,k} \bI_{n}\!-\! \frac1n {\bm \xi_2}^*D_2 {\bm \xi_2})^{-1} \!{\bm \xi_{1,j}}\!\Big)\!.
\end{align*}
by the formula (1.15) of \cite{BaiSilverstein2004}, we have
\be
\rE(\omega_{ii}^2) \to 2\alpha_{k}^2\um_2(\phi_k)+\alpha_{k}^2\frac 1{n} \sum\limits_{s=1}^n(\sum\limits_{t=1}^pu_{ti}^4\rE|x_{11}|^4-3)a_{ss}^2
\label{wii}
\ee
where ${\bf u}_i=(u_{1i}, \cdots, u_{pi})'$ is the $i$th column of the matrix $U_1$,  $a_{ss}$ is the $(s,s)$-th element of the matrix $(\!\phi_{n,k} \bI_{n}\!-\! \displaystyle\frac1n {\bm \xi_2}^*D_2 {\bm \xi_2})^{-1}$. For a further step,  we detail the $a_{ss}$ as follows. By equation (\ref{inout}), we have 
\begin{align}
&\phi_{n,k}(\!\phi_{n,k} \bI_{n}\!-\! \frac1n {\bm \xi_2}^*D_2 {\bm \xi_2})^{-1}
=\phi_{n,k}(\!\phi_{n,k} \bI_{n}\!-\! \frac1n \bX U_2D_2U_2^*\bX )^{-1}\\
=&\bI_n +\frac{1}{n}\bX^*U_2D_2^{1 \over 2}(\!\phi_{n,k} \bI_{p-M}\!-\! \frac1n D_2^{1 \over 2}U_2^*\bX\bX^*U_2D_2^{1 \over 2})^{-1}D_2^{1 \over 2}U_2^*\bX
\end{align}
then $a_{ss}=(1+b_{ss})/\phi_{n,k}$, where
$b_{ss}$ is the $(s,s)$-th element of the matrix $\displaystyle\frac{1}{n}\bX^*U_2D_2^{1 \over 2}(\phi_{n,k} \bI_{p-M}- \frac1n D_2^{1 \over 2}U_2^*\bX\bX^*U_2D_2^{1 \over 2})^{-1}D_2^{1 \over 2}U_2^*\bX$. 
Let $e_s$ be the $n$-dimensional  column vector with the $s$th element equal to 1and others being 0. Since 
\begin{align}
b_{ss}&=e_s^*\bX^*U_2D_2^{1 \over 2}(\!\phi_{n,k} \bI_{p-M}\!-\! \frac1n D_2^{1 \over 2}U_2^*\bX\bX^*U_2D_2^{1 \over 2})^{-1}D_2^{1 \over 2}U_2^*\bX e_s\\
&=-\frac{\frac{1}{n}e_s^*\bX^*U_2D_2^{1 \over 2}(\!\phi_{n,k} \bI_{p-M}\!-\! \frac1n D_2^{1 \over 2}U_2^*\bX_{-s}\bX_{-s}^*U_2D_2^{1 \over 2})^{-1}D_2^{1 \over 2}U_2^*\bX e_s}{1+\frac{1}{n}e_s^*\bX^*U_2D_2^{1 \over 2}(\!\phi_{n,k} \bI_{p-M}\!-\! \frac1n D_2^{1 \over 2}U_2^*\bX_{-s}\bX_{-s}^*U_2D_2^{1 \over 2})^{-1}D_2^{1 \over 2}U_2^*\bX e_s},
\end{align}
where $\bX_{-s}$ is the matrix $\bX$ without the $s$th column.
Then, by the similar derivation of Lemma~6.1 in \cite{BaiYao2008}, we have 
\[b_{ss} \to -\frac{cm(\phi_{k})}{1+cm(\phi_{k})},\]
 where $m(\phi_{k})$ is the limit of $\tilde m_p(\phi_{k})= \frac{1}{p-M} \sum_{q=M+1}^p \frac{d_q}{l_q-\phi_{k}}$ with 
  $d_q$ being the $q$th diagonal element of the matrix $D_2$, and $l_q$'s are the eigenvalues of  the matrix $\displaystyle\frac1n D_2^{1 \over 2 }U_2^* \bX\bX^* U_2D_2^{1 \over 2 }$.
  If $D_2=\bI_{p-M}$, then $m(\phi_{k})$ is the Stieltjes transform of the LSD of  the matrix $\displaystyle\frac1n U_2^* \bX\bX^* U_2$, {\rm i.e.}
  \be
m(\phi_{k})=\displaystyle\int \frac{1}{x-\phi_{k}} \md F(x)  \ee
with $F(x)$  being the LSD of  the matrix $\displaystyle\frac1n U_2^* \bX\bX^* U_2$.
Therefore,
\[a_{ss}\to 1/\big(\phi_{k}(1+cm(\phi_{k}))\big).\]

Define 
\[\nu_k=\alpha_k^2/\big(\phi_k(1+cm(\phi_k))\big)^2; ~\beta_x=(\sum\limits_{t=1}^pu_{ti}^4\rE|x_{11}|^4-3), \]
where ${\bf u}_i=(u_{1i}, \cdots, u_{pi})'$ is the $i$th column of the matrix $U_1$. If the covariance matrix is a diagonal matrix, then $\beta_x=(\sum\limits_{t=1}^pu_{ti}^4\rE|x_{11}|^4-3)=\rE|x_{11}|^4-3$. 
Then, by equation~(\ref{wii}), we obtain
\[\rE(\omega_{ii}^2) \to 2\alpha_{k}^2\um_2(\phi_k)+\beta_x \nu_k\]
for the diagonal or diagonal block independent population covariance matrix in the real case.
Moreover, $\rE(\omega_{ij}^2) \to \alpha_{k}^2\um_2(\phi_k), ~\text{if}~i \neq j$
for the real case.
For the complex case.
\[\rE(\omega_{ij}^2) \to \alpha_{k}^2\um_2(\phi_k), ~i \leq j.\]
 \end{proof}


\begin{thebibliography}{9}
\bibitem[Bai and Ng(2002)]{BaiNg2002}
Bai, J. and  Ng, S. (2002). \newblock{Determining the number of factors in approximate factor models.}
\newblock\emph{Econometrica}, \textbf{70}, 191-221.

\bibitem[Bai, {\rm et al.}(1991)]{BaiMiaoRao1991}
Bai, Z.D., Miao, B.Q.  and Rao, C. Radbakrisbna. (1991). \newblock{Estimation of directions of arrival  of signals: Asymptotic results.}
\newblock\emph{Advances in Spectrum Analysis and Array Processing, Vol. I, edited by Simon Haykin, Prentice Hall's West Nyack, New York}, pp 327-347.

\bibitem[Bai and Silverstein(1998)]{BaiSilverstein1998}
Bai, Z. D.  and Silverstein, J.W. (1998). \newblock{No eigenvalues outside the support of the limiting spectral distribution of large-dimensional sample covariance matrices.}
\newblock\emph{Ann. Probab.}, Vol. \textbf{26},  \textbf{1}, 316-345.

\bibitem[Bai and Silverstein(1999)]{BaiSilverstein1999}
Bai, Z. D. and Silverstein, J. W. (1999). \newblock{Exact separation of eigenvalues of large dimensional sample covariance matrices.}  \newblock\emph{Ann. Probab.} \textbf{27(3)}, 1536-1555.

\bibitem[Bai and Silverstein(2004)]{BaiSilverstein2004}
Bai, Z. D.  and Silverstein, J.W.  (2004). \newblock{CLT for linear spectral  statistics of large-dimensional sample covariance matrices.}
\newblock\emph{The Annals of Probability}, Vol. \textbf{32}, No. \textbf{1A}, 553-605.

\bibitem[Bai and Silverstein(2010)]{BaiSilverstein2010}
Bai, Z. D.  and Silverstein, J.W. (2010). \newblock{Spectral Analysis of Large Dimensional Random Matrices. }
\newblock\emph{Springer Series in Statistics, Springer-Verlag, New York}, ISSN: 0172-7397.

\bibitem[Bai and Yao(2008)]{BaiYao2008}
Bai, Z. D. and Yao, J. F.  (2008). \newblock{Central limit theorems for eigenvalues in a spiked population model.}
\newblock\emph{Annales de l'Institut Henri Poincar$\acute{e}$ - Probabilit$\acute{e}$s et Statistiques}, Vol. \textbf{44}, No. \textbf{3}, 447-474.

\bibitem[Bai and Yao(2012)]{BaiYao2012}
Bai, Z. D. and Yao, J. F. (2012). \newblock{On sample eigenvalues in a generalized spiked population model.}
\newblock\emph{Journal of Multivariate Analysis}, \textbf{106}, 167-177.

\bibitem[Bai, {\rm et~al.}(2018)]{Baietal2018}
Bai, Z.D.,  Kwok, P.C. and Fujikoshi, Y. (2018). \newblock{Consistency of AIC and BIC in estimating the number of significant components in high-dimensional principal component analysis.} \newblock\emph{The Annals of Statistics}, \textbf{46}, No. 3, 1050-1076.
\bibitem[Bai and Zhou(2008)]{BaiZhou2008}
Bai, Z. D. and Zhou, W. (2008). \newblock{Large sample covariance matrices without independence structures in columns.}
\newblock\emph{Statist. Sinica}, \textbf{18}, 425-442. MR2411613.


\bibitem[Baik, {\rm et al.}(2005)]{Baik2005}
Baik, J., Arous, G. B., P$\acute{\rm e}$ch$\acute{\rm e}$, S. (2005). \newblock{Phase transition of the largest eigenvalue for nonnull complex sample covariance matrices.}
\newblock\emph{The Annals of Probability}, \textbf{33}, 1643-1697.

\bibitem[Baik and Silverstein(2006)]{BaikSilverstein2006}
Baik, J and Silverstein, J. W. (2006). \newblock{Eigenvalues of large sample covariance matrices of spiked population models.}
\newblock\emph{Journal of Multivariate Analysis}, \textbf{97}, 1382-1408.


\bibitem[Ben and P\'{e}ch\'{e}(2005)]{BP2005}
Ben Arous, G. and P$\acute{\rm e}$ch$\acute{\rm e}$, S. (2005). \newblock{Universality of local eigenvalue statistics for some sample covariance matrices.}
\newblock\emph{Comm. Pure Appl. Math.}, \textbf{58}, 1316-1357.

\bibitem[Berthet and Rigollet(2013)]{BerthetRigollet2013}
Berthet, Q. and Rigollet, P. (2013). \newblock{Optimal detection of sparse principal components in high dimension.}
\newblock\emph{The Annals of Statistics}, \textbf{41}, 1780-1815.

\bibitem[Birnbaum, {\rm et al.}(2013)]{Birnbaum2013}
Birnbaum, A., Johnstone, I. M., Nadler, B. and Paul, D. (2013). \newblock{Minimax bounds for sparse PCA with noisy high-dimensional data.}
\newblock\emph{The Annals of Statistics}, Vol. \textbf{41}, No. \textbf{3}, 1055-1084. 

\bibitem[Cai, {\rm et al.}(2019)]{CaiHanPan2017}
Cai,T.T., Han,X. and Pan,G.M. (2019). \newblock{Limiting Laws for Divergent Spiked Eigenvalues and Largest Non-spiked Eigenvalue of Sample Covariance Matrices.},
\newblock\emph{ The Annals of Statistics},  to appear. {http://arxiv.org/abs/1711.00217v2}.

\bibitem[Dyson(1970)]{Dyson1970}
Dyson, F. J. (1970). \newblock{Correlations between eigenvalues of a random matrix.}
\newblock\emph{Comm. Math. Phys.}, \textbf{19}, 235-250.


\bibitem[Erd{\H o}s, {\rm et al.}(2010a)]{Erdos2010a}
Erd{\H o}s, L., P$\acute{e}$ch$\acute{e}$, S., Ram$\acute{i}$rez, J. A., Schlein, B. and Yau, H.-T.  (2010). \newblock{Bulk universality for Wigner matrices.}
\newblock\emph{Comm. Pure Appl. Math.}, \textbf{63}, 895-925.

\bibitem[Erd{\H o}s, {\rm et al.}(2010b)]{Erdos2010b}
Erd{\H o}s, L., Ram$\acute{i}$rez, J. A., Schlein, B. and Yau, H.-T. (2010). \newblock{Universality of sine-kernel for Wigner matrices with a small Gaussian perturbation.}
\newblock\emph{Electron. J. Probab.}, \textbf{15}, 526-603.

\bibitem[Fan, {\rm et al.}(2013)]{FanLiaoM2013}
Fan, J., Liao, Y. and Mincheva, M. (2013). \newblock{Large covariance estimation by thresholding principal orthogonal complements.}
\newblock\emph{Journal of the Royal Statistical Society: Series B}, \textbf{75}, 1-44.

\bibitem[Hoyle and Rattray(2004)]{HoyleRattray2004}
{Hoyle, D.C. and Rattray, M.} (2004). \newblock{Principal-component-analysis eigenvalue spectra from data with symmetry-breaking structure.}
 \newblock\emph{Physics Review E}, \textbf{69}, 026124.

\bibitem[Hu and Bai(2014)]{HuBai2014}
{Hu, J. and Bai, Z.D.} (2014). \newblock{Estimation of directions of arrival  of signals: Asymptotic results.}
 \newblock\emph{Science China Mathematics}, \textbf{57}(11), DOI: 10.1007/s11425-014-4855-6.

\bibitem[Johnstone(2001)]{Johnstone2001}
{Johnstone, I.} (2001). \newblock{On the distribution of the largest eigenvalue in principal components analysis.}
 \newblock\emph{Ann. Statist.}, \textbf{29}, 295-327.  

\bibitem[Johansson(2001)]{Johansson2001}
{Johansson, K.} (2001). \newblock{Universality of the local spacing distribution in certain ensembles of Her- mitian Wigner matrices.}
  \newblock\emph{Comm. Math. Phys.}, \textbf{215}, 683-705.
  
\bibitem[Jung and Marron(2009)]{JungMarron2009}
{Jung, S. and Marron, J. S.} (2009). \newblock{PCA consistency in High dimension, low sample size context. }
 \newblock\emph{The Annals of Statistics}, \textbf{37}, 4104-4130.

\bibitem[Li, {\rm et~al.}(2016)]{LiBaiHu2016}
{Li, H.Q., Bai, Z.D. and Hu, J.} (2016). \newblock{Convergence of empirical spectral distributions of large dimensional quaternion sample covariance matrices.}
 \newblock\emph{Ann Inst Stat Math}, \textbf{68}, 765-785.

\bibitem[Lindeberg(1922)]{Lindeberg1922}
{Lindeberg, J.W.} (1922). \newblock{Eine neue Herleitung des Exponential gesetzes in der Wahrscheinlichkeitsrechnung.}
 \newblock\emph{Math. Z.}, \textbf{15}, 211-225.

\bibitem[Mehta(1967)]{Mehta1967}
{Mehta, M. L.} (1967). \newblock{Random Matrices and the Statistical Theory of Energy Levels.}  \newblock\emph{Academic Press, New York.}

\bibitem[Onatski(2009)]{Onatski2009}
{Onatski, A.} (2009). \newblock{Testing hypotheses about the number of factors in large factor models.}
 \newblock\emph{Econometrica}, \textbf{77}, 1447-1479.

\bibitem[Paul(2007)]{Paul2007}
{Paul, D.} (2007). \newblock{Asymptotics of sample eigenstructure for a large dimensional
spiked covariance model.}
 \newblock\emph{Statistica Sinica}, \textbf{17}, 1617-1642.
 
\bibitem[Passemier and Yao(2012)]{PYao2012}
Passemier, D. and Yao, J.F. (2012). \newblock{On determining the number of spikes in a high-dimensional spiked population model.}
 \newblock\emph{ Random Matrices: Theory and Applications}, \textbf{1} No.1, 1150002.

\bibitem[Soshnikov(1999)]{Soshnikov1999}
{Soshnikov, A.} (1999). \newblock{Universality at the edge of the spectrum in Wigner random matrices.}
 \newblock\emph{Comm. Math. Phys.}, \textbf{207}, 697-733.

\bibitem[Skorokhod(1956)]{Skorokhod1956}
{Skorokhod, A. V.} (1956). \newblock{Limit theorems for stochastic processes.}
 \newblock\emph{Theory Probab Appl.}, \textbf{1}, 261-290.

\bibitem[Shen, {\rm et al.}(2016)]{Shen2013}
{Shen, D., Shen, H., Zhu, H. and Marron, J. S.} (2016). \newblock{The statistics and mathematics of high dimension low sample size asymptotics.} \newblock\emph{Statistica Sinica}, \textbf{26}, 1747-1770.

\bibitem[Tao and Vu(2015)]{TaoVu2015}
Tao, T. and Vu,V. (2015). \newblock{Random matrices: Universality of local eigenvalue statistics.}
 \newblock\emph{The Annals of Probability}, Vol. \textbf{43}, No. 2, 782-874. 

\bibitem[Wang and Fan(2017)]{FanWang2015}
{Wang, W. and Fan, J.} (2017). \newblock{Asymptotics of empirical eigen-structure for high dimensional spiked Covariance.} 
\newblock\emph{The Annals of Statistics,} \textbf{45}, No. 3, 1342-1374.



\bibitem[Wigner(1958)]{Wigner1958}
Wigner, E. P. (1958). \newblock{On the distribution of the roots of certain symmetric matrices.}
 \newblock\emph{Ann. of Math.}, \textbf{67}, 325-327.

\end{thebibliography}
 \end{document}